\documentclass[11pt]{article}

\usepackage{amsthm}
\usepackage{amsmath}
\usepackage{amssymb}
\usepackage{bm} 
\usepackage{bbm}
\usepackage{mathrsfs}
\usepackage{mathtools}
\usepackage{dsfont}
\usepackage{float}
\usepackage{bm}
\usepackage{natbib}
\setcitestyle{authoryear,open={(},close={)}}

\usepackage{graphicx}
\usepackage{adjustbox}
\usepackage{booktabs}
\usepackage{multirow}
\usepackage{caption}
\usepackage[dvipsnames]{xcolor}
\usepackage{comment}
\usepackage{nicefrac}
\usepackage{enumitem}
\usepackage{apptools}
\AtAppendix{\counterwithin{lemma}{section}}
\AtAppendix{\counterwithin{corollary}{section}}
\AtAppendix{\counterwithin{proposition}{section}}
\AtAppendix{\counterwithin{theorem}{section}}
\AtAppendix{\counterwithin{figure}{section}}
\AtAppendix{\counterwithin{table}{section}}
\AtAppendix{\counterwithin{remark}{section}}
\AtAppendix{\counterwithin{condition}{section}}
\AtAppendix{\counterwithin{definition}{section}}

\usepackage{xr} 

\makeatletter
\newcommand*{\addFileDependency}[1]{
  \typeout{(#1)}
  \@addtofilelist{#1}
  \IfFileExists{#1}{}{\typeout{No file #1.}}
}
\makeatother

\theoremstyle{definition}
\newtheorem{theorem}{Theorem} 
 
\newtheorem{lemma}{Lemma}

\newtheorem{corollary}{Corollary} 

\newtheorem{condition}{Condition}
\newtheorem{definition}{Definition}

\usepackage{xspace}
\usepackage[nolist,smaller]{acronym}  
\PassOptionsToPackage{hyphens}{url}\usepackage[colorlinks=true, allcolors=blue]{hyperref}

\newcommand{\acro}[1]{\textsc{#1}\xspace }
\newcommand{\KLD}{\acro{\smaller KLD}}
\newcommand{\TVD}{\acro{\smaller TVD}}

\newcommand{\textBD}{\acro{$\beta$D}}
\newcommand{\textGD}{\acro{$\gamma$D}}

\newcommand{\BD}{\acro{\smaller$D_{B}^{(\beta)}$}}
\newcommand{\GD}{\acro{\smaller$D_{G}^{(\gamma)}$}}

\newcommand{\bayesrule}{Bayes' rule } 
\newcommand{\logscore}{log-score }

\newcommand{\logscorecomma}{log-score, }

\newcommand{\DGP}{\acro{\smaller DGP}} 
\newcommand{\DGPs}{\acro{\smaller DGPs}} 
\newcommand{\DM}{\acro{\smaller DM}}
\newcommand{\DMs}{\acro{\smaller DMs}}

\newcommand{\MSE}{\acro{\smaller MSE}}

\newcommand{\CDF}{\acro{\smaller CDF}}

\newcommand{\pdf}{\acro{\smaller pdf}}

\newcommand{\Student}{Student's-$t$ }
\newcommand{\GLM}{\acro{\smaller GLM}}

\newcommand{\TGFB}{\acro{\smaller TGF-$\beta$}}
\newcommand{\DLD}{\acro{\smaller DLD}}

\newcommand{\R}{\textit{R}}

\DeclareMathOperator*{\argmin}{arg\,min}

\DeclareMathOperator*{\esssup}{ess\,sup}

\def\*#1{\bm{#1}} 

\def\@fnsymbol#1{\ensuremath{\ifcase#1\or *\or \dagger\or \ddagger\or
   \mathsection\or \mathparagraph\or \|\or **\or \dagger\dagger
   \or \ddagger\ddagger \else\@ctrerr\fi}}
\newcommand{\ssymbol}[1]{^{\@fnsymbol{#1}}}

\let\OLDthebibliography\thebibliography
\renewcommand\thebibliography[1]{
  \OLDthebibliography{#1}
  \setlength{\parskip}{0pt}
  \setlength{\itemsep}{3pt plus 0.3ex}
}

\def\*#1{\bm{#1}} 

\usepackage[left=2.3cm, right=2.3cm, top=3cm]{geometry}

\usepackage{authblk}
\title{\bf  {On the Stability of General Bayesian Inference}}
\author[1]{Jack Jewson}
\author[2, 4]{Jim Q. Smith}
\author[3, 4]{Chris Holmes}
\affil[1]{Monash University, Melbourne, Australia}
\affil[2]{University of Warwick, Coventry, CV4 7AL}
\affil[3]{University of Oxford, Oxford, OX1 3LB}
\affil[4]{Alan Turing Institute, London, NW1 2DB}
\affil[ ]{\textit {\textcolor{blue}{jack.jewson@monash.edu, j.q.smith@warwick.ac.uk, chris.holmes@stats.ox.ac.uk}}}
\date{April 2024}

\begin{document}




\def\spacingset#1{\renewcommand{\baselinestretch}%
{#1}\small\normalsize} \spacingset{1}

\setcounter{Maxaffil}{0}
\renewcommand\Affilfont{\itshape\small}

\spacingset{1.42} 

\maketitle
\begin{abstract}
We study the stability of posterior predictive inferences to the specification of the likelihood model and perturbations of the data generating process. 
In modern big data analyses, useful broad structural judgements may be elicited from the decision-maker but a level of interpolation is required to arrive at a likelihood model. 
\color{black}
As a result, an often computationally convenient canonical form is used in place of the decision-maker's true beliefs.
\color{black}
Equally, in practice, observational datasets often contain unforeseen heterogeneities and recording errors 
\color{black}
and therefore do not necessarily correspond to how the process was idealised by the decision-maker.
\color{black}
Acknowledging such imprecisions, a faithful Bayesian analysis should ideally be stable across reasonable equivalence classes of such inputs. 
We are able to guarantee that traditional Bayesian updating provides stability across only a very strict class of likelihood models and data generating processes,
\color{black}
requiring the decision-maker to elicit their beliefs and understand how the data was generated with an unreasonable degree of accuracy. On the other hand, 
\color{black}
a generalised Bayesian alternative using the $\beta$-divergence loss function is shown to be stable across practical and interpretable neighbourhoods, 
\color{black}
providing assurances that posterior inferences are not overly dependent on accidentally introduced spurious specifications or data collection errors.
\color{black}
We illustrate this in linear regression, binary classification, and mixture modelling examples, showing that stable updating does not compromise the ability to learn about the data generating process.
These stability results provide a compelling justification for using generalised Bayes to facilitate inference under simplified canonical models. 
\end{abstract}

\noindent %

\spacingset{1.45}

\section{Introduction}{\label{Sec:Introduction}}

Bayesian inferences are driven by the posterior distribution  
\begin{equation}
\pi(\theta|y)= \frac{\pi(\theta)f(y;\theta)}{\int \pi(\theta)f(y;\theta)d\theta}.\label{Equ:bayesrule}
\end{equation}
which provides the provision to update parameter prior $\pi(\theta)$ using observed data $y = (y_1, \ldots, y_n) \in\mathcal{Y}^n$ assumed to have been generated according to likelihood $f(\cdot;\theta)$. 
The quality of such posterior inference depends on the specification of the prior, likelihood, and collection of the data. 
In controlled experimental environments where time is available to carefully consider such specifications, a posterior calculated in this way might be credible. However, modern applications often involve high-dimensional observational data and are undertaken 
without the supervision of a trained statistician.
In such scenarios, it is natural to question the quality of the specification of $\pi(\theta)$ and $f(\cdot;\theta)$ and the collection of $y$ and therefore wonder to what extent posterior inference through \eqref{Equ:bayesrule} can be trusted.
Much work has previously investigated the stability of \eqref{Equ:bayesrule} to the specification of the prior $\pi(\theta)$, therefore our focus here will be on the likelihood $f(\cdot;\theta)$ and data $y$.
The likelihood model captures the decision maker's (\DM's) beliefs regarding the generation of data $y$. However, accurately formulating expert judgements as probability densities is difficult. Even for a well-trained expert, so doing requires many more probability specifications to be made at a much higher precision than is possible within the time constraints of a typical problem \citep{goldstein1990influence}. 
This is not to say that an elicited model is useless. 
It is certainly possible to reliably elicit important broad structural information from domain experts.
However, the resulting ``\textit{functional}'' model $f(\cdot;\theta)$ generally involves some form of interpolating approximation of the \DM's ``\textit{true}'' beliefs. So doing so is not unreasonable. 
\color{black}
However, a consequence of such expediency is that the \DM does not fully believe all the judgements expressed through their model $f(\cdot;\theta)$. 
\color{black}
A typical example of the above is when applied practitioners deploy computationally convenient canonical models, for which there are software and illustrative examples available, to their domain specific problems. While the broad structure of such models may be suitable across domains, it is the practitioner's familiarity with its form, its software implementation, or the platform on which it was published that often motivates its use for inference, rather than a careful consideration of how it captures beliefs about the new environment. 

Similarly, the data were not necessarily collected exactly how the \DM imagined when specifying their model. There may be unforeseen heterogeneities, outliers, or recording errors. 
Alternatively, the \DM may be deploying someone else's carefully elicited model to an analogous but not necessarily exchangeable scenario. 


\color{black}
Given the inevitable lack of specificity in $f$ and how the data $y$ were generated, a faithful Bayesian analysis should be able to demonstrate that it is not overly dependent on their exact specification. 
\color{black}
Such stability would allow \DMs to continue using familiar models in the knowledge that their arbitrary selection is not driving critical posterior inferences. This paper shows that the requirement for such stability necessitates the consideration of an updating rule different from \eqref{Equ:bayesrule}. 

\color{black}
Consider, for example, a situation where the \DM's \textit{true} beliefs for data $y$ corresponds to a \Student distribution $t_{5}(y; \mu,\sigma^2)$ with 5 degrees of freedom. The top left of Figure \ref{Fig:norm_t_neighbourhood_predictives} shows that the ubiquitous Gaussian likelihood, $\mathcal{N}(y; \mu,\sigma^2)$ captures many of the same judgements. The two likelihoods appear almost indistinguishable for all values of their shared $\mu$ and $\sigma^2$.
Therefore, given finite time and introspection the \DM may reasonably settle on the Gaussian likelihood as a suitable \textit{functioning} approximation of their beliefs. However, the bottom left of Figure \ref{Fig:norm_t_neighbourhood_predictives} shows that when updating according to \eqref{Equ:bayesrule} using the Gaussian model in place of the \Student results in 
very different posterior inferences. Equally, \eqref{Equ:bayesrule} is not stable to perturbations of the data either, as under the Gaussian model a small proportion of outliers moves the posterior inferences away from the uncontaminated part of the data generating process (\DGP). See Section \ref{Sub:GaussianStudent} for full details of this example. 
\color{black}
\begin{figure}[!ht]
\begin{center}
\includegraphics[trim= {0.0cm 0.00cm 0.0cm 0.0cm}, clip,  
width=0.49\columnwidth]{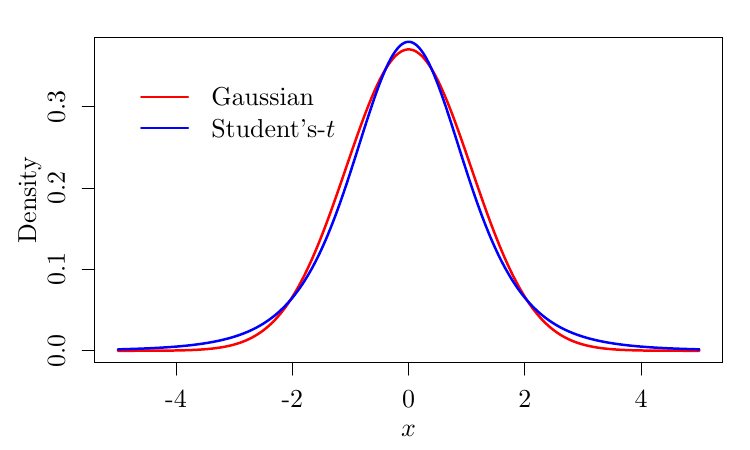}
\includegraphics[trim= {0.0cm 0.00cm 0.0cm 0.0cm}, clip,  
width=0.49\columnwidth]{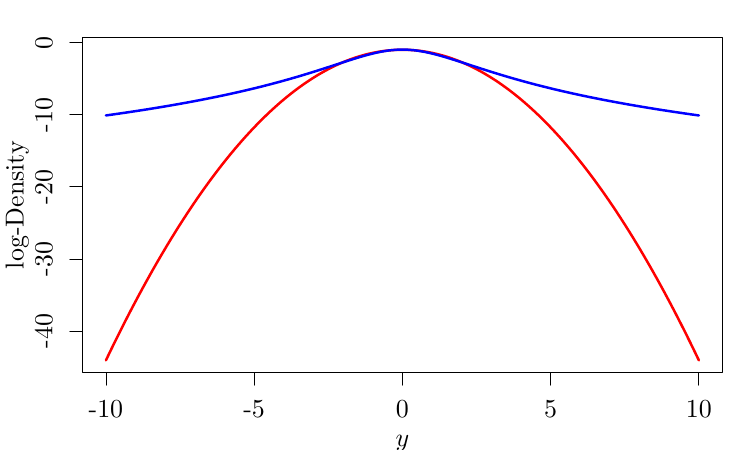}\\
\includegraphics[trim= {0.0cm 0.00cm 0.0cm 0.0cm}, clip,  
width=0.49\columnwidth]{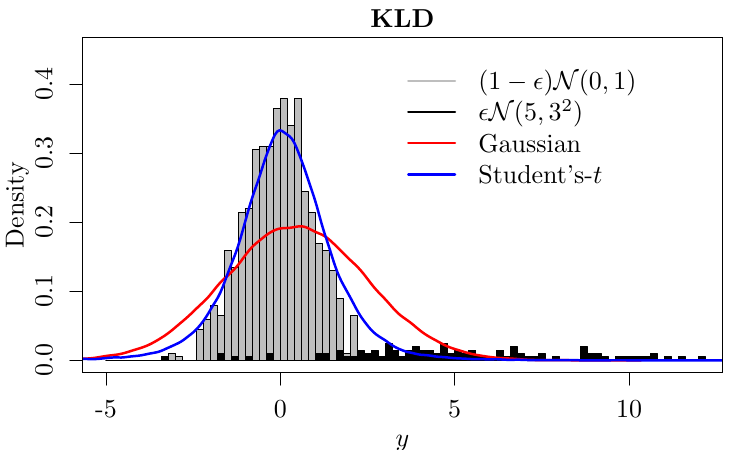}
\includegraphics[trim= {0.0cm 0.00cm 0.0cm 0.0cm}, clip,  
width=0.49\columnwidth]{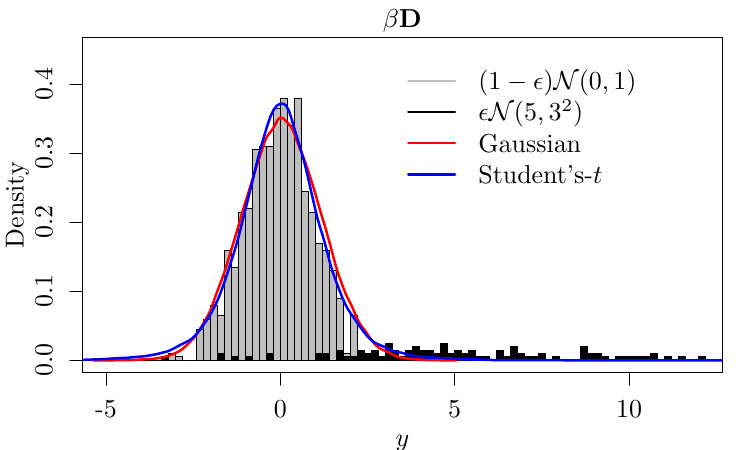}
\caption{\textbf{Top:} Probability density function (\pdf) and log-probability density function of a {\color{black}{Gaussian}} $f_{\sigma_{adj}^2}(y;\theta)=\mathcal{N}\left(y;\mu,\sigma_{adj}^2\sigma^2\right)$ and a {\color{black}{Student's-t}} $h_{\nu}(y;\eta)=t_{\nu}(y;\mu,\sigma^2)$ random variable, with  $\mu=0$, $\sigma^2=1$, $\nu=5$ and $\sigma_{adj}^2=1.16$. 
\textbf{Bottom:} The resulting posterior predictive distributions using {\color{black}{traditional}} and {\color{black}{\textBD-Bayes}} updating with $\beta = 1.22$ on $n=1000$ observations from an $\epsilon$ contamination model $g(y) = 0.9\times\mathcal{N}\left(y;0,1\right) + 0.1 \times \mathcal{N}\left(y;5,3^2\right)$. 
}
\label{Fig:norm_t_neighbourhood_predictives}
\end{center}
\end{figure}

\color{black}
We demonstrate that the instability observed in Figure \ref{Fig:norm_t_neighbourhood_predictives} results from the fact that implicitly \eqref{Equ:bayesrule} learns about the parameter of the model minimising the Kullback-Leibler Divergence (\KLD) between the data generating process (\DGP) and the model, and, as a result, that stability can only be expected when the \DM is sure of the tails of their model specification and the data. The \DM is highly unlikely to be sure of this a priori and therefore, 
\color{black}
under traditional Bayesian updating, it is left up to the \DM to perform some \textit{post hoc} sensitivity analysis 
to examine the impact their chosen model and particular features of the data had on the inference \citep[see][and references within]{box1980sampling,berger1994overview}. 
However, such analyses are usually unsystematic and limited to the investigation of a small number of alternative judgements, models, or data points.


An alternative, motivated by the \textit{M}-open world assumption that the model is misspecified for the \DGP \citep{bernardo2001bayesian}, is to use general Bayes \citep{bissiri2016general} to update beliefs about model parameters minimising a divergence different from the \KLD \citep{jewson2018principles}. A particularly convenient alternative is the $\beta$-divergence (\textBD) which has previously been motivated as providing inference that is robust to outliers \citep{basu1998robust,ghosh2016robust} and desirable from a decision-making point of view \citep{jewson2018principles}. In this paper, we extend the motivation for using \textBD-Bayes further, showing that its posterior predictive inferences are provably stable across an interpretable neighbourhood of likelihood models and \DGP{}s. 
\color{black}
Such results demonstrate that the \textBD-Bayes facilitates the safe use of approximate canonical model specification for modern inference problems.
\color{black}



\color{black}
While inferences should desirably be stable to small perturbations of $f$ and $y$, they should still be sensitive the larger changes in order to extract useful inferences about the \DGP. Importantly, the stability afforded to \textBD-Bayes inference does not compromise this.
\color{black}
The \textBD-Bayes has the appealing property that if the model is correctly specified for the \DGP, then the data generating parameter will be learned. There exists a growing literature that advocates using the \textBD 
for applied analyses 
\citep[e.g.][]{knoblauch2018doubly, knoblauch2022generalized, girardi2020robust, sugasawa2020robust}. This is further demonstrated in our experiments. 
For example, Figure \ref{Fig:norm_t_neighbourhood_predictives} shows that as well as producing similar inference for the Gaussian and \Student likelihood models, the \textBD-Bayes inferences both capture the modal part of the observed data. 
Further, inferences must also not be overly dependent on the selection of hyperparameter, $\beta$, of the \textBD. We discuss methods to select $\beta$ and demonstrate reasonable insensitivity to its selection. 

Results regarding the stability of \eqref{Equ:bayesrule} have largely focused on the parameter prior. \cite{gustafson1995local} proved that the total variation divergence (\TVD) between two posteriors resulting from \textit{functioning} and \textit{true} priors in linear and geometric $\epsilon$-contamination neighbourhoods divergences as $\epsilon\rightarrow 0$ at a rate exponential in the dimension of the parameter space. However, \cite{smith2012isoseparation} showed that the \TVD between two posteriors converges to 0 provided the two priors under consideration are close as measured by the local De Robertis distance. Our first results provide analogies to these for the specification of the likelihood model.
\cite{gilboa1989maxmin,whittle1990risk,hansen2001acknowledging,hansen2001robust,watson2016approximate} consider the stability of optimal decision making and consider minimax decision across neighbourhoods of the posterior. However, they do not consider what perturbations of the inputs of \eqref{Equ:bayesrule} would leave a \DM in such a neighbourhood \textit{a posteriori}. 
Most similar to our work is \cite{miller2018robust}, which considers Bayesian updating conditioning on data arriving within a \KLD ball of the observed data and results concerning `global bias-robustness' to contaminating observations, for example of the kernel-Stein discrepancy posteriors of \cite{matsubara2021robust}. We consider stability to an interpretable neighbourhood of the data which as a special case contains the globally bias-robust contamination.

Bayes linear methods \citep{goldstein1999bayes}, which concern only the sub-collection of probabilities and expectations the \DM considers themselves to be able to specify \citep{goldstein2006subjective}, is an alternative to \eqref{Equ:bayesrule} designed to be stable to interpolating approximations. We prefer, however, to adopt the general Bayesian paradigm in this analysis. Firstly, the general Bayesian paradigm includes traditional Bayesian updating as a special case and produces familiar posterior and predictive distributions.
Secondly, linear Bayes requires the elicitation of expectations and variances of unbounded quantities which are themselves unstable to small perturbations
\citep[see discussion on][]{goldstein1994robustness}. 
Lastly, rather than trying to approximate their beliefs by a single model, the \DM could consider several interpolating approximations and let the data guide any decision the they themselves have not able to make using methods such as penalised likelihood approaches \citep[e.g.][]{akaike1973information, schwarz1978estimating}, 
Bayes' factors \citep{kass1995bayes} or Bayesian model averaging \citep{hoeting1999bayesian}. 
In particular, \cite{williamson2015posterior} propose methods for combining posterior beliefs across an equivalence class of analyses. However, such methods can be computationally burdensome to compute across even a finite class of models \citep[e.g.][]{rossell2021approximate} and can reasonably only consider a handful of models that might fit with the \DM's beliefs, 
\color{black}
all of which contain some level of interpolating approximation.
\color{black}


The rest of the paper is organised as follows: Section \ref{Sec:paradigm} presents our inference paradigm, introducing general Bayesian updating \citep{bissiri2016general}, robustified inference with the \textBD, and defining how we will investigate posterior predictive stability. Section \ref{Sec:StabilityLikelihood} presents our theoretical contributions surrounding the stability of Bayesian analyses to the choice of the likelihood function and Section \ref{Sec:StabilityDGP} presents our results on the stability of inference to perturbations of the \DGP. Proofs of all of our results are deferred to the supplementary material. Section \ref{Sec:SettingBeta} discusses methods to set the $\beta$ hyperparameter and Section \ref{Sec:Experiments} illustrates the stability of the \textBD-Bayes inference in continuous and binary regression examples from biostatistics and a mixture modelling astrophysics example, where stability is shown not to compromise the model's ability to learn about the \DGP. Code to reproduce all of the examples in this paper can be found at \url{https://github.com/jejewson/stabilityGBI}.





\section{A paradigm for inference and stability}{\label{Sec:paradigm}}

\subsection{General Bayesian Inference}

Under the assumption that the model used for inference $f(y; \theta)$ does not exactly capture the \DM's beliefs, 
we find it appealing to adopt the general Bayesian perspective of inference. \cite{bissiri2016general} showed that the posterior update
\begin{align}
\pi^{\ell}(\theta|y)&= \frac{\pi(\theta)\exp\left(-w\sum_{i=1}^n \ell(\theta,y_i)\right)}{\int \pi(\theta)\exp\left(-w\sum_{i=1}^n \ell(\theta,y_i)\right)d\theta}.\label{Equ:GBI}
\end{align} 
%
provides a coherent means to update prior beliefs after observing data $y \sim g(\cdot)$ about parameter $\theta^{\ell}_g:= \argmin_{\theta\in\Theta}\allowbreak \int \ell(\theta,z)g(z)dz$ without requiring that $\theta$ index a model for the data generating density $g(\cdot)$. 

The parameter $w>0$ in \eqref{Equ:GBI} calibrates the loss with the prior to accounts for the fact that unlike the likelihood in \eqref{Equ:bayesrule}, $\exp(-\ell(\theta,y_i))$ is no longer constrained to integrate to 1. \cite{lyddon2018generalized} set $w$ to match the asymptotic information in the general Bayesian posterior to that of a sample from the `loss-likelihood bootstrap', while \cite{giummole2019objective}, building on the work of \cite{ribatet2012bayesian}, directly calibrate the curvature of the posterior to match that of the frequentist loss minimiser. \color{black} A general Bernstein von-Mises Theorem for generalised posterior \eqref{Equ:GBI} was proven in \cite{miller2021asymptotic}.\color{black}

We focus on a subset of loss functions, known as \color{black} proper \color{black} scoring rules \citep{gneiting2007strictly}, that depend upon the \DM's likelihood model, allowing them to use this to encode their beliefs about the \DGP. \color{black} A scoring rule is proper if it is minimised in expectation at the density that generated the data. It therefore provides a means by which the \DM can learn about the \DGP. \color{black}Under the \logscorecomma $\ell(\theta,y)=-\log f(y;\theta)$ \eqref{Equ:GBI} collapses to \eqref{Equ:bayesrule}. The parameter $\theta^{\ell}_g$ associated with the \logscore is the minimiser of the \KLD between the distribution of the sample and the model \citep{berk1966limiting}. We therefore call updating using \eqref{Equ:bayesrule} \KLD-Bayes. 
However, it is well known that minimising the \logscore puts large importance on correctly capturing the tails of the data \citep{bernardo2001bayesian} and can have negative consequences for posterior decision making \citep{jewson2018principles}. This is demonstrated in the bottom left of Figure \ref{Fig:norm_t_neighbourhood_predictives}. 

\subsection{\textBD-Bayes}

An alternative \color{black} proper scoring rule is the $\beta$-divergence 
loss \citep{basu1998robust}, also known as the Tsallis Score  \citep[see e.g.][]{dawid2016minimum})\color{black}
\begin{equation}
\ell^{(\beta)}(y,f(\cdot;\theta))= -\frac{1}{\beta-1}f(y;\theta)^{\beta-1}+\frac{1}{\beta}\int f(z;\theta)^{\beta}dz,\label{Equ:betaDloss}
\end{equation} 
so called as $\argmin_{\theta} \mathbb{E}_{y\sim g}\left[\ell^{(\beta)}(y,f(\cdot;\theta))\right] = \argmin_{\theta} \BD(g || f(\cdot;\theta))$ where $\BD(g || f)$ is the $\beta$-divergence defined in Section A.1. We refer to updating using \eqref{Equ:GBI} and loss \eqref{Equ:betaDloss} as \textBD-Bayes. 
This was first used by \cite{ghosh2016robust} to produce a robustified Bayesian posterior (\textBD-Bayes) and has since been deployed for a variety of examples \citep[e.g.][]{knoblauch2018doubly, knoblauch2022generalized, girardi2020robust, sugasawa2020robust}. 

The implicit robustness to outliers exhibited by the \textBD-Bayes is illustrated in the bottom right of Figure \ref{Fig:norm_t_neighbourhood_predictives}, where, unlike the \KLD-Bayes, the \textBD-Bayes continues to captures the distribution of the majority of observations under outlier contamination. \cite{jewson2018principles} argued that updating in a manner that is automatically robust to outliers, removes the burden on the \DM to specify their beliefs in a way that is robust to the possible existence of occasional outliers. The results of the coming sections provide a formal rationale for adopting this methodology to provide stability to the canonical model choice and departures from the \DGP.


While Bayesian inference has been proposed minimising several alternative divergences including the Hellinger divergence, $\alpha$-divergence, and the \TVD \citep[e.g.][]{hooker2014bayesian,jewson2018principles,knoblauch2020robust} such methods require a non-parametric density estimate, prohibiting their use for high-dimensional problems with continuous data. We restrict our attention to local methods not requiring such an estimate and in particular to the \textBD and \KLD. The $\gamma$-divergence \citep{fujisawa2008robust} has also been shown to produce robust inference without requiring a non-parametric density estimate \citep{hung2018robust,knoblauch2022generalized} and in general behaves very similarly, see Section B.1. 
%
%
%
%
\color{black}
There also exists scoring rules that tailor inference towards improved predictive performance \citep{loaiza2021focused}. However, our focus here is on stably learning about the \DGP in order to facilitate general decision-making without a specific prediction goal in mind.
\color{black}

\subsection{Posterior Predictive Stability }\label{Sub:NotionsStability}

We investigate the stability of general Bayesian posterior predictive distributions
\begin{align}
m^D_{f}(y_{new}|y)&=\int f(y_{new};\theta)\pi^D(\theta|y)d\theta.\label{Equ:PredictiveDensityMetric}
\end{align}
for exchangeable observation $y_{new}\in\mathcal{Y}$ to the specification of the model $f$, and the \DGP $g$. As a result, we focus on the stability of the posterior distribution for observables $y\in\mathcal{Y}$ to perturbations of the prior for observables, $f$, and generating distributions for these observables $g$.

From a decision-making perspective, the posterior predictive is often integrated over to calculate expected utilities, and therefore stable posterior predictive distributions correspond to stable decision-making. 
%
Predictive stability is also a more reasonable requirement than say posterior stability. The parameter posteriors for two distinct models/\DGPs will generally converge in different places \cite[e.g.][]{smith2007local}. However, divergent parameter posteriors do not necessarily imply divergent posterior predictives, as we show. Further, focusing on observables allows us to consider interesting cases of neighbouring models with nested 
parameter spaces (see Section \ref{Sub:MixtureModeling}).

\section{Stability to the specification of the likelihood function}{\label{Sec:StabilityLikelihood}}



\color{black}
In this section, we investigate the stability of inference to the choice likelihood model for a given \DGP.
We consider that the \DM is conducting inference using the functional likelihood model $\left\{f(\cdot;\theta); \theta\in\Theta\subseteq\mathbb{R}^{q_f}\right\}$ in place of their true beliefs $\left\{h(\cdot;\eta); \eta\in\mathcal{A}\subseteq\mathbb{R}^{q_h}\right\}$ for data $y\in\mathcal{Y}$. We assume that $f$ is an approximation of $h$ in the sense that it captures some of the main aspects of $h$ that the \DM has been able to faithfully specify, but interpolates between those in some arbitrary and convenient manner in a way that the \DM does not necessarily believe. In this setting, a faithful posterior belief update should not diverge if $f$ or $h$ is used for inference. That is to say that posterior belief updating should be \textit{stable} to the arbitrary interpolation of belief judgements used for defining a likelihood model. 
\color{black}
\color{black}
In this section 
we investigate sufficient conditions for how $f$ can approximate $h$ that would ensure such stability.
For clarity of argument, we proceed under the assumption that the priors $\pi^D(\theta)$ and $\pi^D(\eta)$ are fixed. All technical conditions are stated in Section A.3.
\color{black}


\subsection{The stability of the \KLD-Bayes}{\label{Sub:StabilityLikelihood_KLD}}

\color{black}
Figure \ref{Fig:norm_t_neighbourhood_predictives} demonstrated that there are examples of models and data where two models that appear very similar to the naked eye (top left), can result in substantially different \KLD-Bayes posterior predictive inference (bottom left). As a result, 
\color{black}
\color{black}
we first examine how $f$ must approximate $h$ in order to guarantee stable traditional Bayesian updating (\KLD-Bayes). 
In particular, Lemma \ref{Thm:StabilityDGPapproxKLD} investigates how stable the posterior predictive approximation of the \DGP $g$ as measured by the \KLD, is to changes in the likelihood model.
\color{black}
\color{black}
Condition A.1, stated in Section A.3, requires that the posterior density on parameter values $\eta$ and $\theta$ such that mapping the parameters $\theta$ of $f$, onto the space of parameters $\mathcal{A}$ for model $h$, using the function $I_f$ leaves $h(\cdot|I_f(\theta))$ \KLD closer to $g$ than $h(\cdot|\eta)$ vanishes exponentially fast and vice verse.
\color{black}

\color{black}
\begin{lemma}[The stability in the posterior predictive approximation of the \DGP of \KLD-Bayes inference]
For any two two likelihood models $\left\lbrace f(\cdot;\theta): \theta\in\Theta\subseteq\mathbb{R}^{q_f}\right\rbrace$ and $\left\lbrace h(\cdot; \eta):\eta\in\mathcal{A}\right\rbrace$, and $y$, $\pi^{\KLD}(\theta)$ and $\pi^{\KLD}(\eta)$ satisfying Condition A.1 for $D = $ \KLD, we have that
\begin{align}
|\KLD(g||m^{\KLD}_{f}(\cdot|y))-
\KLD(g||m^{\KLD}_{h}(\cdot|y))|&\leq C^{\KLD}(f,h,y) + \frac{1}{c} + T(f,h,y), \nonumber
\end{align}
where $c := \min\{c_1, c_2\}$, $I_f:\Theta\mapsto \mathcal{A}$ and $I_h:\mathcal{A}\mapsto\Theta$ are defined in Condition A.1 and
\begin{align}
C^{\KLD}(f,h,y):&= \max \left\lbrace\int\KLD(g||f(\cdot;\theta))\pi^{\KLD}(\theta|y)d\theta-\KLD(g||m^{\KLD}_{f}(\cdot|y)),\right.\nonumber\\
&\qquad\left.\int\KLD(g||h(\cdot;\eta))\pi^{\KLD}(\eta|y)d\eta-\KLD(g||m^{\KLD}_{h}(\cdot|y)) \right\rbrace.\nonumber\\
T(f,h,y):&= \max \left\lbrace\int \int g(\cdot) \log \frac{f(\cdot;\theta)}{h(\cdot;I_f(\theta))}d\mu\pi^{\KLD}(\theta|y)d\theta,\right.\nonumber\\
&\qquad\left. \int\int g(\cdot) \log \frac{h(\cdot;\eta)}{f(\cdot;I_h(\eta))}d\mu\pi^{\KLD}(\eta|y)d\eta \right\rbrace.\label{Equ:StabilityTermKLD}
\end{align}
\label{Thm:StabilityDGPapproxKLD}
\end{lemma}
\color{black}

\color{black}
As a result, sufficient conditions for \KLD-Bayes to provide a stable approximation of the \DGP $g$ when using model $f$ in place of model $h$ are that terms $C^{\KLD}(f,h,y)$, $\frac{1}{c}$, and $T(f,h,y)$ are small. 
The term $C^{\KLD}(f,h,y)$ is the maximal difference between the \KLD of the model from $g$ in expectation under the posterior and the \KLD of the posterior predictive from $g$ under either model $f$ or $h$. This is driven by how concentrated the \KLD-Bayes posteriors are. Similarly, the term $c$ is the minimal rate associated with Condition A.1. This is driven by how quickly the posteriors concentrate around their \KLD minimising parameters. 
\color{black}
\color{black}
\color{black}
We use Lemma \ref{Thm:StabilityDGPapproxKLD} to examine what $f$ must correctly capture about $h$ in order that inference with both achieves similar approximations of the \DGP. We therefore investigate some properties of $T(f,h,y)$.
\color{black}
%
Without loss of generality assume that the second term in \eqref{Equ:StabilityTermKLD} is the largest. 
Then, $T(f,h,y)$ being small requires that
\begin{equation}
\left|\log(h(\cdot;\eta))-\log(f(\cdot;I_h(\eta)))\right|\label{Equ:log_h_minus_log_f}
\end{equation}
is small in regions where $g(\cdot)$ and $\pi^{\KLD}(\eta|y)$ have density. Without knowledge of $g$, this requires that \eqref{Equ:log_h_minus_log_f} be small everywhere for all $\eta$. 

\color{black}
Lemma \ref{Thm:StabilityDGPapproxKLD} establishes that if a \DM can ensure that \eqref{Equ:log_h_minus_log_f} is small everywhere then they can use the approximate model $f$ in place of their true beliefs $h$ and be safe in the knowledge that their \KLD-Bayes posterior inferences cannot be driven by some arbitrary part of the approximate model. 
\color{black}
However, this requires the \DM to be confident in the accuracy of the probability statements made by $f$ on the log scale. Logarithms act to inflate the magnitude of small numbers thus ensuring that $\left|\log(h(\cdot;\eta))-\log(f(\cdot; I_h(\eta)))\right|$ is small requires that $f$ and $h$ are increasingly similar as their values decrease. This requires the \DM to be more and more confident of the accuracy of the probability statements made by $f$ further and further into the tails, something that is known to already be very difficult for low dimensional problems \citep{winkler1968evaluation,o2006uncertain}, and becomes increasingly difficult as the dimension of the observation space increases. 
\color{black}
Tail probabilities definitionally correspond to surprising events and are thus harder to specify accurately. We therefore conclude that this is not a reasonable requirement to ask of any \DM. 
\color{black}

\color{black}

While Lemma \ref{Thm:StabilityDGPapproxKLD} 
does not indicate the tightness of this bound, the example presented in Figure \ref{Fig:norm_t_neighbourhood_predictives} demonstrates the importance of $T(f,h,y)$ being small for stable inference. 
Figure \ref{Fig:norm_t_neighbourhood_predictives} (top right)
shows that while the Gaussian and \Student may appear similar when viewed on the natural scale the difference in their log probabilities is large in their tails. Figure \ref{Fig:norm_t_neighbourhood_predictives}, therefore provides an example of two likelihood models where \eqref{Equ:log_h_minus_log_f} is not small everywhere and a \DGP where the two models result in substantially different posterior beliefs (bottom left).
\color{black}






\subsection{An interpretable neighbourhood of likelihood models}

\color{black}
Motivated by the results of Section \ref{Sub:StabilityLikelihood_KLD}, we consider in what manner a \DM might reasonably be able to accurately approximate their beliefs.  
Firstly, the total variation metric is defined as
\begin{equation}
\TVD(f(\cdot;\theta),h(\cdot;\eta)) := \sup_{Y\in\mathcal{Y}}\left|f(Y;\theta)-h(Y;\eta)\right| = \frac{1}{2}\int \left|f(y;\theta)-h(y;\eta)\right|dy. \label{Equ:TVD}
\end{equation}
Then, functioning likelihood models $f$ for data $y\in\mathcal{Y}$ is considered `$\epsilon$-close' to true belief distribution $h$ if  Definition \ref{Def:LikelihoodNeighbourhood} is satisfied.
\color{black}

\begin{definition}[\TVD neighbourhood of likelihood models]
Likelihood models  
$f(\cdot;\theta)$ and $h(\cdot;\eta)$
for observable $y\in\mathcal{Y}$ are in the neighbourhood $\mathcal{N}_{\epsilon}^{\TVD}$ of size $\epsilon$ if
    \begin{align}
        &\forall \theta \in \Theta, \exists \eta \in \mathcal{A} \textrm{ s.t. } \TVD(f(\cdot;\theta), h(\cdot; \eta)) \leq \epsilon \quad\textrm{and}\quad\nonumber\\
         &\forall \eta \in \mathcal{A}, \exists \theta \in \Theta \quad\textrm{s.t.} \quad \TVD(f(\cdot;\theta), h(\cdot; \eta)) \leq \epsilon \nonumber.
    \end{align}
\label{Def:LikelihoodNeighbourhood}
\end{definition}

Being in the neighbourhood $\mathcal{N}_{\epsilon}^{\TVD}$ entails the existence of functions $I_f: \Theta \mapsto \mathcal{A}$ and $I_h: \mathcal{A}\mapsto \Theta $ such that for all $\theta$, $\TVD(f(\cdot; \theta), h(\cdot; I_f(\theta))$ is small and for all $\eta$, $\TVD(h(\cdot; \eta), f(\cdot; I_h(\eta))$ is also small. 
\color{black}
This means that there must exist mappings between the two parameter spaces such that for any parameter $\theta$ of $f$, mapping $\theta$ to $\eta$ via $I_f$ leaves $h(\cdot;\eta)$ \TVD-close to $f(\cdot, \theta)$. 
\color{black}
Note that the symmetry of Definition \ref{Def:LikelihoodNeighbourhood} allows $\Theta$ and $\mathcal{A}$ to have different dimensions. 

\color{black}
The motivation for using the \TVD in Definition \ref{Def:LikelihoodNeighbourhood} is three-fold. Firstly, and foremost,
the \TVD is interpretable. 
For two likelihoods to be close in terms of \TVD requires that the greatest difference in any of the probability statements made by the two likelihoods be small on the natural scale - where elicitation of probabilities and sample distributions usually takes place - and not the log scale.
%
In a practical sense, two densities that appear `close' to the naked eye will be close according to \TVD, while this heuristic will not be sufficient for close log probability.
As a result, we believe that specifying a model that is \TVD close to their exact beliefs is a feasible and reasonable requirement of a \DM.

Further, 
the \TVD is natural in the context of Bayesian decision-making. Two densities that are close in terms of \TVD will produce similar estimates of bounded expected utility, and thus lead to similar decisions. This has previously been discussed by \cite{smith2010bayesian} and \cite{jewson2018principles}. Therefore, a model that is \TVD close to the \DM’s true beliefs will perform similarly from a decision-making perspective a priori. 
Lastly, the \TVD neighbourhood contains $\epsilon$-contamination models which are popular models for investigating prior stability \citep{gustafson1995local} and outliers \citep[e.g.][]{aitkin1980mixture}.

While Pinsker's inequality \citep{pinkser1964information} shows that \eqref{Equ:log_h_minus_log_f} being small everywhere is a sufficient condition for Definition \ref{Def:LikelihoodNeighbourhood}, it is not necessary to have close log probabilities to have close absolute probabilities. 
This is for example evidenced by Figure \ref{Fig:norm_t_neighbourhood_predictives} where the \TVD between the two densities is less than 0.05.
Therefore, Definition \ref{Def:LikelihoodNeighbourhood} imposes a less strict requirement on the \DM. 
Section A.6 demonstrates this by presenting an example where a shrinking \TVD neighbourhood corresponds to an expanding \KLD neighbourhood.

\color{black}

\color{black}

\subsection{The stability of the \textBD-Bayes}

\color{black}
In this section, we demonstrate that Definition \ref{Def:LikelihoodNeighbourhood} is a sufficient condition for stable updating under \textBD-Bayes. 
In addition to Condition A.1, the results in this section require Condition A.2, stated in Section A.3. This requires the boundedness over the space of data $y$ of the essential supremum of \DGP $g(\cdot)$ and models $f(\cdot;\theta)$ and $h(\cdot;\eta)$ for all values of their parameters $\theta$ and $\eta$. We need this condition to bound the \textBD and relate it to the \TVD. 
In discrete models, this bound is always $1$ and in continuous models such as Gaussian or \Student likelihood a bound can be achieved by lower bounding their scale.
Theorem \ref{Thm:StabilityDGPapproxBeta} provides an analogous result to Lemma \ref{Thm:StabilityDGPapproxKLD} but shows that Definition \ref{Def:LikelihoodNeighbourhood} is sufficient for posterior predictive stability.
\color{black}

\begin{theorem}[The stability in the posterior predictive approximation of two models to the \DGP of \textBD-Bayes inference]
Assume $1< \beta\leq 2$ and that the two likelihood models $\left\lbrace f(\cdot;\theta): \theta\in\Theta\right\rbrace$ and $\left\lbrace h(\cdot; \eta):\eta\in\mathcal{A}\right\rbrace$ are such that $f,h\in\mathcal{N}_{\epsilon}^{\TVD}$ for $\epsilon>0$. Then provided Condition A.1  for $D = $ \BD is satisfied for $y$, $\pi^{(\beta)}(\theta)$ and $\pi^{(\beta)}(\eta)$ and there exists $M<\infty$ such that Condition A.2 holds, then
\begin{equation}
    |\BD(g||m^{(\beta)}_{f}(\cdot|y))-
    \BD(g||m^{(\beta)}_{h}(\cdot|y))|\leq \frac{M^{\beta - 1}(3\beta - 2)}{\beta(\beta - 1)}\epsilon+ \frac{1}{c} + C^{(\beta)}(f,h,y),\nonumber
\end{equation}
where $c = \min\{c_1, c_2\}$ are defined in Condition A.1 and
\begin{align}
C^{(\beta)}(f,h,y):&= \max \left\lbrace\int\BD(g||f(\cdot;\theta))\pi^{(\beta)}(\theta|y)d\theta-\BD(g||m^{(\beta)}_{f}(\cdot|y)),\right.\nonumber\\
&\qquad\left.\int\BD(g||h(\cdot;\eta))\pi^{(\beta)}(\eta|y)d\eta-\BD(g||m^{(\beta)}_{h}(\cdot|y)) \right\rbrace.\nonumber
\end{align}
\label{Thm:StabilityDGPapproxBeta}
\end{theorem}

Additionally, Theorem \ref{Thm:StabilityPosteriorPredictivebetaDiv} proves that as well as providing similar approximations to the \DGP, the \textBD between the \textBD-Bayes posterior predictive distributions themselves can also be bounded.


\begin{theorem}[Stability of the posterior predictive distributions of two models under the \textBD-Bayes inference]
Assume $1< \beta\leq 2$ and that the two likelihood models $\left\lbrace f(\cdot;\theta): \theta\in\Theta\right\rbrace$ and $\left\lbrace h(\cdot; \eta):\eta\in\mathcal{A}\right\rbrace$ are such that $f,h\in\mathcal{N}_{\epsilon}^{\TVD}$ for $\epsilon>0$. 
Then provided Condition A.1  for $D = $ \BD is satisfied for $y$, $\pi^{(\beta)}(\theta)$ and $\pi^{(\beta)}(\eta)$ and there exists $M<\infty$ such that Condition A.2 holds, then
\begin{align}
\BD(m^{(\beta)}_{f}(\cdot|y)||m^{(\beta)}_{h}(\cdot|y)) &\leq \frac{M^{\beta - 1}(3\beta - 2)}{\beta(\beta - 1)}\epsilon + \frac{1}{c_1} + 2\frac{M^{\beta - 1}}{\beta-1}\int \TVD(g, f(\cdot;\theta))\pi^{(\beta)}(\theta|y)d\theta\nonumber\\
\BD(m^{(\beta)}_{h}(\cdot|y)||m^{(\beta)}_{f}(\cdot|y)) &\leq \frac{M^{\beta - 1}(3\beta - 2)}{\beta(\beta - 1)}\epsilon + \frac{1}{c_2} + 2\frac{M^{\beta - 1}}{\beta-1}\int \TVD(g, h(\cdot;\eta))\pi^{(\beta)}(\eta|y)d\eta,\nonumber
\end{align}
where $c_1$ and $c_2$ are defined in Condition A.1
\label{Thm:StabilityPosteriorPredictivebetaDiv}.
\end{theorem}

\color{black}

Theorem \ref{Thm:StabilityDGPapproxBeta} is directly analogous to Lemma \ref{Thm:StabilityDGPapproxKLD} with terms $C^{(\beta)}(f,h,y)$ and $c$ having the same interpretation. Corollaries A.1 and A.2 invoke the generalised Bayesian Bernstein von-Mises theorem \citep{miller2021asymptotic} applied to the \textBD (Theorem A.1) to show that under very general regularity conditions $c\rightarrow\infty$ and $C^{(\beta)}(f,h,y) \rightarrow 0$ as $n\rightarrow\infty$. 
Therefore, Theorem \ref{Thm:StabilityDGPapproxBeta} establishes that Definition \ref{Def:LikelihoodNeighbourhood} is sufficient for the \textBD-Bayes posterior predictive distributions under two models to produce similar approximations of \DGP $g$. This allows a \DM to proceed using a model that well approximates their beliefs, as measured by the \TVD, and know that the imprecision of their beliefs specification cannot lead to substantially different posterior predictive beliefs.

Note that Theorem \ref{Thm:StabilityDGPapproxBeta} and Lemma \ref{Thm:StabilityDGPapproxKLD} are not directly comparable results. Lemma \ref{Thm:StabilityDGPapproxKLD} upper bounds the difference in the \KLD approximation of the \DGP whereas Theorem \ref{Thm:StabilityDGPapproxBeta} bounds the difference in \textBD approximation of the \DGP. 
The two bounds themselves are therefore not directly comparable only the conditions leading to these bounds.
The sufficient conditions for \KLD-Bayes to be stable are impractical to satisfy, while the \textBD-Bayes is provably stable under reasonable conditions that might in practice be plausible to believe. We also do not expect Theorem \ref{Thm:StabilityDGPapproxBeta} (and \ref{Thm:StabilityPosteriorPredictivebetaDiv}) to be tight, however they are not vacuous. Lemma A.8 demonstrates that under Condition A.2, the \textBD between any two densities is bounded by $\frac{{M}^{\beta-1}}{\beta - 1}$ and therefore provided $\frac{(3\beta - 2)}{\beta}\epsilon < 1$, our results provide a tighter upper bound than a trivial bound on the divergence. 

Theorem \ref{Thm:StabilityPosteriorPredictivebetaDiv} demonstrates that \textBD-Bayes updating not only provides a stable approximation of the \DGP (as in Theorem \ref{Thm:StabilityDGPapproxBeta}) but also that the \textBD between posterior predictives under two \TVD close models can be bounded above. This result is slightly weaker than Theorem \ref{Thm:StabilityDGPapproxBeta} because it requires the \TVD between the model and the \DGP to be small in expectation under the posterior. 
A strength of Theorem \ref{Thm:StabilityDGPapproxBeta} is that it holds independent of how well either of the models approximates the \DGP. 
Lastly, note that the choice of $\beta$ away from 1 - the case corresponding to the \KLD - is necessary for Theorems \ref{Thm:StabilityDGPapproxBeta} and \ref{Thm:StabilityPosteriorPredictivebetaDiv} to be practically useful as the bounds in all tend to infinity as $\beta\rightarrow 1$.

\color{black}

\section{Stability to the data generating process}{\label{Sec:StabilityDGP}}

\color{black}

In this section, we investigate the stability of inference to perturbations of the \DGP, the mechanism with which the data was generated. Consider that the \DM is conducting inference using likelihood model $\left\{f(\cdot;\theta); \theta\in\Theta\subseteq\mathbb{R}^{q_f}\right\}$ that was faithfully elicited to capture beliefs about idealised \DGP $g_1(\cdot)$. Whether this corresponds to their true beliefs or an approximation is not relevant to the argument below. 
Now suppose that, for unforeseen reasons, the
data were actually generated according to $g_2(\cdot)$, a perturbation of $g_1(\cdot)$. 
A useful property to demonstrate would be that if, in some appropriate sense, such perturbations were small inferences from what was actually observed $g_2$ would be similar to those where $g_1$ had been assumed.
Therefore, we investigate sufficient conditions for how $g_2(\cdot)$ can differ from $g_1(\cdot)$ and this stability be achieved.
Throughout we consider data sets $y_1:=(y_1,\ldots, y_{n_1})\sim g_1$. and $y_2:=(y_1,\ldots, y_{n_2})\sim g_2$. Although not necessary to our argument we assume for simplicity that $n_1 = n_2$. All regularity conditions for these results to hold and their proofs are given in Section A.3. 

\color{black}



\subsection{The stability of the \KLD-Bayes}{\label{Sec:DGPstabilityKLD}}

\color{black}
Figure \ref{Fig:norm_t_neighbourhood_predictives} considered a case where the data were generated from $g_2(y) = 0.9\times\mathcal{N}\left(y;0,1\right) + 0.1 \times \mathcal{N}\left(y;5,3^2\right)$ while the Gaussian model was an accurate representation of $g_1(y) = \mathcal{N}\left(y;0,1\right)$. Although the \DGP was the same for 90\% of the observations, \KLD-Bayes posterior inference under $g_2$ differs considerable from what one obtains when fitting $f$ to $g_1$ - see Figure B.2. 
Figure \ref{Fig:norm_t_neighbourhood_predictives}, therefore, demonstrates that there are examples of models and data where two largely similar \DGPs result in substantially different posterior predictive inferences from the same model. 
As a result, 
\color{black}
we first investigate how $g_2$ can differ from $g_1$ in order to guarantee stable traditional Bayesian updating (\KLD-Bayes) for $f$. Lemma \ref{Thm:StabilityDGPapproxKLD2} investigates how stable the posterior predictive approximation to the \DGP as measured by the \KLD is to changes in the \DGP. 
\color{black}
Condition A.3, stated in Section A.3, is analogous to Condition A.1. 
This requires that the posterior density on regions of $\theta_1 | g_1$ and $\theta_2 | g_2$  that leaves $f(\cdot; \theta_1)$ \KLD closer to $g_2$ than $f(\cdot; \theta_2)$ or $f(\cdot; \theta_2)$ \KLD closer to $g_1$ than $f(\cdot; \theta_1)$  vanishes exponentially fast.
\color{black}


 \color{black}
\begin{lemma}[The stability in the posterior predictive approximation of two \DGPs under the same model for \KLD-Bayes inference]
For likelihood model $\left\lbrace f(\cdot;\theta): \theta\in\Theta\right\rbrace$ and data sets $y_1:=(y_1,\ldots, y_{n_1}) \sim g_1$ and $y_2:=(y_1,\ldots, y_{n_2}) \sim g_2$ for $n_1, n_2 > 0$, if Condition A.3 holds for $D = $ \KLD, $y_1$, $y_2$ and $\pi^{\KLD}(\theta)$, then
\begin{align}
|\KLD(g_1||m^{\KLD}_{f}(\cdot|y_1))-
\KLD(g_2||m^{\KLD}_{f}(\cdot|y_2))|&\leq C^{\KLD}(f,y_1, y_2) + \frac{1}{c} + T_1(g_1, g_2) + T_2(f,y_1, y_2), \nonumber
\end{align}
where $c:= \min\{c_{\mathcal{S}^{(1)}}, c_{\mathcal{S}^{(2)}}\}$ are defined in Condition A.3 and
\begin{align}
T_1(g_1, g_2):&= \left| \int g_2\log g_2 - g_1 \log g_1d\mu\right|\nonumber\\
T_2(f,y_1, y_2):&= \max \left\lbrace\int \int (g_1 - g_2)\log f(\cdot;\theta_1)d\mu\pi^{\KLD}(\theta_1|y_1)d\theta_1,\right.\nonumber\\
&\qquad\left.\int \int (g_2 - g_1)\log f(\cdot;\theta_2)d\mu\pi^{\KLD}(\theta_2|y_2)d\theta_2 \right\rbrace\nonumber \\
C^{\KLD}(f,y_1, y_2):&= \max \left\lbrace\int\KLD(g_1||f(\cdot;\theta_1))\pi^{\KLD}(\theta_1|y_1)d\theta_1-\KLD(g_1||m^{\KLD}_{f}(\cdot|y_1)),\right.\nonumber\\
&\qquad\left.\int\KLD(g_2||f(\cdot;\theta_2))\pi^{\KLD}(\theta_2|y_2)d\theta_2-\KLD(g_2||m^{\KLD}_{f}(\cdot|y_2)) \right\rbrace \nonumber.
\end{align}
\label{Thm:StabilityDGPapproxKLD2}
\end{lemma}
\color{black}

\color{black}
So \KLD-Bayes can certainly be ensured 
to provide stable approximation to the \DGP when using model $f$ to update beliefs on data from $g_2$ rather than $g_1$ are that terms $C^{\KLD}(f,y_1, y_2)$, $\frac{1}{c}$, $T_1(g_1, g_2)$ and $T_2(f,y_1, y_2)$ are small. 
The term $C^{\KLD}(f,y_1, y_2)$ is the difference between the \KLD of $f$ from $g_j$ in expectation under the posterior and the \KLD of the posterior predictive of $f$ from $g_j$ maximsed over $j = 1, 2$. This is driven by how concentrated the posteriors are. Similarly, the term $c$ is the minimal rate associated with Condition A.3 and is driven by how quickly the posteriors concentrate around their \KLD minimising parameters. 
We are interested in how $g_2$ must be close to $g_1$ for this bound to be small and therefore we focus on terms $T_1(g_1, g_2)$ and $T_2(f,y_1, y_2)$.
Small $T_1(g_1, g_2)$ requires $g_1$ and $g_2$ to have similar Shannon entropy, a measure of the inherent randomness in the data, which seems a reasonable condition. However, as $f(y; \theta)\rightarrow 0$, $|\log f(y;\theta)| \rightarrow \infty$ therefore small $T_2(f, y_1, y_2)$ requires that $|g_1(y) - g_2(y)|$ gets smaller as $f(y;\theta)$ gets smaller for $\theta \sim \pi^{\KLD}(\theta|y)$. 
That is to say that, $T_2(f,y_1, y_2)$ being small requires $g_1$ and $g_2$ to be increasingly close in their tails.

Such a requirement greatly reduces the generalisability of statistical modelling. The tails of the \DGP correspond to rare observations and therefore the \KLD-Bayes only generalises across \DGPs with similar rare observations. Encountering such situations is not only unlikely in practice, but difficult for any \DM to consider following our discussion in Section \ref{Sec:StabilityLikelihood}. This, for example, prohibits outlier $\epsilon$-contamination models where the \DGP for $(1-\epsilon)$\% of the data is the same across $g_1$ and $g_2$, but $g_2$ is contaminated with $\epsilon$\% of outliers, as seen in Figure \ref{Fig:norm_t_neighbourhood_predictives} and B.2.
%
Such an example also provides an indication that although Lemma \ref{Thm:StabilityDGPapproxKLD2} is only an upper bound that is not necessarily tight, the absence of small $T_2(f, y_1, y_2)$ results in substantially different \KLD-Bayes posterior predictive inferences.

%


\color{black}

\subsection{A plausible neighbourhood of data generating process perturbations}

\color{black}

The results of Section \ref{Sec:DGPstabilityKLD} motivated us to consider what perturbations of the \DGP should we reasonably expect our posterior inferences to be stable to. Data generating processes $g_1$ and $g_2$ for data $y \in \mathcal{Y}$ are considered `$\epsilon$-close' if Definition \ref{Def:DGPNeighbourhood} is satisfied.

\color{black}

\begin{definition}[\TVD Neighbourhood of data generating processes]
Data generating processes $g_1$ and $g_2$ for observable $y\in\mathcal{Y}$ are in the neighbourhood $\mathcal{G}_{\epsilon}^{\TVD}$ of size $\epsilon$ if $\TVD(g_1, g_2) \leq \epsilon$
\label{Def:DGPNeighbourhood}
\end{definition}

\color{black}

Following \eqref{Equ:TVD}, $g_2$ is a small perturbation of $g_1$ according to Definition \ref{Def:DGPNeighbourhood} if the probability statements made by either differ by a maximum of $\epsilon$. This gives equal weight to modal or tail discrepancies rather than overly focusing on having the same tails. As a result, the data generating distribution for two populations will be close if distributions of the majority of the observations are close, rather than the distributions for a few of the observations. 

Further, Definition \ref{Def:DGPNeighbourhood} contains $\epsilon$-contamination neighbourhoods as considered by \cite{matsubara2021robust} and demands that the data sets were generated under mechanisms that were absolutely close on the natural scale, rather than the log-score considered in the \KLD neighbourhoods on \cite{miller2018robust}. 

\color{black}

\subsection{The stability of the \textBD}

\color{black}
We now demonstrate that Definition \ref{Def:DGPNeighbourhood} is a sufficient condition on $g_1$ and $g_2$ to bound the consequences of generalising \textBD-Bayes inference for $f$ from $g_1$ to $g_2$. In addition to Condition A.3, the results in this section require Condition A.4 which is analogous to Condition A.2 and requires the bounding of the essential supremum over the space $y$ of \DGPs $g_1(\cdot)$ and $g_2(\cdot)$ and model $f(y;\theta)$ for all $\theta$. Theorem \ref{Thm:StabilityDGPapproxBeta2} is an analogous result to Lemma \ref{Thm:StabilityDGPapproxKLD2} showing that Definition \ref{Def:DGPNeighbourhood} is sufficient for posterior predictive stability
\color{black}


\begin{theorem}[The stability in the posterior predictive approximation of two \DGPs under the same model of \textBD-Bayes inference]
Assume $1< \beta\leq 2$, likelihood model $\left\lbrace f(\cdot;\theta): \theta\in\Theta\right\rbrace$ and that the two data sets $y_1:=(y_1,\ldots, y_{n_1}) \sim g_1$ and $y_2:=(y_1,\ldots, y_{n_2}) \sim g_2$ for $n_1, n_2 > 0$ are such that 
$\{g_1, g_2\}\in \mathcal{G}_{\epsilon}^{\TVD}$. 
Then provided that Condition A.3 holds for $D = $ \BD, $y_1$, $y_2$ and $\pi^{(\beta)}(\theta)$ and there exists $M<\infty$ such Condition A.4 holds, then
\begin{equation}
|\BD(g_1||m^{(\beta)}_{f}(\cdot|y_1))-
\BD(g_2||m^{(\beta)}_{f}(\cdot|y_2))|\leq \frac{M^{\beta - 1}(\beta + 2)}{\beta(\beta - 1)}\epsilon + \frac{1}{c} + C^{(\beta)}(f,y_1, y_2),\nonumber
\end{equation}
where $c:= \min\{c_{\mathcal{S}^{(1)}}, c_{\mathcal{S}^{(2)}}\}$ are defined in Condition A.4 and
\begin{align}
C^{(\beta)}(f,y_1, y_2):&= \max \left\lbrace\int\BD(g_1||f(\cdot;\theta_1))\pi^{(\beta)}(\theta_1|y_1)d\theta_1-\BD(g_1||m^{(\beta)}_{f}(\cdot|y_1)),\right.\nonumber\\
&\qquad\left.\int\BD(g_2||f(\cdot;\theta_2))\pi^{(\beta)}(\theta_2|y_2)d\theta_2-\BD(g_2||m^{(\beta)}_{f}(\cdot|y_2)) \right\rbrace \nonumber
\end{align}
\label{Thm:StabilityDGPapproxBeta2}
\end{theorem}

Additionally, Theorem \ref{Thm:StabilityPosteriorPredictivebetaDiv2} proves that as well as providing similar approximations to the DGPs, the \textBD between the \textBD-Bayes posterior predictive distributions themselves can also be bounded.

\begin{theorem}[The stability of the posterior predictive distribution under two \DGP{}s of the \textBD-Bayes inference]
Assume $1< \beta\leq 2$, likelihood model $\left\lbrace f(\cdot;\theta): \theta\in\Theta\right\rbrace$ and that the two data sets $y_1:=(y_1,\ldots, y_{n_1}) \sim g_1$ and $y_2:=(y_1,\ldots, y_{n_2}) \sim g_2$ for $n_1, n_2 > 0$ are such that  
$\{g_1, g_2\}\in \mathcal{G}_{\epsilon}^{\TVD}$. 
Then provided there exists $M<\infty$ such that Condition A.3 hold, Condition A.4 holds for $D = $ \BD, $y_1$, $y_2$ and $\pi^{(\beta)}(\theta)$, then
\begin{align}
\BD(m^{(\beta)}_{f}(\cdot|y_1)||m^{(\beta)}_{f}(\cdot|y_2)))\leq& 2\frac{M^{\beta - 1}}{\beta-1}\epsilon + \frac{1}{c_{\mathcal{S}^{(1)}}} + 2\frac{M^{\beta - 1}}{\beta-1}\int \TVD(g_1, f(\cdot;\theta_1))\pi^{(\beta)}(\theta_1|y_1)d\theta_1.\nonumber\\
\BD(m^{(\beta)}_{f}(\cdot|y_2)||m^{(\beta)}_{f}(\cdot|y_1))) \leq& 2\frac{M^{\beta - 1}}{\beta-1}\epsilon + \frac{1}{c_{\mathcal{S}^{(2)}}} + 2\frac{M^{\beta - 1}}{\beta-1}\int \TVD(g_2, f(\cdot;\theta_2))\pi^{(\beta)}(\theta_2|y_2)d\theta_2.\nonumber
\end{align}
where $c_{\mathcal{S}^{(1)}}$ and $c_{\mathcal{S}^{(2)}}$ are defined in Condition A.4. 
\label{Thm:StabilityPosteriorPredictivebetaDiv2}
\end{theorem}


\color{black}
Theorems \ref{Thm:StabilityDGPapproxBeta2} and \ref{Thm:StabilityPosteriorPredictivebetaDiv2} are the analogous result to Theorems \ref{Thm:StabilityDGPapproxBeta} and \ref{Thm:StabilityPosteriorPredictivebetaDiv} respectively with terms $C^{(\beta)}(f,y_1, y_2)$ and $c$ having the same interpretation. The value $M$ is still easy to bound here and Corollaries A.1 and A.2 demonstrate that $C^{(\beta)}(f,y_1, y_2)\rightarrow 0$ and  $\frac{1}{c} \rightarrow 0$, as $n\rightarrow\infty$. 
Therefore, Theorem's \ref{Thm:StabilityDGPapproxBeta2} and \ref{Thm:StabilityPosteriorPredictivebetaDiv2} establish that \textBD-Bayes inferences will be similar for any two \DGPs satisfying Definition \ref{Def:DGPNeighbourhood}. This allows a \DM to use their model and know that small unforeseen perturbations of the \DGP will not drive substantially different posterior inference or alternatively use a default model or software from the literature and know that as long as their application area is similar, the model's generalisation will not overly affect posterior inferences.


Once again, we do not invoke a comparison of the bounds from Lemma \ref{Thm:StabilityDGPapproxKLD2} and Theorem \ref{Thm:StabilityDGPapproxBeta2}, as they are bounding different quantities. Instead, we consider the strength of the sufficient conditions required for boundedness. \KLD-Bayes requires strict conditions for the \DGPs that are difficult for the \DM to know would be satisfied, while the \textBD-Bayes is stable across reasonable generalisation of a \DGP.

\color{black}

\section{Setting $\beta$}{\label{Sec:SettingBeta}}

\color{black}

To implement \textBD-Bayes inference it is obviously necessary to choose an appropriate value of $\beta$. We briefly review below a variety of methods that have been proposed to do this and comment on how these relate to the results of this paper. We then demonstrate that inference is not that sensitive to this choice provided that $\beta$ is not chosen too close to one.


\subsection{Data Driven Methods}

One approach is to try to learn a value for $\beta$ that is `optimal' in some sense for the \DM's functioning likelihood model and the particular observed data at hand. Once the \DM has decided upon model $f(\cdot; \theta)$, the value $\beta$ regulates the trade-off between robustness and efficiency \citep[e.g.][]{basu1998robust}. Minimising the \KLD ($\beta=1$) provides the most efficient inference but is very sensitive to outliers. Increasing $\beta$ away from 1 gains robustness to outliers at a cost to efficiency. 
\cite{warwick2005choosing, ghosh2015robust, basak2021optimal} seek to optimise the robustness-efficiency trade-off by estimating $\beta$ to minimise the mean squared error (\MSE) of estimated model parameters, \cite{toma2011dual, kang2014minimum} minimise the maximum perturbation of the parameter estimates resulting from replacing one observation by the population estimated mean, and \cite{ yonekura2023adaptation} build on the work of \cite{jewson2022general} to estimate $\beta$ minimising the Fisher's divergence to the \DGP.
The intuition behind these methods is that values estimated close to $\beta=1$ indicate the model $f$ is pretty well specified for the data at hand, while larger values indicate increasing large levels of possible model misspepcifcation. We use the method of \cite{ yonekura2023adaptation} to learn the value of $\beta$ in the example in Section \ref{Sub:GaussianStudent}.

The results of this paper provide a \DM who has used one of the above methods to set a value for $\beta$ an indication of how sensitive their posterior inferences could be to the specification of their model and the data. For example, a \DM learning a larger value for $\beta$ knows that the term $\frac{M^{\beta - 1}}{(\beta - 1)}\epsilon$ will be small for any value of $\epsilon$ and that even large departures from their model or data would result in similar inferences. This suggests they will only be able to learn slowly about the \DGP. On the other hand, a \DM estimating a very small $\beta$ knows that their posterior inference may be very dependent on the precise model class they have chosen.




\subsection{User Specified}{\label{Sec:UserSpecifiedBeta}}

Other works have advocated for the subjective specification for the value of $\beta$ \citep[e.g.][]{jewson2018principles} 
and the results in this paper help to facilitate this by interpreting $\beta$ in terms of the level of stability it brings. 
The results of this paper demonstrate that $\beta$ controls the amount that imprecision in the specification of the model or data can be magnified into the posterior, allowing for the interpretation of $\beta$ as a meta prior for the \DM's confidence in their elicited model or data collection. The less confident they are, the greater $\beta$ will need to be to prevent non-negligible \textit{a posteriori} divergence. Eliciting $\beta$ as such requires the \DM to reflect on the value of $\epsilon$ associated with their beliefs or the quality of the data. For the neighbourhoods of Definition \ref{Def:LikelihoodNeighbourhood}, this can be obtained by considering for a given set of parameters what the largest possible error in any of the probability statements could be, or for Definition \ref{Def:DGPNeighbourhood} by considering the minimal proportion of a population that they believe is consistent with the \DGP. 

A default implementation, however, would be to set $\beta$ such that $\frac{M^{\beta - 1}(3\beta - 2)}{\beta(\beta - 1)} = U$ 
ensuring that the posterior predictive imprecision as measured by Theorem \ref{Thm:StabilityDGPapproxBeta} is only $U > 1$ times the level of prior imprecision $\epsilon$. We demonstrate such an approach for the example in Section \ref{Sec:Classification} for $U = 2$. 
%
%
Importantly, a \DM could not hope to set $\beta$ to provide maximal stability. 
Maximum stability, i.e. minimising the right-hand side of the bounds in Theorems \ref{Thm:StabilityDGPapproxBeta} and \ref{Thm:StabilityDGPapproxBeta2} would set $\beta \rightarrow\infty$ and result in the posterior under any model and data collapsing to the prior, providing absolutely stable inference but not learning anything from the data. For minimally efficient learning to take place, posterior beliefs should not be closer, in the worst case, than the models were a priori.


\subsection{Sensitivity}

Finally, \textBD-Bayes inference appears not to be overly sensitive to the exact value of $\beta$. Figure \ref{Fig:norm_t_neighbourhood_predictives_sensitivity} demonstrates that for the example introduced in Section \ref{Sec:Introduction}, inference for the Gaussian and \Student models is almost identical for values of $\beta\geq 1.2$. Section B.1 provides further demonstration of this.

\color{black}

\begin{figure}
\begin{center}
\includegraphics[trim= {0.0cm 0.00cm 0.4cm 0.0cm}, clip,  
width=0.32\columnwidth]{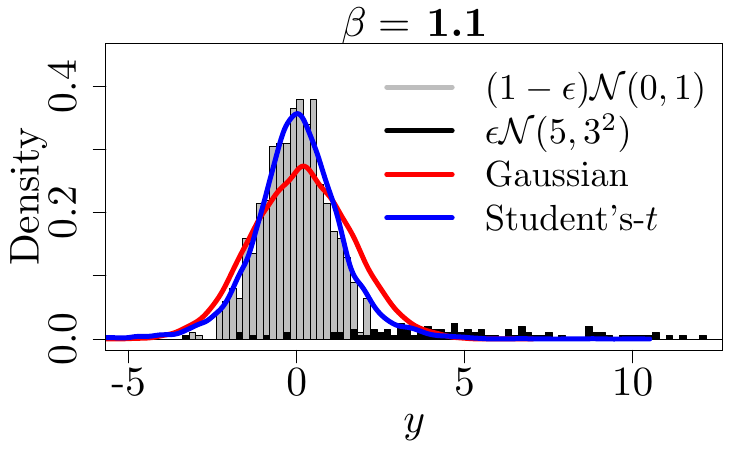}
\includegraphics[trim= {0.0cm 0.00cm 0.4cm 0.0cm}, clip,  
width=0.32\columnwidth]{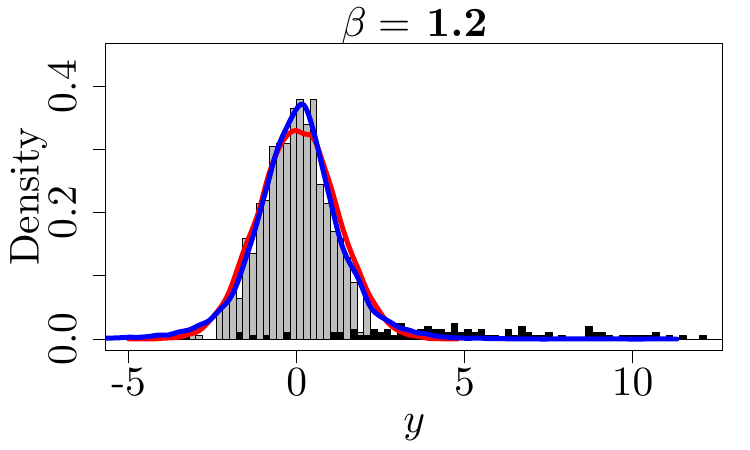}
\includegraphics[trim= {0.0cm 0.00cm 0.4cm 0.0cm}, clip,  
width=0.32\columnwidth]{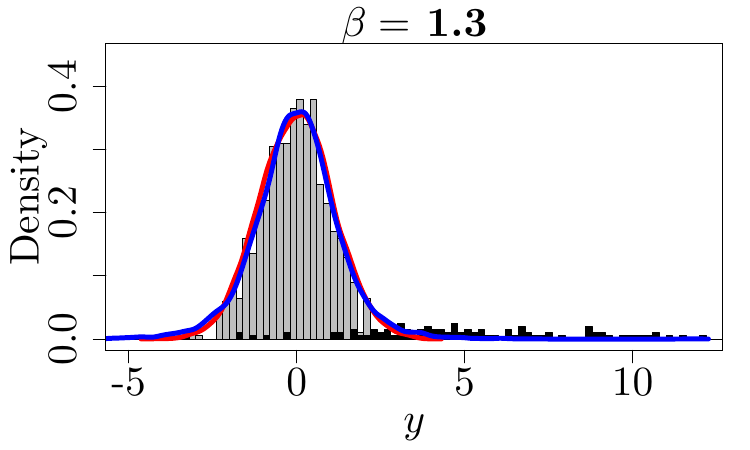}\\
\includegraphics[trim= {0.0cm 0.00cm 0.4cm 0.0cm}, clip,  
width=0.32\columnwidth]{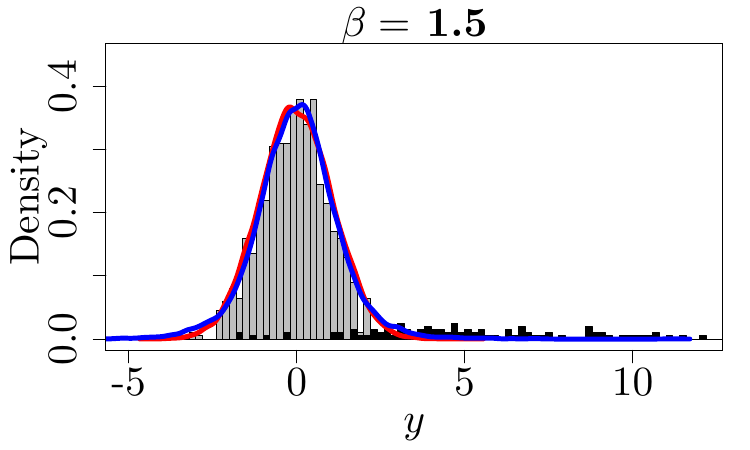}
\includegraphics[trim= {0.0cm 0.00cm 0.4cm 0.0cm}, clip,  
width=0.32\columnwidth]{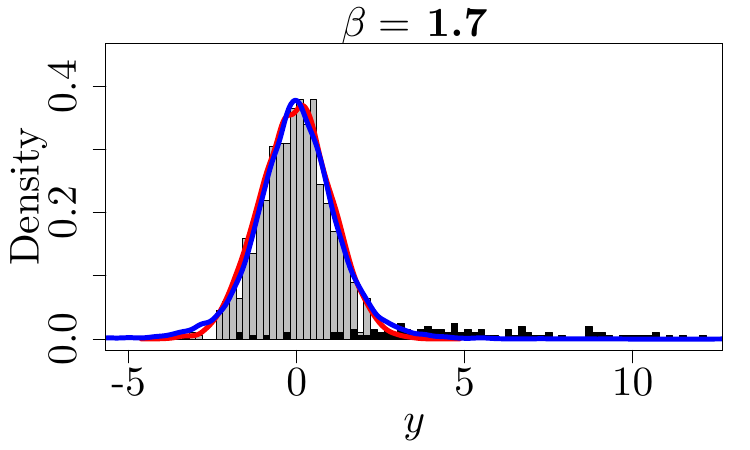}
\includegraphics[trim= {0.0cm 0.00cm 0.4cm 0.0cm}, clip,  
width=0.32\columnwidth]{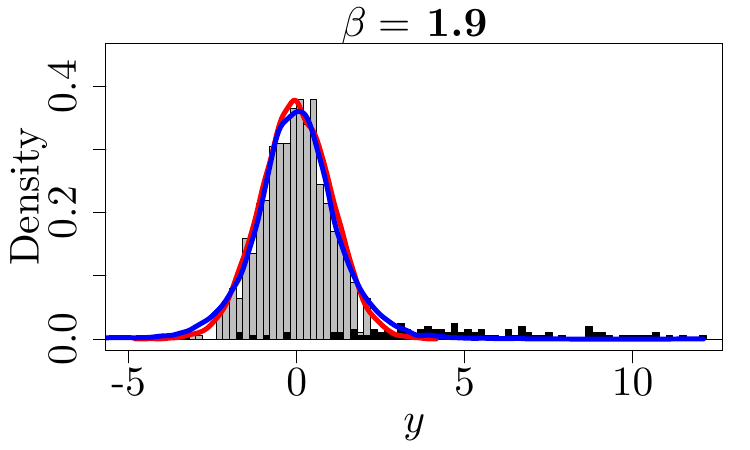}
\caption{Posterior predictive distributions using  {\color{black}{\textBD-Bayes}} updating on $n=1000$ observations from an $\epsilon$-contamination model $g(y) = 0.9\times\mathcal{N}\left(y;0,1\right) + 0.1 \times \mathcal{N}\left(y;5,3^2\right)$ for different values of $\beta$.}
\label{Fig:norm_t_neighbourhood_predictives_sensitivity}
\end{center}

\end{figure}

\section{Experiments}{\label{Sec:Experiments}}


\subsection{Gaussian and \Student likelihood}{\label{Sub:GaussianStudent}}

We revisit the Gaussian and \Student example briefly introduced in Section \ref{Sec:Introduction}. The likelihood models considered here are 
\begin{align}
f_{\sigma^2_{adj}}(y;\theta):= \mathcal{N}\left(y;\mu,\sigma^2\times\sigma^2_{adj}\right)\textrm{ and }h_{\nu}(y;\eta):= \textrm{Student's}-t_{\nu}\left(y;\mu,\sigma^2\right).\label{Equ:GaussianStudent}
\end{align}
Hyperparameters, $\nu=5$ and $\sigma^2_{adj}=1.16$ are fixed to match the quartiles of the two distributions for all $\mu$ and $\sigma^2$. These were inspired by \cite{o2012probabilistic}, who argued that for absolutely continuous probability distributions, it is only reasonable to ask an expert to make a judgement about the median and the quartiles of a distribution along with maybe a few specially selected features. 
This is justified as adequate as any two distributions with similar percentiles will look very similar, see for example Figure \ref{Fig:norm_t_neighbourhood_predictives}.
However, Section \ref{Sub:StabilityLikelihood_KLD} suggests that greater precision is required to ensure the stability of \bayesrule updating. 
On the other hand, the likelihoods in \eqref{Equ:GaussianStudent} are contained in $\mathcal{N}^{\TVD}_{0.043}$. We generated $n=1000$ observations from the $\epsilon$-contamination model $g(x) = 0.9\times\mathcal{N}\left(y;0,1\right) + 0.1 \times \mathcal{N}\left(y;5,3^2\right)$ contained within the $\mathcal{G}^{\TVD}_{0.1}$ neighbourhood of $\mathcal{N}\left(y;0,1\right)$. 
We then conducted Bayesian updating under the Gaussian and Student's-$t$ likelihood using both \bayesrule and the \textBD-Bayes 
under shared priors $\pi(\mu,\sigma^2) = \mathcal{N}\left(\mu;\mu_0,v_0\sigma^2\right)\mathcal{IG}(\sigma^2;a_0,b_0)$, with hyperparameters $(a_0=0.01, b_0=0.01, \mu_0=0, v_0=10)$. 
\color{black}
We used the method of \cite{yonekura2023adaptation} to set $\beta = 1.22$ when using the Gaussian distribution and use the same value for the Student's-$t$. 
\color{black}
Figure \ref{Fig:norm_t_neighbourhood_predictives} and Figure B.1, which plots the parameter posterior distributions for both models under both updating mechanisms, clearly demonstrate the stability of the \textBD-Bayes across these two models and the lack of stability of traditional Bayesian updating. 
Not only is the \textBD inference more stable across $\mathcal{N}^{\TVD}_{\epsilon}$, the \textBD predictive better captures the majority of the \DGP than either of the \KLD-Bayes predictives. The capturing of the $\mathcal{N}\left(y;0,1\right)$ mode further illustrates the \textBD-Bayes' stability across neighbourhoods of the \DGP.

Figure \ref{Fig:norm_t_influence_functions} plots influence functions \citep{west1984outlier} for the \KLD-Bayes and \textBD-Bayes under the Gaussian and \Student model. Influence functions are the gradient of the loss function evaluated at parameter estimates as a function of the observations and show the impact that observation had on the analysis. Under the \textBD-Bayes, the influence functions of the Gaussian and \Student likelihoods are closer for almost every $y$, illustrating the stability to the model, and additionally, the influence functions for both models under the \textBD-Bayes vary less with $y$, illustrating stability to the \DGP.

\begin{figure}
\begin{center}
\includegraphics[trim= {0.0cm 0.00cm 0.0cm 0.0cm}, clip,  
width=0.49\columnwidth]{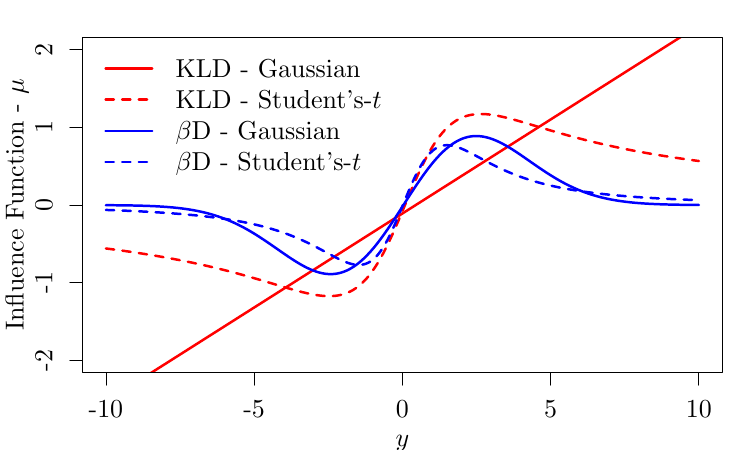}
\includegraphics[trim= {0.0cm 0.00cm 0.0cm 0.0cm}, clip,  
width=0.49\columnwidth]{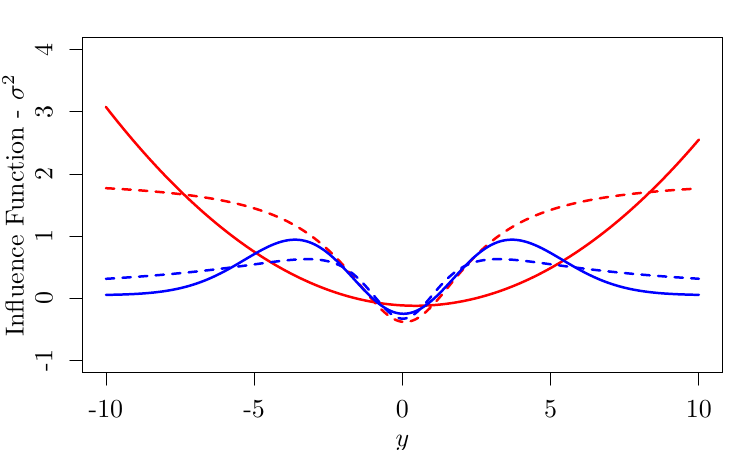}
\caption{Influence functions for parameter $\mu$ and $\sigma^2$ of the Gaussian and \Student likelihood models under the \KLD-Bayes and \textBD-Bayes with $\beta = 1.22$.
}
\label{Fig:norm_t_influence_functions}
\end{center}

\end{figure}

\subsubsection{\DLD data}{\label{Sub:DLD}}

We consider an RNA-sequencing data set from \cite{yuan2016plasma} measuring gene expression for $n = 192$ patients with different types of cancer. \cite{rossell2018tractable} studied the impact of 57 predictors on the expression of \DLD, a gene that can perform several functions such as metabolism regulation. To illustrate our results, we selected the 15 variables with the 5 highest loadings in the first 3 principal components, and fitted regression models using the neighbouring models in \eqref{Equ:GaussianStudent} for the residuals. Section B.1 lists the selected variables. 
\color{black}
Once again, we used the method of \cite{yonekura2023adaptation} to set $\beta = 1.34$ when using the Gaussian distribution, and use the same value for the Student's-$t$. 
\color{black}

Figure \ref{Fig:DLDRegressions} demonstrates that \textBD-Bayes produces more stable estimates of the fitted residuals (top-left), the estimated density of the residuals (top-right), parameter estimates (bottom-left), and posterior predictive density for the observed data (bottom-right) than the traditional Bayesian inference. \cite{rossell2018tractable} found evidence that this data is heavy-tailed, further demonstrated in Figure B.5, which caused the \KLD-Bayes to estimate very different densities under the Gaussian and \Student model, while the \textBD-Bayes is stable to this feature of the data. Figure B.4 shows the fit of the models to the posterior mean estimates of the standardised residuals, showing that as well as being stable, the \textBD-Bayes produces good estimation around the mode of the \DLD data under both models. 
Section B.1 considers a further regression example showing that even when one of the models under consideration is `well-specified' for the data, the \textBD-Bayes inference continues to perform adequately.

\begin{figure}[!ht]
\begin{center}
\includegraphics[trim= {0.0cm 0.00cm 0.0cm 0.0cm}, clip,  
width=0.49\columnwidth]{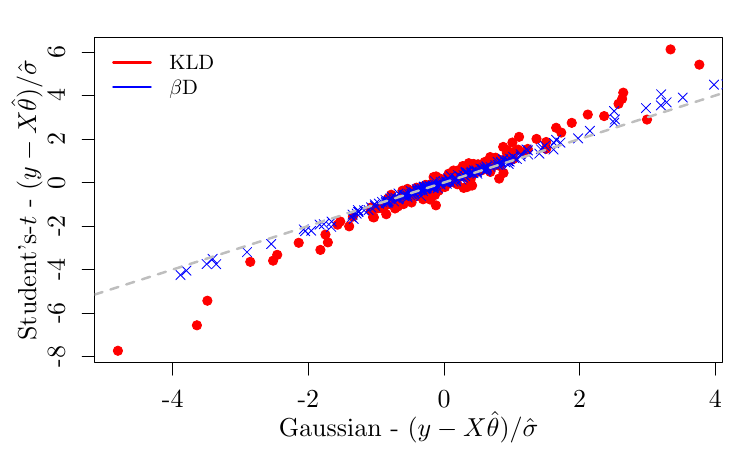}
\includegraphics[trim= {0.0cm 0.00cm 0.0cm 0.0cm}, clip,  
width=0.49\columnwidth]{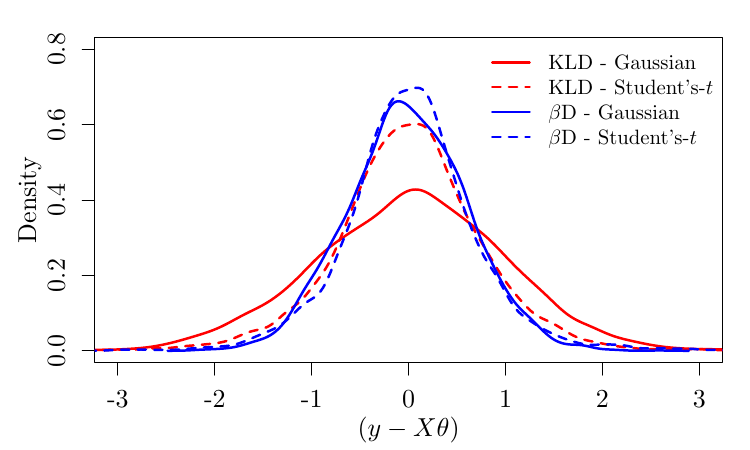}\\
\includegraphics[trim= {0.0cm 0.00cm 0.0cm 0.0cm}, clip,  
width=0.49\columnwidth]{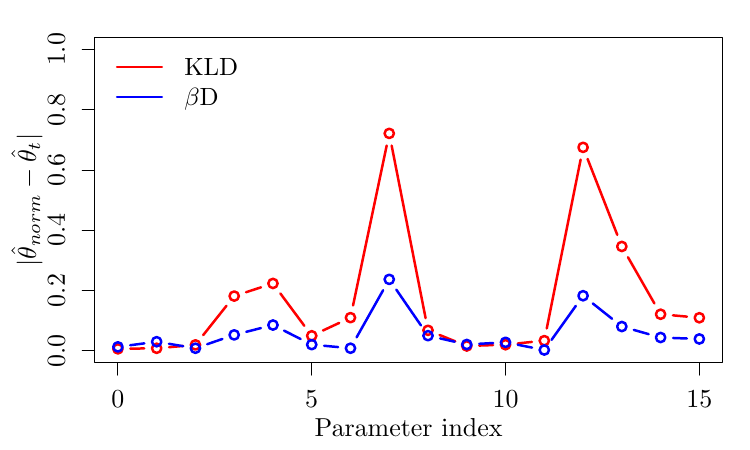}
\includegraphics[trim= {0.0cm 0.00cm 0.0cm 0.0cm}, clip,  
width=0.49\columnwidth]{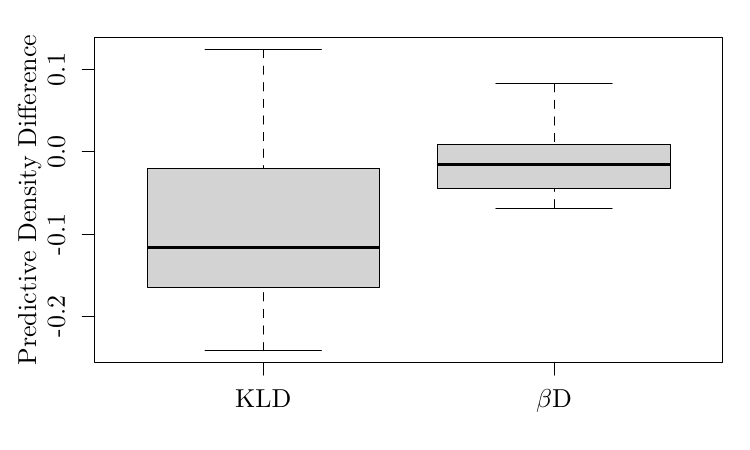}\\
\caption{Posterior mean estimates of standardised residuals (\textbf{top left}), posterior predictive residual distributions (\textbf{top-right}), absolute difference in posterior mean parameter estimates (\textbf{bottom left}) and difference in posterior predictive densities of the observations (\textbf{bottom right}) under the Gaussian and \Student model of \KLD-Bayes and \textBD-Bayes ($\beta = 1.34$) for the \DLD data.}
\label{Fig:DLDRegressions}
\end{center}
\end{figure}

\subsection{Binary Classification}{\label{Sec:Classification}}



Binary classification models predict $y \in \{0, 1\}$ from $p$-dimensional regressors $X$. The canonical model in such a setting is logistic regression where
\begin{align}
    P_{LR}(y = 1| X, \theta) = \frac{1}{1 + \exp\left(- X\theta\right)}, \quad P_{LR}(y = 0 | X, \theta) = 1 - P_{LR}(Y = 1| X, \theta)\nonumber,
\end{align}
where $\theta\in\mathbb{R}^p$ are the regression parameters. Alternative, less ubiquitous models include probit regression, which uses an alternative \GLM link function depending on the standard Gaussian \CDF $\Phi(\cdot)$, `heavier tailed' $t$-logistic regression \citep{ding2010t, ding2013t} and a mixture type model that explicitly models the chance of mislabelling of the observed classes.
\begin{align}
    P_{PR}(y = 1| X, \eta) &= \Phi(w_{PR}X\theta), \quad P_{tLR}(y = 1| X, \eta) = \exp_t((0.5w_{tLR}X\theta-G_t(w_{tLR}X\theta)))\nonumber\\
    &P_{ML}(y = 1| X, \eta) = (1-\nu_1)P_{LR}(y = 1 | X, \theta) + \nu_0(1-P_{LR}(y = 1 | X, \theta))\nonumber
\end{align}
where $0 < t < 2$,  $0 < \nu_0, \nu_1 < 1$, `$\exp_t$' is the so-called $t$-exponential and $G_t$ ensures that $P_{tLR}(y = 1| X, \eta)$ is normalised, both are defined in Section B.3. Setting $t > 1$ results in heavier-tailed probabilities than the logistic model. For the probit and $t$-logistic models parameters $\theta$ are scalar multiples $w_{PR}, w_{tLR}\in\mathbb{R}$ of the logistic regression parameters $\theta \mapsto w\theta$. These are calculated in order to minimise the \textit{a priori} \TVD between the models and the logistic regression baseline according to $\mathcal{N}_{\epsilon}^{\TVD}$ (see Section B.3). We upper bound $\nu_0$ and $\nu_1$ by 0.05 making $\epsilon = 0.05$ for these models. Figure \ref{Fig:BinaryClassifiers} plots $P(y = 1 | X, \theta)$ as a function of $X\theta$ for all four models (left) and the \TVD between each alternative model and the logistic regression (right), demonstrating that all four produce very similar binary probabilities. 

\begin{figure}
\begin{center}
\includegraphics[trim= {0.0cm 0.00cm 0.0cm 0.0cm}, clip,  
width=0.49\columnwidth]{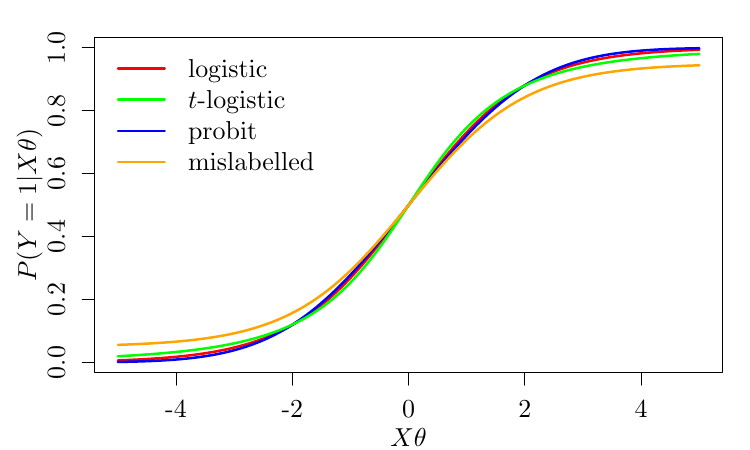}
\includegraphics[trim= {0.0cm 0.00cm 0.0cm 0.0cm}, clip,  
width=0.49\columnwidth]{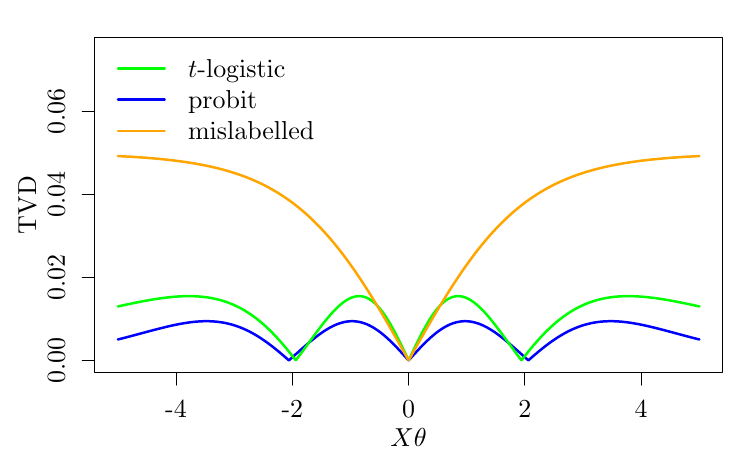}
\caption{\textbf{Left}: $P(y = 1 | X, \theta)$ for logistic, probit, $t$-logistic and mislabelled models. \textbf{Right}: \TVD between the logistic regression canonical model and the probit, $t$-logistic and mislabelled models. The $\theta$ parameters of the probit and $t$-logistic models are scalar multiplied in a fashion that minimise the \TVD to the logistic regression}
\label{Fig:BinaryClassifiers}
\end{center}
\end{figure}

\subsubsection{Colon Cancer Dataset}

To investigate the stability of posterior predictive inferences across the logistic, probit, $t$-logistic, and mislabelled binary regression models we consider the colon cancer dataset of \cite{alon1999broad}. The dataset contains the expression levels of 2000 genes from 40 tumours and 22 normal tissues and there is purportedly evidence that certain tissue samples may have been cross-contaminated \citep{tibshirani2013robust}. Rather than consider the full 2000 genes we first run a frequentist LASSO procedure, estimating the hyperparameter via cross-validation, and focus our modelling only on the nine genes selected by this procedure. We understand that such post-model selection biases parameter estimates, but the stability of the predictive inference is our focus here. 
\color{black}
We set $\beta = 2$ so that $U := \frac{M^{\beta - 1}(3\beta - 2)}{\beta(\beta - 1)} = 2$ with $M = 1$ as was proposed in Section \ref{Sec:UserSpecifiedBeta}.
\color{black}

Figure \ref{Fig:TVDColonCancer} compares the
\textit{a posteriori} \TVD distance between the posterior predictive distributions for each observation with the
\textit{a priori} \TVD distance between each of the models  (top) and the difference between the posterior mean regression parameter estimates of the two models (bottom)  under the \KLD-Bayes and \textBD-Bayes. The stability of the \textBD-Bayes is once again demonstrated here. For almost every observation and every pair of models, the posterior predictive inference is as stable as it was \textit{a priori}, while the KLD-Bayes inference is more often divergent. For the $t$-logistic and mislabelled models the predictive stability of the \textBD-Bayes also provides greater stability in the posterior mean parameter estimates.

\begin{figure}[!ht]
\begin{center}
\includegraphics[trim= {0.0cm 0.00cm 0.0cm 0.0cm}, clip,  
width=0.32\columnwidth]{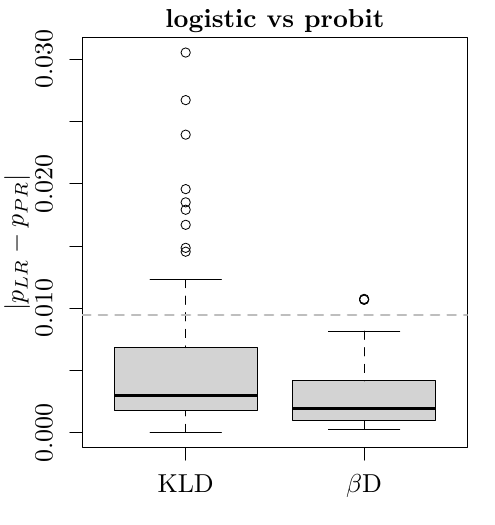}
\includegraphics[trim= {0.0cm 0.00cm 0.0cm 0.0cm}, clip,  
width=0.32\columnwidth]{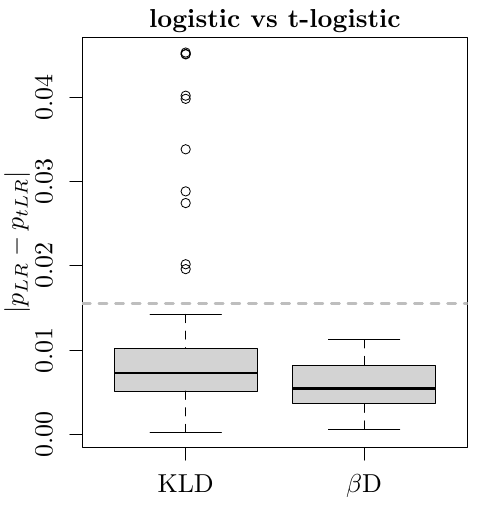}
\includegraphics[trim= {0.0cm 0.00cm 0.0cm 0.0cm}, clip,  
width=0.32\columnwidth]{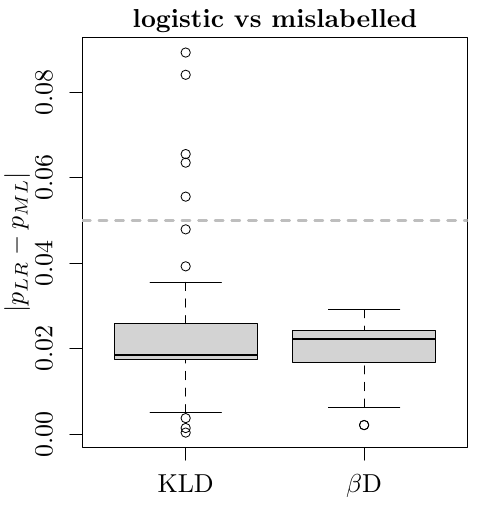}\\
\includegraphics[trim= {0.0cm 0.00cm 0.0cm 0.0cm}, clip,  
width=0.32\columnwidth]{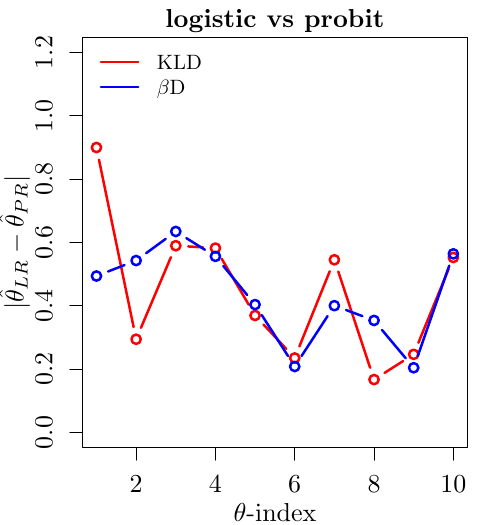}
\includegraphics[trim= {0.0cm 0.00cm 0.0cm 0.0cm}, clip,  
width=0.32\columnwidth]{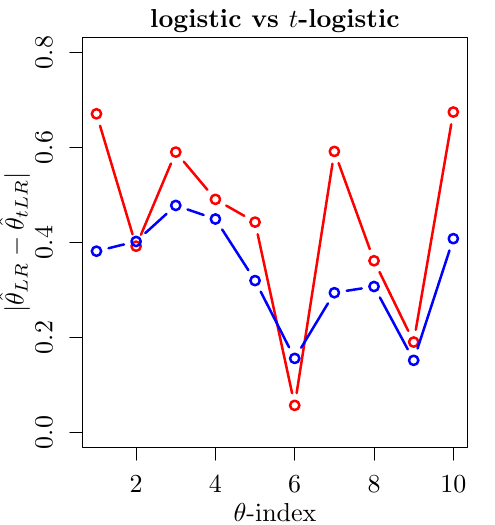}
\includegraphics[trim= {0.0cm 0.00cm 0.0cm 0.0cm}, clip,  
width=0.32\columnwidth]{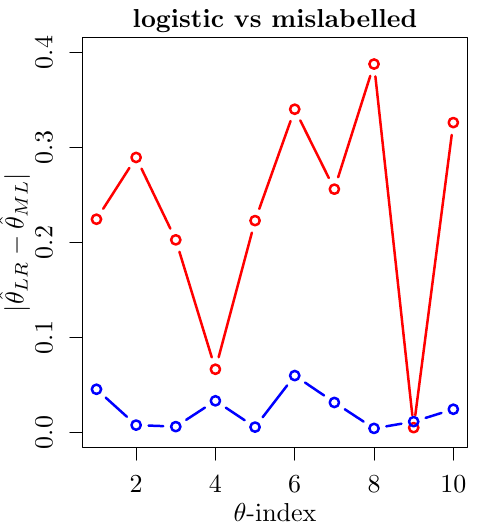}
\caption{Colon Cancer Data. \textbf{Top}: \TVD between the posterior predictive estimated probabilities for each observation of the probit (\textbf{left}), $t$-logistic (\textbf{centre}) and mislabelled (\textbf{right}) models and the canonical logistic regression under the \KLD-Bayes and \textBD-Bayes ($\beta = 2$). The dotted line represented the \textit{a priori} \TVD distance between the models.
\textbf{Bottom}: Absolute differences between posterior mean parameter estimates 
and those of the logistic regression.
}
\label{Fig:TVDColonCancer}
\end{center}
\end{figure}

\subsection{Mixture Modeling}{\label{Sub:MixtureModeling}}

An advantage of considering the stability of the distributions for observables rather than parameters is that it allows `neighbouring' models to have different dimensions to their parameter space. For example, consider initial model $f(\cdot;\theta)$ and then `neighbouring' model
\begin{align}
h(\cdot;\eta)& = (1 - \omega)\times f(\cdot;\theta) + \omega\times h^{'}(\cdot;\kappa),\nonumber
\end{align}
for $\eta = \left\lbrace \theta,\kappa, \omega\right\rbrace$. Here, 
$h(\cdot;\eta)$ 
is a mixture model combining the likelihood model $f(\cdot;\theta)$, which could itself already be a mixture model, and some other density $h^{'}(\cdot;\kappa)$ with additional parameters $\kappa$. For all $\theta\in\Theta$ and any $\kappa \in K$ we have that $\TVD(f\left(\cdot; \theta\right), h\left(\cdot;\left\lbrace \theta,\kappa, \omega\right\rbrace\right))<\omega$ and therefore a \TVD neighbourhood can be defined by upper bounding $\omega$. 

\subsubsection{Shapley Galaxy Dataset}

We examine the Shapley galaxy dataset of \cite{drinkwater2004large}, recording the velocities of 4215 galaxies in
the Shapley supercluster, a large concentration of gravitationally-interacting galaxies; see Figure \ref{Fig:MixtureModels}.
The clustering tendency of galaxies continues to be a subject of interest in astronomy. \cite{miller2018robust} investigate this data using Gaussian mixture models and use their coarsened posterior to select the number of mixture components, finding considerable instability in the number of estimated components $K$ under different specifications of the coarsening parameter. See \cite{cai2021finite} for further issues with estimating the number of components in mixture models.

We estimate Gaussian mixture models of the form 
\begin{align}
f(y; \theta)= \sum_{k=1}^K \omega_j \mathcal{N}(y; \mu_j, \sigma_j),
    \nonumber
\end{align}
under the \KLD-Bayes and \textBD-Bayes, considering number of components $K \in \{2, 3, 4, 5, 6\}$ and using the normal-inverse Wishart priors of \cite{fuquene2019choosing} (full details available in Section B.2). 
\textBD-Bayes inference for such one-dimensional mixture models is easy to implement using adaptive quadrature to approximate the necessary integral term $\frac{1}{\beta}\int h(z;\eta)^{\beta}dz$. 
We do not formally place any constraint on the estimation of $\omega_k$, however, any model that estimates a component with small $\omega_k$ can be seen as a neighbour of a model with one fewer component.

Figure \ref{Fig:MixtureModels} shows the posterior predictive approximation to the histogram of the data of the Gaussian mixture models under the \KLD-Bayes and \textBD-Bayes and Table \ref{Tab:MixtureModels} records the \TVD between the posterior predictive distribution of recursively adding components to the model. The \textBD-Bayes inference for $\beta = 1.25$ and $1.5$ is more stable to the addition of an extra component. In particular, for $K \geq 3$ the \textBD-Bayes inference stably estimates the biggest components of the data centered approximately at  $5,000$ and $15,000$ $km/s$, while the \KLD-Bayes produces very different inference for these modes depending on the number of clusters selected.

\begin{figure}
\begin{center}
\includegraphics[trim= {0.0cm 0.00cm 0.0cm 0.0cm}, clip,  
width=0.9\columnwidth]{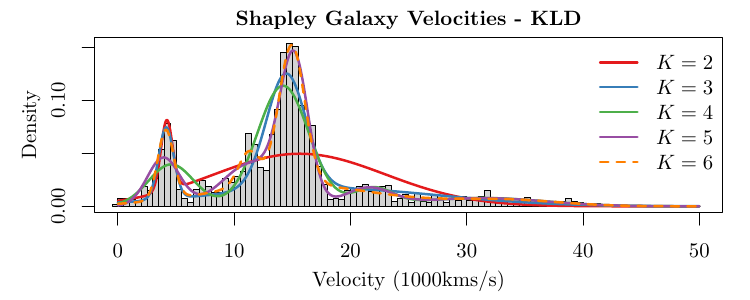}
\includegraphics[trim= {0.0cm 0.00cm 0.0cm 0.0cm}, clip,  
width=0.9\columnwidth]{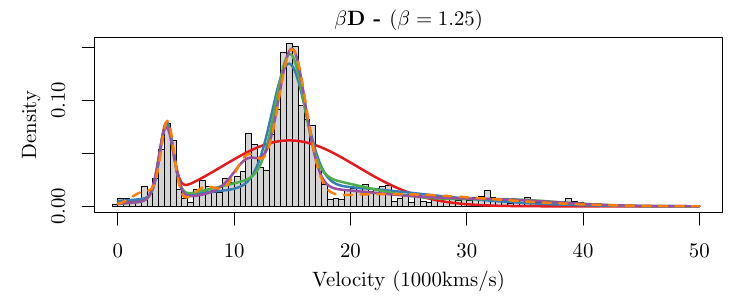}
\includegraphics[trim= {0.0cm 0.00cm 0.0cm 0.0cm}, clip,  
width=0.9\columnwidth]{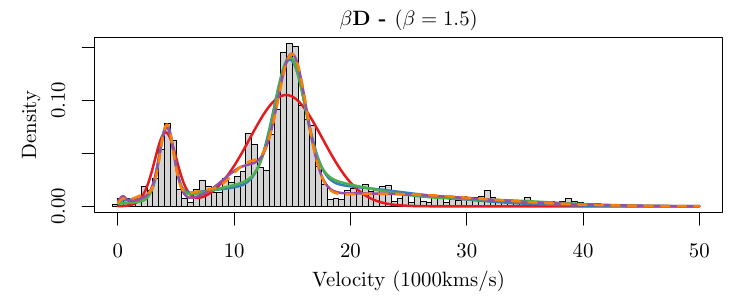}
\caption{Shapley Galaxy Data: Histograms of the data, in units of 1,000 km/s, excluding a small amount of
data extending in a tail up to 80,000 km/s, and posterior predictive distributions of the fitted Gaussian mixture models with $K = 2-6$ components under the \KLD-Bayes (\textbf{top}), \textBD-Bayes with $\beta = 1.25$ (\textbf{middle}) and \textBD-Bayes with $\beta = 1.5$ (\textbf{bottom}).}
\label{Fig:MixtureModels}
\end{center}
\end{figure}

\begin{table}[ht]
\caption{Total variation distances between \color{black} posterior predictive distributions \color{black} for different number of mixture components $K$ under the \KLD-Bayes and \textBD for $\beta = 1.25$ and $1.5$.}
\vspace{1em}
\centering
\begin{tabular}{rcccc}
  \hline
Method &  $K = 2$ vs $3$ & $K = 3$ vs $4$ & $K = 4$ vs $5$ & $K = 5$ vs $6$ \\ 
    \hline
  \KLD & 0.27 & 0.12 & 0.13 & 0.08 \\ 
  \textBD ($\beta = 1.25$) & 0.26 & 0.06 & 0.06 & 0.05 \\ 
  \textBD ($\beta = 1.5$) & 0.22 & 0.04 & 0.07 & 0.02 \\ 
  \hline
\end{tabular}
\label{Tab:MixtureModels}
\end{table}

\section{Discussion}

This paper investigated the posterior predictive stability of traditional Bayesian updating and a generalised Bayesian alternative minimising the \textBD. 
In practice, the model used for inference is usually a convenient and canonical interpolation of the broad belief statements made by the \DM and the observed data was not necessarily collected in the manner the \DM imagined. 
We proved that \textBD-Bayes inference is provably stable across a class of likelihood models and data generating processes whose probability statements are absolutely close, a \TVD neighbourhood, by establishing bounds on how far their predictive inferences can diverge. On the other hand, 
our results require the \DM to be sure about the tail properties of their beliefs and the \DGP to guarantee stability for standard Bayesian inference. 

The results of this paper simplify the process of belief elicitation for the \textBD-Bayes, bounding the \textit{a posteriori} consequences for a given level of \textit{a priori} inaccuracy, leaving the \DM free to use the best guess approximation of their beliefs that they are most comfortable with, rather than switch to a less familiar model with better outlier rejection properties \citep{o1979outlier}.
Such stability is achieved through a minimal amount of extra work compared with traditional \bayesrule inference, and it provides a similarly recognisable output. We hope such results help to justify the increased use of the \textBD to make robust inferences in statistics and machine learning applications.


A key issue motivating the departure from standard Bayesian methods here is a lack of concordance between the likelihood model and the data. Such an issue can be attributed to either a failure of the modeller to think carefully enough about the \DGP, or errors in data collection. However, we treat these results separately to exemplify two different manifestations of the instability of Bayes' rule.


\color{black}
The main limitation of our work is that we do not consider a universal measure of posterior predictive stability. Lemmas \ref{Thm:StabilityDGPapproxKLD} and \ref{Thm:StabilityDGPapproxKLD2} use the \KLD divergence to the \DGP and Theorems \ref{Thm:StabilityDGPapproxBeta} and \ref{Thm:StabilityDGPapproxBeta2} use the \textBD. 
It could of course be reasonably argued that the \TVD should be used directly. However, the \TVD is notoriously difficult to compute directly for large problems and is complicated by the intractability of the \KLD and \textBD-Bayes posterior distributions. So instead, we focused on comparing the strength of the sufficient conditions required by each method for some measure of stability and used examples to indicate that these translate into meaningful differences. 

The feasibility of \textBD-Bayes is dependent on the model likelihood being available in closed form -- although robust general Bayesian method exists to deal with cases when it is not \citep{matsubara2021robust} -- and the integral term in \eqref{Equ:betaDloss} being either available in closed form or fast to approximate accurately. These conditions are met by many standard models including exponential family and Student's-$t$ models. When they are not then there are various method available to make such calculations. For example, quadrature can be used for low-dimensional data. This integral is over the data not parameters and is therefore invariant to the parametrisation of the model. Further, one contribution of this paper is to show that the \textBD-Bayes allows a \DM to use a canonical model, where this integral would be available, in place of their true beliefs and know
that any only approximately probabilistic specifications this might make will not have had an undue influence on their inference.
\color{black}

Future work could explore the applicability of such results in multivariate settings where belief specification and data collection are harder, and further investigate our \KLD-Bayes results. While we argued when you could guarantee the stability of such methods, identifying for which statements \KLD-Bayes is not stable would provide important and useful results to facilitate more focused belief elicitation.

To continue to facilitate the deployment of \textBD-Bayes methods in practice, more work is required to study and build upon existing methods to select $\beta$, particularly in high dimensions. While it is clear that considerable gains can be made over standard methods in certain scenarios, an adversarial analysis of the \textBD performance compared with its \KLD-Bayes analogue would further motivate its wider applications. 
\color{black}
Other, interesting theoretical developments could seek to extend the posterior predictive stability to the stability of marginal posterior distributions in the case where there is an interpretable parameter of interest that is shared across models. 
\color{black}


\section*{Acknowledgements}

The authors would like to thank Danny Williamson, Christian Robert, and Sebastian Vollmer for their insightful discussions on the topics in this paper. We would also like to thank two anonymous reviewers, the Associated Editor and the Editor for their help in improving this paper. JJ was partially funded by the Ayudas Fundación BBVA a Equipos de Investigación Cientifica 2017, the Government of Spain's Plan Nacional PGC2018-101643-B-I00, and a Juan de la Cierva Formación fellowship  FJC2020-046348-I. CH was supported by the EPSRC Bayes4Health programme grant and both CH and JQS were supported by the Alan Turing Institute, UK.




\bibliography{bib}

\newpage

\appendix



Section \ref{Sec:DefsProofsConds} contains additional information for our theoretical analysis, including definitions of the Kullback-Leibler Divergence (\KLD) and $\beta$-divergence (\textBD), full definitions of notation and technical conditions and proofs of the results of Sections 3 and 4. Section \ref{Sec:ExtExperimentalResults} contains additional details of the experimental results of Section 6, including full specifications of the models and data used, as well as additional sensitivity analysis for $\beta$ and a comparison with the $\gamma$-Divergence (\textGD).

\section{Definitions, Conditions and Proofs}{\label{Sec:DefsProofsConds}

\subsection{Divergence Definitions}{\label{sec:DivergenceDefinitions}}

Here we provide definitions of the \KLD and \textBD.

\begin{definition}[The Kullback-Leibler Divergence (\KLD) \citep{kullback1951information}]
The \KLD between probability densities $g(\cdot)$ and $f(\cdot)$ is given by
\begin{equation}
    \KLD(g||f) = \int g\log\frac{g}{f}d\mu.\nonumber
\end{equation}
\end{definition}

%

%
\begin{definition}[The $\beta$-divergence (\textBD)  \citep{basu1998robust,mihoko2002robust}]
The \textBD  is defined as
\begin{align}
\BD(g||f) = \frac{1}{\beta(\beta-1)}\int g^{\beta}d\mu+\frac{1}{\beta}\int f^{\beta}d\mu -\frac{1}{\beta-1}\int gf^{\beta-1}d\mu,\nonumber
\end{align}
where $\beta\in\mathbb{R}\setminus \left\lbrace0,1\right\rbrace$. 
The \textBD is Bregman-divergence \citep{bregman1967relaxation} with associated function $\psi(t)=\frac{1}{\beta(\beta-1)}t^{\beta}$. When $\beta = 1$, $D_B^{(1)}(g(x)||f(x)) = \KLD(g(x)||f(x))$.
\label{Def:betaD}
\end{definition}

The \textBD has often been referred to as the Density-Power Divergence in the statistics literature \citep{basu1998robust} where it is often parametrised as $\beta=\beta_{DPD}+1$.

\subsection{Notation}

Now we define the paper's notation in full. The focus of the paper is on different densities for $p$-dimensional observations $y\in\mathcal{Y}\subset\mathbb{R}^p$. Let $g$, $g_1$ and $g_2$ be potential data generating densities for $y$. Consider likelihood models for $y$ 
\begin{align}
\left\lbrace f(y;\theta): y\in\mathcal{Y}\subset\mathbb{R}^p, \theta\in\Theta\subseteq \mathbb{R}^{q_f}\right\rbrace, \nonumber
\end{align}
and where appropriate potential alternative likelihood model
\begin{align}
\left\lbrace h(y;\eta): y\in\mathcal{Y}\subset\mathbb{R}^p,\eta\in\mathcal{A}\subseteq \mathbb{R}^{q_h}\right\rbrace,\nonumber
\end{align}
with functions $I_f: \Theta \mapsto \mathcal{A}$ and $I_h:\mathcal{A}\mapsto \Theta$ mapping between their parameter spaces. 
%
%
The parameter of model $f(\cdot;\theta)$ minimising divergence $D$ to \DGP $g$ is defined as
\begin{align}
\theta^{D}_g &=\argmin_{\theta\in\Theta}D(g(\cdot),f(\cdot;\theta)) = \argmin_{\theta\in\Theta} \int \ell^D(y,f(\cdot;\theta))dG(y).\nonumber
\end{align}
\color{black}
The empirical loss minimser from data $y = (y_1, \ldots, y_n)\sim g$ is 
\begin{align}
    \hat{\theta}^{D}_{g, n} = \argmin_{\theta\in\Theta} \sum_{i=1}^n\ell^D(y_i, f(\cdot;\theta))\nonumber
\end{align}
\color{black}
%
The general Bayesian posterior learning about $\theta^{D}_g$ from $y \sim g$ given prior $\pi^D(\cdot)$ is
\begin{align}
\pi^D(\theta|y)&=\frac{\pi^D(\theta)\exp(-\sum_{i=1}^n\ell^D(x_i,f(\cdot;\theta)))}{\int \pi^D(\theta)\exp(-\sum_{i=1}^n\ell^D(x_i,f(\cdot;\theta)))d\theta}\nonumber.
\end{align}
The posterior predictive for exchangeable observation $y_{new} \in\mathcal{Y}$ is
\begin{align}
m^D_{f}(y_{new}|y)&=\int f(y_{new};\theta)\pi^D(\theta|y)d\theta.\nonumber
\end{align}
Throughout this section, we will use the $\cdot$ notation within divergence functions to indicate the variable that is being integrated over in the divergence, i.e. the divergence does not depend on a value for this variable. 
\color{black}
Further, define the expected and empirical Hessian matrices as
\begin{align}
    H^{D}_g(\theta) &:= \left(\frac{\partial}{\partial \theta_i \partial \theta_j}\mathbb{E}_y\left[\ell^{D}(y, f(\cdot; \theta))\right]\right)_{i,j}\label{Equ:BetaDHessian}\\
    \hat{H}^{D}_{g,n}(\theta) &:= \left(\frac{\partial}{\partial \theta_i \partial \theta_j}\frac{1}{n}\sum_{i=1}^n\ell^{D}(y_i, f(\cdot; \theta))\right)_{i,j}\nonumber.
\end{align}
\color{black}


\subsection{Technical Conditions}{\label{Sub:Conditions}}

The results of Sections 3 and 4 require the following conditions. 




\subsubsection{Stability to the likelihood model}

We require the following stochastic concentration condition of the general Bayesian posterior, which we argue below will hold given sufficient regularity of the observations and the prior specification. This condition was inspired by the Stochastic Lipschitz continuity assumption of \cite{norkin1986stochastic}.

\begin{condition}[Stochastic Concentration of the posterior for $f$ and $h$ around $g$]
For divergence $D(\cdot || \cdot)$ and likelihood models $\left\lbrace f(\cdot;\theta):\theta\in\Theta\right\rbrace$ and $\left\lbrace h(\cdot;\eta):\eta\in\mathcal{A}\right\rbrace$, define the subsets of parameters
\begin{align}
    \mathcal{S}_d^{(1)} :&= \left\lbrace \theta \in \Theta, \eta \in \mathcal{A} \textrm{ s.t. } D(g||h(\cdot|\eta)) - D(g||h(\cdot|I_f(\theta))) \leq d\right\rbrace \nonumber\\
    \mathcal{S}_d^{(2)}  :&= \left\lbrace \theta \in \Theta, \eta \in \mathcal{A} \textrm{ s.t. } D(g||f(\cdot|\theta))) - D(g||f(\cdot|I_h(\eta))) \leq d\right\rbrace, \nonumber
\end{align}
where $\mathcal{S}_d^{(1)}, \mathcal{S}_d^{(2)}\subset \Theta \times \mathcal{A}$ 
and  $I_f: \Theta \mapsto \mathcal{A}$ and $I_h: \mathcal{A}\mapsto \Theta$ are mappings between the $\Theta$ and $\mathcal{A}$. 
Then, for dataset $y\sim g(\cdot)$ with $n > 0$ and priors $\pi^D(\theta)$ and $\pi^D(\eta)$ there exists $c_1, c_2 > 0$ such that for all $d > 0$ the product posterior $\pi^D(\theta, \eta | y) = \pi^D(\theta |y)\pi^D(\eta | y)$ satisfies 
\begin{align}
   \pi^{D}(\mathcal{S}_d^{(1)} | y) &\geq  1 - \exp(-c_1d)\label{Equ:StochasticCondition_fh_S1}\\
   \pi^{D}(\mathcal{S}_d^{(2)} | y) &\geq  1 - \exp(-c_2d)\label{Equ:StochasticCondition_fh_S2}.
\end{align}
\label{Cond:StochasticPosteriorConcentration}
\end{condition}

Condition \ref{Cond:StochasticPosteriorConcentration} ensures that $n$ is large enough and $\pi^D(\theta)$ and $\pi^D(\eta)$ have sufficient prior density around $\theta^{D}_g$ and $\eta^{D}_g$ for the posterior based on the likelihoods $f$ and $h$ to have concentrated around their optimal parameter such that the \textit{a posteriori} probabilities that $h(y|I_f(\theta))$ is closer to $g$ than $h(y | \eta)$ and $f(y|I_h(\eta))$ is closer to $g$ than $f(y | \theta)$ according to divergence $D$ vanish  sufficiently quickly. 
\color{black}
This condition allows us to upper bound $\int D(g || h(\cdot | \eta))\pi^D(\eta | y)d\eta$ by $\int D(g || h(y|I_f(\theta)))\pi^D(\theta | y)d\theta$ (and equivalently $\int D(g || f(y | \theta))\pi^D(\theta | y)d\eta$ by $\int D(g || f(y|I_h(\eta)))\pi^D(\eta | y)d\eta$) and therefore compare $h(y|I_f(\theta))$ and $f(y|\theta)$ (and equivalently $f(y|I_h(\eta))$ and $h(y|\eta))$) acorss the values of their shared parameter $\theta$ (or $\eta$).
\color{black}

\color{black}
\cite{miller2021asymptotic} proved a Bernstein von-Mises theorem for generalised Bayesian posteriors \citep[see also][]{chernozhukov2003mcmc,lyddon2018generalized} which provides sufficient conditions under which they concentrate at their loss minimising parameter $\theta_g^D$ and $\eta_g^D$. Under these conditions, $\pi^{D}(\mathcal{S}_d^{(1)} | y) \overset{P}{\rightarrow} 1$ and $\pi^{D}(\mathcal{S}_d^{(2)} | y) \overset{P}{\rightarrow} 1$ as $n\rightarrow\infty$, as $D(g,f(\cdot;\theta^D_{g}))\leq D(g,f(\cdot;I_h(\eta^D_{g})))$ and vice versa by definition. Theorem \ref{Thm:MillerThm4}, proves a special case of the result of \cite{miller2021asymptotic} for the \textBD-Bayes and Corollary \ref{Cor:ReminderConvergence2} formalises how this convergence implies Condition \ref{Cond:StochasticPosteriorConcentration} holds.
\color{black}


Conditions \ref{Cond:StochasticPosteriorConcentration} (and \ref{Cond:StochasticPosteriorConcentration2} below) are the only part of these theorems where the observed data appears. So the following theorems simply require that the Bayesian updating is being done conditional on a dataset satisfying Condition \ref{Cond:StochasticPosteriorConcentration} or \ref{Cond:StochasticPosteriorConcentration2} where appropriate. 
Extensions could look at whether Condition \ref{Cond:StochasticPosteriorConcentration} and the following theorems hold in expectation under the data generating process (\DGP), however, this may require additional assumptions to be made about the \DGP that we wish to avoid.

\color{black}
Additionally, triangle-type inequalities relating the \textBD to the \TVD will require the bounding of the value of the density functions according to Condition \ref{Cond:BoundedDensities}.
\color{black}
\color{black}
\begin{condition}[Boundedness of $g$, $f$ and $h$]   
For data generating process $g(\cdot)$ and likelihood models $\left\lbrace f(\cdot;\theta):\theta\in\Theta\right\rbrace$ and $\left\lbrace h(\cdot;\eta):\eta\in\mathcal{A}\right\rbrace$ 
there exists $0<M<\infty$ such that
\begin{equation}
\esssup_y g(y)\leq M,\quad \esssup_y f(y;\theta) \leq M~ \forall \theta\in\Theta \textrm{ and } \esssup_y h(y, \eta) ~\forall \eta\in\mathcal{A}.\nonumber
\end{equation}
\label{Cond:BoundedDensities}
\end{condition}
\color{black}
  
Given base measure $\mu$ - assumed to be the Lebesque measure for continuous random variables and the counting measures for discrete random variables - $\esssup_y f(y)=M$ if the set defined by $f^{-1}(M,\infty)$ has measure 0, i.e. $\mu\left(f^{-1}(M,\infty)\right)=0$. For discrete random variables it is always the case that $M \leq 1$ and $M$ can be bounded for continuous random variables such as a Gaussian or \Student by lower bounding the model's scale parameter by some reasonable value.

\subsubsection{Stability to the \DGP}

Conditions \ref{Cond:StochasticPosteriorConcentration2} and \ref{Cond:BoundedDensities2} are required for the results of Section 4 and are analogous to Conditions \ref{Cond:StochasticPosteriorConcentration} and \ref{Cond:BoundedDensities}  introduced in the previous section.

\begin{condition}[Stochastic Concentration of the posterior for $f$ around $g_1$ and $g_2$]
For divergence $D(\cdot || \cdot)$ and likelihood model $\left\lbrace f(\cdot;\theta):\theta\in\Theta\right\rbrace$, define subsets of $\Theta \times \Theta$
    \begin{align}
    \mathcal{S}_d^{(1)} :&= \left\lbrace \theta_1, \theta_2 \in \Theta \textrm{ s.t. } D(g_2||f(\cdot;\theta_2)) - D(g_2||f(\cdot;\theta_1)) \leq d\right\rbrace\nonumber \\
    \mathcal{S}_d^{(2)}  :&= \left\lbrace \theta_1, \theta_2 \in \Theta \textrm{ s.t. } D(g_1||f(\cdot;\theta_1)) - D(g_1||f(\cdot;\theta_2)) \leq d\right\rbrace. \nonumber
\end{align}
Then, for datasets $y_{1:n_1}\sim g_1(\cdot)$ and $y^{\prime}_{1:n_2}\sim g_2(\cdot)$ with $n_1, n_2 > 0$ and prior $\pi^D(\theta)$  there exists $c_{\mathcal{S}^{(1)}}, c_{\mathcal{S}^{(2)}} > 0$ such that for all $d > 0$ the product posterior $\pi^D(\theta_1, \theta_2| y_1, y_2) = \pi^D(\theta_1 | y_1)\pi^D(\theta_2 | y_2)$ satisfies
\begin{align}
   \pi^{D}(\mathcal{S}_d^{(1)} | y_1, y_2) &\geq  1 - \exp(-c_{\mathcal{S}^{(1)}}d), \label{Equ:StochasticCondition_g1g2_S1} \\ 
   \pi^{D}(\mathcal{S}_d^{(2)} | y_1, y_2) &\geq  1 - \exp(-c_{\mathcal{S}^{(2)}}d).\label{Equ:StochasticCondition_g1g2_S2}
\end{align}
\label{Cond:StochasticPosteriorConcentration2}
\end{condition}

\color{black}
\begin{condition}[Boundedness of $g_1$, $g_2$ and $f$]   
For data generating processes $g_1(\cdot)$ and $g_2(\cdot)$ and likelihood model $\left\lbrace f(\cdot;\theta):\theta\in\Theta\right\rbrace$ there exists $0<M<\infty$ such that
\begin{equation}
\esssup_y g_1(y)\leq M, \quad \esssup_y g_2(y)\leq M, \textrm{ and }\esssup_y f(y; \theta)\leq M~\forall\theta\in\Theta.\nonumber
\end{equation}
\label{Cond:BoundedDensities2}
\end{condition}
\color{black}



\color{black}
While not necessary for any of our results, the Bernstein-von Mises theorem for generalised posteriors \citep[Theorem 4;][]{miller2021asymptotic} can be applied to the \textBD-Bayes and helps to interpret our results. Here, we state the required Condition \ref{Cond:MillerThm4} before later stating the result in Theorem \ref{Thm:MillerThm4}.

\begin{condition}[Assumptions of Theorem 4 of \cite{miller2021asymptotic} for the \textBD]
Fix $\theta_g^{(\beta)} \in \mathbb{R}^p$ and let prior $\pi(\theta)$ is continuous at $\theta_g^{(\beta)}$ with $\pi(\theta_g^{(\beta)}) > 0$. Let $L^{(\beta)}_n:\mathbb{R}^p \rightarrow \mathbb{R}$ with $L^{(\beta)}_n(\theta) = \frac{1}{n}\sum_{i=1}^n\ell^{(\beta)}(y_i, f(\cdot; \theta)$ for $n \in \mathbb{N}$ and assume:
\begin{enumerate}
\item[(1)] $L^{(\beta)}_{g,n}$ can be represented as 
\begin{align}
L^{(\beta)}_{g,n}(\theta) = L^{(\beta)}_{g,n}(\hat{\theta}_{g,n}^{(\beta)}) + \frac{1}{2}(\theta-\hat{\theta}_{g,n}^{(\beta)})^T\hat{H}^{(\beta)}_{g,n}(\theta-\hat{\theta}_{g,n}^{(\beta)}) + r^{(\beta)}_{g,n}(\theta - \hat{\theta}_{g,n}^{(\beta)})\nonumber
\end{align}
where $\hat{\theta}_{g,n}^{(\beta)} \in \mathbb{R}^p$ and $\hat{\theta}_{g,n}^{(\beta)}\rightarrow \theta_g^{(\beta)}$, with $\hat{H}^{\ell^{(\beta)}}_{g,n}\rightarrow H^{(\beta)}_g$ for positive definite $H^{(\beta)}_g$, and $r^{(\beta)}_{g,n} : \mathbb{R}^p \rightarrow \mathbb{R}$ has the following property: there exist $\epsilon_0, c_0 > 0$ such that for all $n$ sufficiently large, for all $x \in B_{\epsilon_0}(0)$, we have $|r^{(\beta)}_{g,n}(x)| \leq c_0|x|^3$.
\item[(2)] For any $\epsilon > 0$, $\lim\inf_n\inf_{\theta\in B_{\epsilon}(\hat{\theta}_{g,n}^{(\beta)})^c}(L^{(\beta)}_{g,n}(\theta) - L^{(\beta)}_{g,n}(\hat{\theta}_{g,n}^{\ell^{(\beta)}})) > 0$, where $B_r(x_0) = \{x\in\mathbb{R}^D:|x-x_0| < r\}$
\end{enumerate}
\label{Cond:MillerThm4}
\end{condition}
Condition \ref{Cond:MillerThm4} (1) requires that the \textBD loss can be approximated by a quadratic form and  (2) requires that as $n$ grows the \textBD loss is uniquely minimised at $\hat{\theta}_n^{\ell^{(\beta)}}$. \cite{miller2021asymptotic} sought general conditions and did not condition on $L^{(\beta)}_{g,n}$ being differentiable. The \textBD loss applied to standard probability models is generally differentiable and in this case Condition \ref{Cond:MillerThm4} (1) is immediately provided by Taylor's Theorem providing the 3rd order partial derivatives of $L^{(\beta)}_{g,n}$ are continuous in $B_{\epsilon_0}(\hat{\theta}_{g,n}^{(\beta)})$ and bounded.

\color{black}

\subsection{Proofs: Stability to the Model}
\label{ssec:proofs}

Before we prove Theorem 1, Theorem 2 and Lemma 1 we first introduce some useful Lemmas that simplify their proofs.

\subsubsection{Useful Lemmas for proving Theorems 1 and 2}

Lemma \ref{Lem:TVDsimp} establishes a convenient representation for the \TVD.

\begin{lemma}[A simplification of the \TVD]
The following relationship holds for the \TVD between two densities $f$ and $h$. 
\begin{align}
\TVD(f,h)&=\int_{A^{+}}\left(h(y)-f(y)\right)dy=\int_{A^{-}}\left(f(y)-h(y)\right)dy.\nonumber
\end{align}
where $A^{+}:=\left\lbrace y: h(y)>f(y)\right\rbrace$ and $A^{-}:=\left\lbrace y: f(y)>h(y)\right\rbrace$.
\label{Lem:TVDsimp}
\end{lemma}

\begin{proof}
Firstly, 
by definition
\begin{align}
\TVD(f,h)&=\frac{1}{2}\int\left|h(y)-f(y)\right|dy\nonumber\\
&=\frac{1}{2}\int_{A^{+}}\left(h(y)-f(y)\right)dy + \frac{1}{2}\int_{A^{-}}\left(f(y)-h(y)\right)dy\nonumber.
\end{align}
Next, consider $L_{f,h}: \mathcal{Y}\rightarrow\mathbb{R}$ with $L_{f,h}(y) := \min(f(y), h(y))$ as the lower of the two probability densities for every $y$. Given that both $f$ and $h$ are probability densities and thus integrate to 1 we have that 
\begin{align}
\int_{A^{+}} (h(y)-f(y)) dy &= 1- \int L_{f,h}(y)dy\nonumber\\
\int_{A^{-}} (f(y)-h(y)) dy &= 1- \int L_{f,h}(y)dy.\nonumber
\end{align}
The two right-hand sides are identical and therefore the two left-hand sides must be equal. As a result, 
\begin{align}
\TVD(f,h)&=\frac{1}{2}\int_{A^{+}}\left(h(y)-f(y)\right)dy + \frac{1}{2}\int_{A^{-}}\left(f(y)-h(y)\right)dy\nonumber.\\
&=\int_{A^{+}}\left(h(y)-f(y)\right)dy\nonumber\\
&=\int_{A^{-}}\left(f(y)-h(y)\right)dy,\nonumber
\end{align}
proving the result.
\end{proof}
 

Lemma \ref{Lem:betaDivTVDTriangle} establishes a triangle-type inequality relating the \textBD and the \TVD. 
Triangle-type inequalities fit naturally with Section 3's requirements for stability. If two models are close, then they ought to provide similar approximations to a third distribution, the \DGP. The \textBD does not strictly satisfy the triangle inequality. However, we can prove the following results connecting the \TVD and the \textBD in a triangle-type inequality. The result relies on $1\leq \beta\leq 2$, which places the \textBD in between the \KLD at $\beta=1$ and the $L_2$-distance $D_{B}^{(2)}(g||f)=\frac{1}{2}\int\left(f-g\right)^2$. We are yet to come across scenarios where setting $\beta$ outside this range is appropriate from a practical viewpoint \citep[see e.g.][]{jewson2018principles,knoblauch2018doubly}. 

\begin{lemma}[A triangle inequality relating the \textBD and the \TVD]
For densities $f$, $h$ and $g$ with the property that there exists $M<\infty$ satisfying Condition \ref{Cond:BoundedDensities} and $1< \beta\leq 2$ we have that 
\begin{align}
\left|\BD(g||h) - \BD(g||f)\right| &\leq\frac{M^{\beta-1}(3\beta - 2)}{\beta(\beta-1)}\TVD(h,f).\nonumber
\end{align} 
\label{Lem:betaDivTVDTriangle}
\end{lemma}

\begin{proof}
By the definition of the \textBD, we can rearrange 
\begin{align}
\BD&(g||h)\nonumber\\
=&\BD(g||f)+\left(\int\left[ \frac{1}{\beta}h(y)^{\beta}-\frac{1}{\beta}f(y)^{\beta}-\frac{1}{\beta-1}g(y)h(y)^{\beta-1}+\frac{1}{\beta-1}g(y) f(x)^{\beta-1}\right]dy\right)\nonumber\\
=&\BD(g||f)+\left(\frac{1}{\beta}\int \left(h(y)^{\beta}-f(y)^{\beta}\right) dy +\frac{1}{\beta-1}\int g(y)\left(f(y)^{\beta-1}-h(y)^{\beta-1}\right)dy \right)\nonumber
\end{align}
As in Lemma \ref{Lem:TVDsimp}, define $A^{+}:=\left\lbrace y: h(y)>f(y)\right\rbrace$ and $A^{-}:=\left\lbrace y: f(y)>h(y)\right\rbrace$. 
Now by the monotonicity of the function $y^{\beta}$ and $y^{\beta - 1}$ when $1\leq \beta\leq 2$ we have that 
\begin{align}
\int_{A^{-}} h(y)^{\beta}-f(y)^{\beta} dy&<0\nonumber\\
\int_{A^{+}} g(y)\left(f(y)^{\beta-1}-h(y)^{\beta-1}\right)dy &<0\nonumber
\end{align}
therefore removing these two terms provides an upper bound 
\begin{align}
\BD&(g||h)\nonumber\\
=&\BD(g||f)+\frac{1}{\beta}\int \left(h(y)^{\beta}-f(y)^{\beta}\right) dy +\frac{1}{\beta-1}\int g(y)\left(f(y)^{\beta-1}-h(y)^{\beta-1}\right)dy\nonumber \\
\leq& \BD(g||f)+\frac{1}{\beta}\int_{A^{+}} \left(h(y)^{\beta}-f(y)^{\beta} \right)dy +\frac{1}{\beta-1}\int_{A^{-}} g(y)\left(f(y)^{\beta-1}-h(y)^{\beta-1}\right)dy.\nonumber
\end{align}
Next, adding and subtracting $\frac{1}{\beta}h(y)f(y)^{\beta-1}$ provides
\begin{align}
\BD&(g||h)\nonumber\\
\leq& \BD(g||f)+\frac{1}{\beta}\int_{A^{+}} \left(h(y)^{\beta}-f(y)^{\beta}\right) dy +\frac{1}{\beta-1}\int_{A^{-}} g(y)\left(f(y)^{\beta-1}-h(y)^{\beta-1}\right)dy\nonumber.\\
=& \BD(g||f)+\frac{1}{\beta}\int_{A^{+}} \left(h(y)^{\beta}-h(y)f(y)^{\beta - 1}+h(y)f(y)^{\beta - 1}-f(y)^{\beta}\right) dy\nonumber\\
&+\frac{1}{\beta-1}\int_{A^{-}} g(y)\left(f(y)^{\beta-1}-h(y)^{\beta-1}\right)dy\nonumber.\\
=& \BD(g||f)+\frac{1}{\beta}\int_{A^{+}} h(y)\left(h(y)^{\beta - 1}-f(y)^{\beta-1}\right)dy+\frac{1}{\beta}\int_{A^{+}}f(y)^{\beta-1}\left(h(y)-f(y)\right) dy\nonumber\\
&+\frac{1}{\beta-1}\int_{A^{-}} g(y)\left(f(y)^{\beta-1}-h(y)^{\beta-1}\right)dy\nonumber.\\
=& \BD(g||f)+\frac{1}{\beta}\int_{A^{+}} h(y)^{\beta}\left(1-\frac{f(y)^{\beta - 1}}{h(y)^{\beta - 1}}\right)dy+\frac{1}{\beta}\int_{A^{+}}f(y)^{\beta-1}\left(h(y)-f(y)\right) dy\nonumber\\
&+\frac{1}{\beta-1}\int_{A^{-}} g(y)f(y)^{\beta-1}\left(1-\frac{h(y)^{\beta-1}}{f(y)^{\beta-1}}\right)dy\nonumber.
\end{align}
%
%
Now on $A^{+}$ $h(y)>f(y)$ and so $\left(\frac{f(y)}{h(y)}\right)^{\beta-1}>\frac{f(y)}{h(y)}$ for $1\leq\beta\leq2$ so 
\begin{equation}
\left(1-\frac{f(y)^{\beta-1}}{h(y)^{\beta-1}}\right)\leq \left(1-\frac{f(y)}{h(y)}\right)\nonumber
\end{equation}
with the exact same logic holding in reverse on $A^{-}$. We can use this to show that
\begin{align}
\BD&(g||h) \leq \BD(g||f)+\frac{1}{\beta}\int_{A^{+}} h(y)^{\beta}\left(1-\frac{f(y)^{\beta - 1}}{h(y)^{\beta - 1}}\right)dy+\frac{1}{\beta}\int_{A^{+}}f(y)^{\beta-1}\left(h(y)-f(y)\right) dy\nonumber\\
&+\frac{1}{\beta-1}\int_{A^{-}} g(y)f(y)^{\beta-1}\left(1-\frac{h(y)^{\beta-1}}{f(y)^{\beta-1}}\right)dy\nonumber\\
\leq& \BD(g||f)+\frac{1}{\beta}\int_{A^{+}} h(y)^{\beta}\left(1-\frac{f(y)}{h(y)}\right)dy+\frac{1}{\beta}\int_{A^{+}}f(y)^{\beta-1}\left(h(y)-f(y)\right) dy\nonumber\\
&+\frac{1}{\beta-1}\int_{A^{-}} g(y)f(y)^{\beta-1}\left(1-\frac{h(y)}{f(y)}\right)dy\nonumber\\
=& \BD(g||f)+\frac{1}{\beta}\int_{A^{+}} h(y)^{\beta - 1}\left(h(y)-f(y)\right)dy+\frac{1}{\beta}\int_{A^{+}}f(y)^{\beta-1}\left(h(y)-f(y)\right) dy\nonumber\\
&+\frac{1}{\beta-1}\int_{A^{-}} g(y)f(y)^{\beta-2}\left(f(y)-h(y)\right)dy\nonumber.
\end{align}
We now use the fact that we defined $\max\left\lbrace \esssup f, \esssup h, \esssup g\right\rbrace\leq M<\infty$ and Lemma \ref{Lem:TVDsimp} to leave
\begin{align}
\BD&(g||h) = \BD(g||f)+\frac{1}{\beta}\int_{A^{+}} h(y)^{\beta - 1}\left(h(y)-f(y)\right)dy+\frac{1}{\beta}\int_{A^{+}}f(y)^{\beta-1}\left(h(y)-f(y)\right) dy\nonumber\\
&+\frac{1}{\beta-1}\int_{A^{-}} g(y)f(y)^{\beta-2}\left(f(y)-h(y)\right)dy\nonumber\\
\leq& \BD(g||f)+\frac{M^{\beta-1}}{\beta}\int_{A^{+}}\left(h(y)-f(y)\right)dy+\frac{M^{\beta-1}}{\beta}\int_{A^{+}}\left(h(y)-f(y)\right) dy\nonumber\\
&+\frac{M^{\beta-1}}{\beta-1}\int_{A^{-}} \left(f(y)-h(y)\right)dy\nonumber\\
=& \BD(g||f)+2\frac{M^{\beta-1}}{\beta}\TVD(h,f) +\frac{M^{\beta-1}}{\beta-1}\TVD(h,f)\nonumber\\
=&\BD(g||f)+\frac{M^{\beta - 1}(3\beta - 2)}{\beta(\beta - 1)}\TVD(h,f).\nonumber
\end{align}
\end{proof}

Lemma \ref{Lem:betaDivConvexity} proves the convexity of the $\BD(g, f)$ in both $g$ and $f$

\begin{lemma}[The convexity of the \textBD]
The \textBD between two densities $g(y)$ and $f(y)$ is convex in both densities for $1< \beta\leq 2$, when fixing the other. That is to say that for $\lambda\in[0,1]$
\begin{align}
\BD(\lambda g_1+(1-\lambda) g_2, f)&\leq \lambda \BD(g_1,f)+(1-\lambda)\BD(g_2,f) \textrm{ for all } f\nonumber\\
\BD(g, \lambda f_1+(1-\lambda) f_2)&\leq \lambda \BD(g,f_1)+(1-\lambda)\BD(g,f_2)\textrm{ for all } g\nonumber
\end{align}
for $1< \beta\leq2$.
\label{Lem:betaDivConvexity}
\end{lemma}

\begin{proof}


First, we fix $f$ and look at convexity in the function $g$. let $\lambda\in [0,1]$. The function $x^{p}$ for $x\geq 0$ and $p>1$ is convex and thus satisfies 
\begin{align}
\left(\lambda x_1+(1-\lambda)x_2\right)^p\leq \lambda x_1^p+(1-\lambda)x_2^p \nonumber
\end{align}
therefore we have that provided $\BD(g_1||f)<\infty$ and $\BD(g_2||f)<\infty$ 
\begin{align}
\BD&(\lambda g_1+(1-\lambda) g_2||f)\nonumber\\
=&\int \frac{1}{\beta(\beta-1)}\left(\lambda g_1+(1-\lambda)g_2\right)^{\beta}+\frac{1}{\beta}f^{\beta}-\frac{1}{\beta-1}\left(\lambda g_1+(1-\lambda)g_2\right)f^{\beta-1}d\mu\nonumber\\
\leq&\int \frac{1}{\beta(\beta-1)}\left(\lambda g_1^{\beta}+(1-\lambda)g_2^{\beta}\right)+\frac{1}{\beta}f^{\beta}-\frac{1}{\beta-1}\left(\lambda g_1+(1-\lambda)g_2\right)f^{\beta-1}d\mu\nonumber\\
=&\lambda\BD(g_1||f)+(1-\lambda)\BD(g_2||f)\nonumber.
\end{align}
Next, we fix $g$ and look at the convexity in $f$. Similarly to above we know that when $x\geq 0$ and $1\leq p\leq 2$ that $\frac{1}{p}x^{p}$ and $-\frac{1}{p-1}x^{p-1}$ are both convex in $y$. We therefore have that provided $\BD(g||f_1)<\infty$ and $\BD(g||f_2)<\infty$
\begin{align}
\BD(g||&\lambda f_1+(1-\lambda) f_2)\nonumber\\
=&\int \frac{1}{\beta(\beta-1)}g^{\beta}+\frac{1}{\beta}\left(\lambda f_1+(1-\lambda)f_2\right)^{\beta}-\frac{1}{\beta-1}g \left(\lambda f_1+(1-\lambda)f_2\right)^{\beta-1}d\mu\nonumber\\
\leq&\int \frac{1}{\beta(\beta-1)}g^{\beta}+\frac{1}{\beta}\left(\lambda f_1^{\beta}+(1-\lambda)f_2^{\beta}\right)-\frac{1}{\beta-1}g \left(\lambda f_1^{\beta-1}+(1-\lambda)f_2^{\beta-1}\right)d\mu\nonumber\\
=&\lambda\BD(g||f_1)+(1-\lambda)\BD(g||f_2).\nonumber
\end{align}
\end{proof}

Lemma \ref{Lem:betaDivTriangle} introduces a useful the ``three-point property'' \citep{cichocki2010families} associated with the \textBD




\begin{lemma}[Three-point property of the \textBD]
The following relationship for the \textBD holds for densities $g$, $f$ and $h$
\begin{align}
\BD(f||h)&=\BD(g||h)-\BD(g||f)+R(g||f||h)\nonumber
\end{align}
where
\begin{align}
R(g||f||h)&=\frac{1}{\beta-1}\int (g-f)\left(h^{\beta-1}-f^{\beta-1}\right)d\mu\label{Equ:Rghf}
\end{align}
\label{Lem:betaDivTriangle}
\end{lemma}


\begin{proof}
Following the definition of the \textBD  \eqref{Def:betaD}
\begin{align}
\BD&(g||f) + \BD(f||h)\nonumber\\
=& \int\frac{1}{\beta(\beta-1)} g^{\beta} + \frac{1}{\beta} f^{\beta} -\frac{1}{\beta-1}gf^{\beta-1}d\mu + \int\frac{1}{\beta(\beta-1)} f^{\beta} + \frac{1}{\beta} h^{\beta} -\frac{1}{\beta-1}fh^{\beta-1}d\mu\nonumber\\
=&\int\frac{1}{\beta(\beta-1)} g^{\beta} + \frac{1}{\beta} h^{\beta} + \frac{1}{\beta - 1}f^{\beta} -\frac{1}{\beta-1}gf^{\beta-1}- \frac{1}{\beta-1}fh^{\beta-1}d\mu\nonumber\\
=&\int\frac{1}{\beta(\beta-1)} g^{\beta} + \frac{1}{\beta} h^{\beta} - \frac{1}{\beta - 1}gh^{\beta - 1} + \frac{1}{\beta - 1}gh^{\beta - 1} + \frac{1}{\beta - 1}ff^{\beta-1} -\frac{1}{\beta-1}gf^{\beta-1}d\mu\nonumber\\
&- \frac{1}{\beta-1}fh^{\beta-1}d\mu\nonumber\\
=&\BD(g||h) + \frac{1}{\beta-1}\int(g-f)\left(h^{\beta - 1} - f^{\beta-1}\right)d\mu\nonumber
\end{align}
\end{proof}

Lemma \ref{Lem:remainder_TVD_bound} provides a useful bound for the interpretation of the remainder term from Lemma \ref{Lem:betaDivTriangle}.

\begin{lemma}[A bound on $R(g||f||h)$ from Lemma \ref{Lem:betaDivTriangle}]
For densities $f$, $h$ and $g$ with the property that there exists $M<\infty$ satisfying Condition \ref{Cond:BoundedDensities} and $1< \beta\leq 2$, the remainder term from Lemma \ref{Lem:betaDivTriangle} can be bounded as 
\begin{align}
R(g||f||h)&\leq 2\frac{M^{\beta - 1}}{\beta-1}\TVD(g, f)\nonumber
\end{align}
\label{Lem:remainder_TVD_bound}
\end{lemma}

\begin{proof}
Define $A^{+}_f := \left\lbrace y : g(y) \geq f(y))\right\rbrace$ and $A^{-}_f:= \left\lbrace y :  g(y) \leq f(y))\right\rbrace$ as
\begin{align}
    R(g||f||h)&=\frac{1}{\beta-1}\int (g-f)\left(h^{\beta-1}-f^{\beta-1}\right)d\mu\nonumber\\
    &= \frac{1}{\beta-1}\int h^{\beta-1}\left(g-f\right)d\mu + \frac{1}{\beta-1}\int f^{\beta - 1}\left(f-g\right)d\mu\nonumber\\
    &\leq \frac{1}{\beta-1}\int_{A^{+}} h^{\beta-1}\left(g-f\right)d\mu + \frac{1}{\beta-1}\int_{A^{-}} f^{\beta - 1}\left(f-g\right)d\mu\nonumber\\
    &\leq \frac{M^{\beta - 1}}{\beta-1}\int_{A^{+}} \left(g-f\right)d\mu + \frac{M^{\beta - 1}}{\beta-1}\int_{A^{-}} \left(f-g\right)d\mu\nonumber\\
    &\leq 2\frac{M^{\beta - 1}}{\beta-1}\TVD(g, f)\nonumber
\end{align}
by Condition \ref{Cond:BoundedDensities} and Lemma \ref{Lem:TVDsimp}.
\end{proof}

Lemma \ref{Lem:BoundingBetaDTVD} shows that the \textBD can be bounded above by the \TVD, which is useful when interpreting the bound in Theorem 1.

\begin{lemma}
For densities $f$, $h$ and $g$ with the property that there exists $M<\infty$ satisfying Condition \ref{Cond:BoundedDensities} and $1< \beta\leq 2$ we have that

\begin{equation}
\BD(g||f)\leq \left( \frac{M^{\beta-1}}{\beta-1 }\right) \TVD(g,f).\nonumber
\end{equation}
{\label{Lem:BoundingBetaDTVD}}
\end{lemma}

\begin{proof}
Firstly, define $A^{-}=\left\lbrace y : g(y)<f(y)\right\rbrace$ and $A^{+}=\left\lbrace y : g(y)\geq f(y)\right\rbrace$ and note on $A_{+}$ that $(f(y)-g(y))<0$ and on $A_{-}$ that $g(y)<f(y)\Rightarrow g^{\beta-1}(y)<f^{\beta-1}(y)$ for $1\leq \beta\leq 2$. The \textBD can, then, be rearranged as
\begin{align}
&\BD(g||f) \nonumber\\
=&\frac{1}{\beta} \int f(y)^{\beta-1 }\left(f(y)-g(y)\right) dy + \frac{1}{\beta(\beta-1)}\int \left( g(y)^{\beta-1 }-f(y)^{\beta-1 }\right) g(y) dy \nonumber\\
\leq &\frac{1}{\beta} \int_{A_{-}}f(y)^{\beta-1 }\left( f(y)-g(y)\right) dy +\frac{1}{\beta(\beta-1)}\int_{A_{+}}\left(
g(y)^{\beta-1 }-f(y)^{\beta-1 }\right) g(y) dy.\nonumber
\end{align}
Since $f\leq M$ by Condition \ref{Cond:BoundedDensities} and using Lemma \ref{Lem:TVDsimp} we can write
\begin{equation*}
\int_{A_{-}}f(y)^{\beta-1 }\left( f(y)-g(y)\right)dy \leq M^{\beta-1 }\int_{A_{-}}\left(
f(y)-g(y)\right)dy = M^{\beta-1 }\TVD(f,g).
\end{equation*}
Further, on $A_{+}$ we have that $g(y)>f(y)$ which implies that $\frac{f(y)}{g(y)}<1$ and that when $1<\beta <2$,  $\frac{f(y)}{g(y)}^{\beta-1}>\frac{f(y)}{g(y)}$ so
\begin{align}
\int_{A_{+}}&\left( g(y)^{\beta-1 }-f(y)^{\beta-1 }\right) g(y) dy=\int_{A_{+}}g(y)^{\beta-1 }\left( 1-\left( \frac{f(y)}{g(y)}\right) ^{\beta-1 }\right) g(y)dy\nonumber\\
\leq& \int_{A_{+}}g(y)^{\beta-1 }\left( 1- \frac{f(y)}{g(y)}\right) g(y)dy\nonumber
\\
=& \int_{A_{+}}g(y)^{\beta-1 }\left( g(y)-f(y)\right)dy\nonumber \\
\leq& M^{\beta-1 }\TVD(f,g),\nonumber
\end{align}
since $g\leq M$  by Condition \ref{Cond:BoundedDensities}. Combining the two bounds leaves  
\begin{equation*}
\BD(g||f)\leq \frac{M^{\beta-1}}{\beta}  \TVD(f,g)+\frac{M^{\beta-1}}{\beta(\beta-1)}\TVD(f,g),
\end{equation*}%
which proves the theorem.
\end{proof}

The implications of Lemma \ref{Lem:BoundingBetaDTVD} are as follows. Provided $\left(\frac{M^{\beta-1}}{\beta-1}\right)$ does not get too small, we can be confident that any value of $\theta$ such that $f(\cdot;\theta)$ is close to the data generating density $g$ in terms of \TVD, will receive high posterior mass under an update targeting the \textBD.


Lastly, Lemma \ref{Lem:StochasticCDFResult} provides a convenient result for the deployment of Conditions \ref{Cond:StochasticPosteriorConcentration} (and \ref{Cond:StochasticPosteriorConcentration2} below). This result was inspired by part of the proof of Theorem 7 of \cite{dimitrakakis2017differential} (p27).
\begin{lemma}[Stochastic Concentration \citep{norkin1986stochastic, dimitrakakis2017differential}]
If random variable $\omega \in \Omega \subset\mathbb{R}$ is distributed according to $\pi$ and there exists $c > 0$ such that for all $t > t_0$
\begin{align}
    F_{\omega}(t) = \pi(\left\{\omega \leq t\right\}) \geq 1 - \exp(- c( t - t_0)),\label{Equ:concentration}
\end{align}
then 
\begin{align}
    \int \omega \pi(\omega)d\omega \leq t_0 + \frac{1}{c}\nonumber
\end{align}
\label{Lem:StochasticCDFResult}
\end{lemma}


\begin{proof}
We can write the expectation of $\omega$ in terms of its cumulative distribution function (\CDF) as 
\begin{align}
    \int \omega \pi(\omega)d\omega &= \int_0^{\infty}( 1 - F_{\omega}(t))dt - \int^0_{-\infty}F_{\omega}(t)dt  \nonumber\\
    &\leq \int_0^{\infty}( 1 - F_{\omega}(t))dt   \nonumber\\
    &= \int_0^{t_0}( 1 - F_{\omega}(t))dt  + \int_{t_0}^{\infty}( 1 - F_{\omega}(t))dt \nonumber\\
    &\leq \int_0^{t_0}1dt + \int_{t_0}^{\infty}( 1 - F_{\omega}(t))dt  \nonumber\\
    &= t_0 + \int_{t_0}^{\infty}( 1 - F_{\omega}(t))dt.  \nonumber
\end{align}
Then, invoking \eqref{Equ:concentration} leaves
\begin{align}
    \int \omega \pi(\omega)d\omega &\leq t_0 + \int_{t_0}^{\infty}( 1 - F_{\omega}(t))dt  \nonumber\\
    &\leq t_0 + \int_{t_0}^{\infty}\exp(- c( t - t_0))dt\nonumber\\
    &\leq t_0 + \int_{0}^{\infty}\exp(- c t)dt\nonumber\\
    &=t_0 + \frac{1}{c}\nonumber
\end{align}
as required.
\end{proof}

\color{black}

Lemma \ref{Lem:BoundingBetaD} shows that the \textBD between any two distributions satisfying Condition \ref{Cond:BoundedDensities} is bounded and that its integrand is also bounded.

\begin{lemma}
For densities $f$, $h$ and $g$ with the property that there exists $M<\infty$ satisfying Condition \ref{Cond:BoundedDensities} and $1< \beta$ we have that the \textBD is bounded from above as

\begin{equation}
\BD(g||f)\leq \left( \frac{M^{\beta-1}}{\beta-1 }\right).\nonumber
\end{equation}
Further, its integrand is bounded in absolute value
\begin{equation}
    \left|\frac{1}{\beta(\beta-1)}g(y)^{\beta} + \frac{1}{\beta}f(y)^{\beta} - \frac{1}{\beta-1} g(y)f(y)^{\beta-1}\right| \leq \frac{2M^{\beta}}{\beta-1}. \nonumber
\end{equation}
{\label{Lem:BoundingBetaD}}
\end{lemma}

\begin{proof}
To prove that the \textBD is bounded we use Condition \ref{Cond:BoundedDensities} to show
\begin{align}
\BD(g||f) &= \frac{1}{\beta(\beta-1)}\int g(y)^{\beta}dy+\frac{1}{\beta}\int f(y)^{\beta}dy -\frac{1}{\beta-1}\int g(y)f(y)^{\beta-1}dy\nonumber\\
& \leq \frac{1}{\beta(\beta-1)}\int g(y)^{\beta}dy+\frac{1}{\beta}\int f(y)^{\beta}dy\nonumber\\
& \leq \frac{1}{\beta(\beta-1)}\int M^{\beta-1}g(y)dy+\frac{1}{\beta}\int M^{\beta-1}f(y)dy\nonumber\\
& \leq \frac{M^{\beta-1}}{\beta(\beta-1)}+\frac{M^{\beta-1}}{\beta}\nonumber\\
& = \frac{M^{\beta-1}}{(\beta-1)}.\nonumber
\end{align}
Further, bounded its integrand is
\begin{align}
    &\left|\frac{1}{\beta(\beta-1)}g(y)^{\beta} + \frac{1}{\beta}f(y)^{\beta} - \frac{1}{\beta-1} g(y)f(y)^{\beta-1}\right|\nonumber\\
    \leq& \frac{1}{\beta(\beta-1)}g(y)^{\beta} + \frac{1}{\beta}f(y)^{\beta} + \frac{1}{\beta-1} g(y)f(y)^{\beta-1}\nonumber\\
    \leq& \frac{1}{\beta(\beta-1)}M^{\beta} + \frac{1}{\beta}M^{\beta} + \frac{1}{\beta-1} M^{\beta}\nonumber\\
    \leq& \frac{2M^{\beta}}{\beta-1}.\nonumber
\end{align}
\end{proof}

Theorem \ref{Thm:MillerThm4} establishes the concentration of the \textBD-Bayes posterior around the \textBD minimising parameter $\theta_g^{(\beta)}$ and posterior asymptotic normality. This result directly follows from Theorem 4  of \cite{miller2021asymptotic} and this result for the \textBD previously appeared in \cite{jewson2024differentially}.

\begin{theorem}[Theorem 4 of \cite{miller2021asymptotic} for the \textBD]
Assume Condition \ref{Cond:MillerThm4}. Define $\pi^{(\beta)}_{n}(\theta) := \pi^{(\beta)}(\theta|y = \{y_1, \dots, y_n\})$ and let $\{y_i\}_{i=1}^n \sim g$. Then
\begin{enumerate}
    \item[(i)]  The \textBD posterior concentrates on $\theta^{(\beta)}_g$
    \begin{equation}
        \int_{B_{\varepsilon}(\theta^{(\beta)}_g)} \pi^{(\beta)}_{n}(\theta)d\theta \underset{n\rightarrow \infty}{\longrightarrow} 1, \forall \varepsilon > 0 \label{Equ:betaD_concentration}
        \end{equation}
        where $B_r(x_0) = \{x \in \mathbb{R}^p: |x - x_0| < r\}$
        \item[(ii)] The \textBD posterior is asymptotically Gaussian
        \begin{align}
    \int \left|\tilde{\pi}_{n}^{(\beta)}(\phi) - \mathcal{N}_p\left(\phi ; 0, (H^{(\beta)}_g)^{-1}\right)\right|d\phi \underset{n\rightarrow\infty}{\longrightarrow} 0\label{Equ:betaD_asymptoticNormality}
\end{align}
where $\tilde{\pi}_{n}^{(\beta)}$ denotes the density of $\sqrt{n}(\tilde{\theta} - \hat{\theta}^{(\beta)}_{g, n})$ when $\tilde{\theta}\sim \pi^{(\beta)}_{n}$, $\mathcal{N}_p(x; \mu, \Sigma)$ denotes the $p$-dimensional multivariate Gaussian distribution with mean vector  $\mu$ and covariance matrix $\Sigma$
\end{enumerate}
\label{Thm:MillerThm4}
\end{theorem}

\begin{proof}
    The result is proved as a direct application of Theorem 4 of \cite{miller2021asymptotic} with $f_n(\theta) = \frac{1}{n}\sum_{i=1}^n\ell^{(\beta)}(D_i, f(\cdot; \theta))$ being the \textBD-loss function.
\end{proof}

Theorem \ref{Thm:MillerThm4} shows that the \textBD posterior concentrates on $\theta^{(\beta)}_g$ and converges to a Gaussian distribution centred around the $\beta$D minimising parameter $\hat{\theta}^{(\beta)}_{g, n}$ in total variation distance. 

\begin{corollary}[Posterior Predictive Convergence I]
Assume Conditions \ref{Cond:BoundedDensities} and \ref{Cond:MillerThm4}. Define $\pi^{(\beta)}_{n}(\theta | y) := \pi^{(\beta)}(\theta|y = \{y_1, \dots, y_n\})$, $m^{(\beta)}_{f,n}(\tilde{y} | y) := m^{(\beta)}_f(\tilde{y}|y = \{y_1, \dots, y_n\})$ and $\{y_i\}_{i=1}^n \sim g$. Then
\begin{equation}
    C^{(\beta)}(f,y):= \int\BD(g||f(\cdot;\theta))\pi^{(\beta)}_{n}(\theta|y)d\theta - \BD(g||m^{(\beta)}_{f,n}(\cdot|y))\underset{n\rightarrow \infty}{\longrightarrow} 0
\end{equation}
\label{Cor:ReminderConvergence1}
\end{corollary}

\begin{proof}
Firstly, under Condition \ref{Cond:MillerThm4}, Theorem \ref{Thm:MillerThm4} proves that 
\begin{equation}
    \int_{B_{\varepsilon}(\theta^{(\beta)}_g)} \pi^{(\beta)}_{n}(\theta | y)d\theta \underset{n\rightarrow \infty}{\longrightarrow} 1, \forall \varepsilon > 0
\end{equation}
where $B_r(x_0) = \{x \in \mathbb{R}^p: |x - x_0| < r\}$. This is the same a saying that $\theta \overset{P}{\rightarrow}\theta^{(\beta)}_g$ for $\theta \sim \pi^{(\beta)}_{n}(\theta | y)$. 

We now investigate both terms in $C^{(\beta)}(f,y)$ and show that both of their limits are $\BD(g||f(\cdot;\theta^{(\beta)}_g))$ and therefore the limit of their difference is 0.
Under Condition \ref{Cond:BoundedDensities}, Lemma \ref{Lem:BoundingBetaD} shows that $\psi^{(1)}(\theta) := \BD(g||f(\cdot;\theta))$ is bounded and continuous function of $\theta$ and therefore as convergence in probability implies convergence in distribution we have that 
\begin{align}
    &\mathbb{E}[\psi^{(1)}(\theta)] \longrightarrow \mathbb{E}[\psi^{(1)}(\theta^{(\beta)}_{g})] = \psi^{(1)}(\theta^{(\beta)}_g)\\
    \Rightarrow&\int\BD(g||f(\cdot;\theta))\pi^{(\beta)}_{n}(\theta|y)d\theta \longrightarrow\BD(g||f(\cdot;\theta^{(\beta)}_g))
\end{align}

Secondly, under \ref{Cond:BoundedDensities} $\psi^{(2)}_{\tilde{y}}(\theta):= f(\tilde{y};\theta)$  is also a continuous and bounded function of $\theta$ for all $\tilde{y}$ and therefore
\begin{align}
    &\mathbb{E}[\psi^{(2)}_{\tilde{y}}(\theta)] \longrightarrow \mathbb{E}[\psi^{(2)}_{\tilde{y}}(\theta^{(\beta)}_g)] = \psi^{(2)}_{\tilde{y}}(\theta^{(\beta)}_g)\\
    \Rightarrow&~ m^{(\beta)}_{f,n}(\tilde{y}|y) := \int f(\tilde{y};\theta)\pi^{(\beta)}_{n}(\theta|y)d\theta\overset{}{\longrightarrow} f(\tilde{y};\theta^{(\beta)}_g)
\end{align}
%
This proves that $m^{(\beta)}_{f,n}(\cdot|y)$ converges as $n\rightarrow\infty$ to $f(\cdot;\theta^{(\beta)}_g)$ pointwise.
Under Condition \ref{Cond:BoundedDensities} the integrand of the \textBD is bounded (Lemma \ref{Lem:BoundingBetaD}) and therefore Lebesgue's Dominated Convergence Theorem provides that
\begin{align}
    \BD(g||m^{(\beta)}_{f, n}(\cdot|y)) &\overset{}{\longrightarrow} \BD(g||f(\cdot;\theta^{(\beta)}_g)).
\end{align}

Therefore, both $\int\BD(g||f(\cdot;\theta))\pi^{(\beta)}_{n}(\theta|y)d\theta$ and $\BD(g||m^{(\beta)}_{f,n}(\cdot|y))$ converge as $n\rightarrow\infty$ to $\BD(g||f(\cdot;\theta^{(\beta)}_g))$ which by Slutsky's Theorem proves that $C^{(\beta)}(f,y) \rightarrow 0$.

\end{proof}

\begin{corollary}[Posterior Predictive Convergence II]
Assume Conditions \ref{Cond:MillerThm4}. 
\begin{itemize}
    \item[1)] Define $\pi^{(\beta)}_{n}(\theta | y) := \pi^{(\beta)}(\theta|y = \{y_1, \dots, y_n\})$, $\pi^{(\beta)}_{n}(\eta | y) := \pi^{(\beta)}(\eta|y = \{y_1, \dots, y_n\})$ and $\{y_i\}_{i=1}^n \sim g$. Then for all $d \geq 0$
    \begin{equation}
        \pi^{(\beta)}_n(\mathcal{S}_d^{(j)} | y) \overset{P}{\rightarrow} 1, \quad\textrm{ as } n\rightarrow\infty\nonumber
    \end{equation}
    for $j = 1,2$ and $\mathcal{S}_d^{(j)}$ defined in Condition \ref{Cond:StochasticPosteriorConcentration} for the \textBD.
    \item[2)] Define $\pi^{(\beta)}_{n_1}(\theta_1 | y) := \pi^{(\beta)}(\theta|y = \{y_1, \dots, y_{n_1}\})$, $\pi^{(\beta)}_{n_2}(\theta_2 | y^{\prime}) := \pi^{(\beta)}(\theta|y^{\prime} = \{y^{\prime}_1, \dots, y^{\prime}_{n_2}\})$ with $\{y_i\}_{i=1}^n \sim g_1$ and $\{y^{\prime}_i\}_{i=1}^n \sim g_2$. Then for all $d \geq 0$
    \begin{equation}
        \pi^{(\beta)}_{n_1, n_2}(\mathcal{S}_d^{(j)} | y) \overset{P}{\rightarrow} 1, \quad \textrm{ as } \min(n_1, n_2)\rightarrow\infty\nonumber
    \end{equation}
    for $j = 1,2$ and $\mathcal{S}_d^{(j)}$ defined in Condition \ref{Cond:StochasticPosteriorConcentration2} for the \textBD.
\end{itemize}

\label{Cor:ReminderConvergence2}
\end{corollary}

\begin{proof}
\textbf{Part 1)}

We prove the result for $j = 1$, with $j = 2$ following immediately by symmetry. 
%
%
Under Condition \ref{Cond:MillerThm4}, Theorem \ref{Thm:MillerThm4} proves that 
\begin{align}
    \int_{B_{\varepsilon}(\theta^{(\beta)}_g)} \pi^{(\beta)}_{n}(\theta | y)d\theta \underset{n\rightarrow \infty}{\longrightarrow} 1\textrm{ and }
    \int_{B_{\varepsilon}(\eta^{(\beta)}_g)} \pi^{(\beta)}_{n}(\eta | y)d\eta \underset{n\rightarrow \infty}{\longrightarrow} 1, \forall \varepsilon > 0\nonumber
\end{align}
where $B_r(x_0) = \{x \in \mathbb{R}^p: |x - x_0| < r\}$. Therefore defining $\mathcal{B}_{\varepsilon}(\theta^{(\beta)}_g, \eta^{(\beta)}_g) := B_{\varepsilon}(\theta^{(\beta)}_g) \cap B_{\varepsilon}(\eta^{(\beta)}_g)$ we have that 
\begin{align}
    \int_{\mathcal{B}(\theta^{(\beta)}_g, \eta^{(\beta)}_g)} \pi^{(\beta)}_{n}(\theta | y)\pi^{(\beta)}_{n}(\eta | y)d\eta d\theta \underset{n\rightarrow \infty}{\longrightarrow} 1, \forall \varepsilon > 0\nonumber
\end{align}
Now by definition
\begin{equation}
    \BD(g||h(\cdot|\eta^{(\beta)}_g)) - \BD(g||h(\cdot|I_f(\theta^{(\beta)}_g))) \leq 0,\nonumber
\end{equation}
and under Condition \ref{Cond:MillerThm4} for every $d > 0$ there exists a $\varepsilon$ such that $\theta, \eta \in \mathcal{B}_{\varepsilon}(\theta^{(\beta)}_g, \eta^{(\beta)}_g)$ implies that $\theta, \eta \in \mathcal{S}_d^{(1)}$ and therefore
\begin{align}
    \int_{ \mathcal{S}_d^{(1)}} \pi^{(\beta)}_{n}(\theta | y)\pi^{(\beta)}_{n}(\eta | y)d\eta d\theta \underset{n\rightarrow \infty}{\longrightarrow} 1.\nonumber
\end{align}

\textbf{Part 2)}

We prove the result for $j = 1$, with $j = 2$ following immediately by symmetry.
%
%
Under Condition \ref{Cond:MillerThm4}, Theorem \ref{Thm:MillerThm4} proves that 
\begin{align}
    \int_{B_{\varepsilon}(\theta^{(\beta)}_{g_1})} \pi^{(\beta)}_{n_1}(\theta_1 | y)d\theta_1 \underset{n_1\rightarrow \infty}{\longrightarrow} 1 \textrm{ and }
    \int_{B_{\varepsilon}(\theta^{(\beta)}_{g_2})} \pi^{(\beta)}_{n_2}(\theta_2 | y^{\prime})d\theta_2 \underset{n_2\rightarrow \infty}{\longrightarrow} 1, \forall \varepsilon > 0\nonumber\\
\end{align}
where $B_r(x_0) = \{x \in \mathbb{R}^p: |x - x_0| < r\}$. Therefore defining $\mathcal{B}_{\varepsilon}(\theta^{(\beta)}_{g_1}, \theta^{(\beta)}_{g_2}) := B_{\varepsilon}(\theta^{(\beta)}_{g_1}) \cap B_{\varepsilon}(\theta^{(\beta)}_{g_2})$ we have that 
\begin{align}
    \int_{\mathcal{B}(\theta^{(\beta)}_{g_1}, \theta^{(\beta)}_{g_2})} \pi^{(\beta)}_{n_1}(\theta_1 | y)\pi^{(\beta)}_{n_2}(\theta_2 | y^{\prime})d\theta_1 d\theta_2 \underset{n\rightarrow \infty}{\longrightarrow} 1, \forall \varepsilon > 0\nonumber
\end{align}
Now by definition
\begin{equation}
    \BD(g_2||f(\cdot|\theta^{(\beta)}_{g_2})) - \BD(g_2||h(\cdot|\theta^{(\beta)}_{g_1})) \leq 0,\nonumber
\end{equation}
and under Condition \ref{Cond:MillerThm4}, for every $d > 0$ there exists a $\varepsilon$ such that $\theta_1, \theta_2 \in \mathcal{B}_{\varepsilon}(\theta^{(\beta)}_{g_1}, \theta^{(\beta)}_{g_2})$ implies that $\theta_1, \theta_2 \in \mathcal{S}_d^{(1)}$ and therefore
\begin{align}
    \int_{ \mathcal{S}_d^{(1)}} \pi^{(\beta)}_{n_1}(\theta_1 | y)\pi^{(\beta)}_{n_2}(\theta_2 | y^{\prime})d\theta_1 d\theta_2 \underset{\min(n_1, n_2)\rightarrow \infty}{\longrightarrow} 1.\nonumber
\end{align}
\end{proof}

\color{black}


%
%
%

\subsubsection{Proof of Lemma 1}

We are now able to investigate the stability in the posterior predictive approximation to the \DGP of inference using the \KLD-Bayes.

\begin{proof}
Firstly, the logarithm is a concave function and therefore the negative logarithm is a convex function which is sufficient to prove the convexity of \KLD in its second argument. 
%
Further, by the definition of the \KLD, we can see that
\begin{align}
\KLD(g||f) &= \KLD(g||h)+\int g \log \frac{h}{f}d\mu.\label{Equ:KLDTriangle}
\end{align}
Now we can use the convexity of the \KLD and  Jensen's inequality, to show that
\begin{align}
\KLD(g||&m^{\KLD}_{f}(\cdot|y))\leq \int \KLD(g||f(\cdot;\theta))\pi^{\KLD}(\theta|y)d\theta\nonumber
\end{align}
Now adding and subtracting $\int \KLD(g||f(\cdot;I_h(\eta)))\pi^{\KLD}(\eta|y)d\eta$ provides
\begin{align}
&\KLD(g||m^{\KLD}_{f}(\cdot|y)) \leq \int \KLD(g||f(\cdot;\theta))\pi^{\KLD}(\theta|y)d\theta\nonumber\\
=& \int \KLD(g||f(\cdot;\theta))\pi^{\KLD}(\theta|y)d\theta + \int \KLD(g||f(\cdot;I_h(\eta)))\pi^{\KLD}(\eta|y)d\eta\nonumber\\
&- \int \KLD(g||f(\cdot;I_h(\eta)))\pi^{\KLD}(\eta|y)d\eta\nonumber\\
=& \int \KLD(g||f(\cdot;I_h(\eta)))\pi^{\KLD}(\eta|y)d\eta\nonumber\\
&+ \int\int \left\{\KLD(g||f(\cdot;\theta))-  \KLD(g||f(\cdot;I_h(\eta)))\right\}\pi^{\KLD}(\theta|y)d\theta \pi^{\KLD}(\eta|y)d\eta\nonumber
\end{align}
%
Now we can apply Lemma \ref{Lem:StochasticCDFResult} to random variable $\left\{\KLD(g||f(\cdot;\theta))-  \KLD(g||f(\cdot;I_h(\eta)))\right\}\in\mathbb{R}$ on $\Theta\times\mathcal{A}$ which by using \eqref{Equ:StochasticCondition_fh_S2} of Condition \ref{Cond:StochasticPosteriorConcentration} applied to the \KLD provides 
\begin{align}
    \int\int \left\{\KLD(g||f(\cdot;\theta))-  \KLD(g||f(\cdot;I_h(\eta)))\right\}\pi^{(\beta)}(\theta|y)d\theta \pi^{(\beta)}(\eta|y)d\eta\nonumber \leq  \frac{1}{c_2}.\nonumber
\end{align}
We can now use the triangular-type relationship of  \eqref{Equ:KLDTriangle} to show that
\begin{align}
\KLD(g||&m^{\KLD}_{f}(\cdot|y))\leq \int \KLD(g||f(\cdot;I_h(\eta)))\pi^{\KLD}(\eta|y)d\eta + \frac{1}{c_2}\nonumber\\
&= \int \left(\int g(\cdot) \log \frac{h(\cdot;\eta)}{f(\cdot;I_h(\eta))}d\mu+\KLD(g||h(\cdot;\eta))\right)\pi^{\KLD}(\eta|y)d\theta + \frac{1}{c_2}\nonumber\\
&= \int \int g(\cdot) \log \frac{h(\cdot;\eta)}{f(\cdot;I_h(\eta))}d\mu\pi^{\KLD}(\eta|y)d\eta+\int\KLD(g||h(\cdot;\eta))\pi^{\KLD}(\eta|y)d\eta\nonumber\\
&\quad + \frac{1}{c_2} +\KLD(g||m^{\KLD}_{h}(\cdot|y))-\KLD(g||m^{\KLD}_{h}(\cdot|y)).\nonumber
\end{align}
The same arguments of Lemma \ref{Lem:StochasticCDFResult} and \eqref{Equ:StochasticCondition_fh_S1} of Condition \ref{Cond:StochasticPosteriorConcentration} show
\begin{align}
\KLD(g||&m^{\KLD}_{h}(\cdot|y))\leq \int \KLD(g||h(\cdot;\eta))\pi^{\KLD}(\eta|y)d\eta\nonumber\\
&\leq \int \KLD(g||h(\cdot;I_f(\theta)))\pi^{\KLD}(\theta|y)d\theta + \frac{1}{c_1}\nonumber\\
&= \int \left(\int g(\cdot) \log \frac{f(\cdot;\theta)}{h(\cdot;I_f(\theta))}d\mu+\KLD(g||f(\cdot;\theta))\right)\pi^{\KLD}(\theta|y)d\theta + \frac{1}{c_1}\nonumber\\
&= \int \int g(\cdot) \log \frac{f(\cdot;\theta)}{h(\cdot;I_f(\theta))}d\mu\pi^{\KLD}(\theta|y)d\theta+\int\KLD(g||f(\cdot;\theta))\pi^{\KLD}(\theta|y)d\theta\nonumber\\
&\quad + \frac{1}{c_1} +\KLD(g||m^{\KLD}(\cdot|y))-\KLD(g||m^{\KLD}(\cdot|y))\nonumber
\end{align}

Combining the above two results provides the following bound,
\begin{align}
|\KLD(g||m^{\KLD}_{f}(\cdot|y))-
\KLD(g||m^{\KLD}_{h}(\cdot|y))|&\leq C^{\KLD}(f,h,y) + \frac{1}{c} + T(f,h,y) \nonumber
\end{align}
where $c := \min\{c_1, c_2\}$ and
\begin{align}
T(f,h,y):&= \max \left\lbrace\int \int g(\cdot) \log \frac{f(\cdot;\theta)}{h(\cdot;I_f(\theta))}d\mu\pi^{\KLD}(\theta|y)d\theta,\right.\nonumber\\
&\qquad \left.\int\int  g(\cdot) \log \frac{h(\cdot;\eta)}{f(\cdot;I_h(\eta))}d\mu\pi^{\KLD}(\eta|y)d\eta \right\rbrace\nonumber \\
C^{\KLD}(f,h,y):&= \max \left\lbrace\int\KLD(g||f(\cdot;\theta))\pi^{\KLD}(\theta|y)d\theta-\KLD(g||m^{\KLD}(\cdot|y)),\right.\nonumber\\
&\qquad\left.\int\KLD(g||h(\cdot;\eta))\pi^{\KLD}(\eta|y)d\eta-\KLD(g||m^{\KLD}_{h}(\cdot|y)) \right\rbrace \nonumber
\end{align}
as required

\end{proof}

\subsubsection{Proof of Theorem 1}

Theorem 1 uses the convexity of the \textBD (Lemma \ref{Lem:betaDivConvexity}) and the triangular relationship between the \textBD and the \TVD (Lemma \ref{Lem:betaDivTVDTriangle}) to prove stability in the posterior predictive approximation to the \DGP of inference using the \textBD.

\begin{proof}
Using the convexity of the \textBD (Lemma \ref{Lem:betaDivConvexity}) and Jensen's inequality, 
\begin{align}
\BD(g||&m^{(\beta)}_{f}(\cdot|y)) \leq \int \BD(g||f(\cdot;\theta))\pi^{(\beta)}(\theta|y)d\theta\nonumber
\end{align}
Adding and subtracting $\int \BD(g||f(\cdot;I_h(\eta)))\pi^{(\beta)}(\eta|y)d\eta$ we have that
\begin{align}
&\BD(g||m^{(\beta)}_{f}(\cdot|y)) \leq \int \BD(g||f(\cdot;\theta))\pi^{(\beta)}(\theta|y)d\theta\nonumber\\
=& \int \BD(g||f(\cdot;\theta))\pi^{(\beta)}(\theta|y)d\theta\nonumber\\
&+ \int \BD(g||f(\cdot;I_h(\eta)))\pi^{(\beta)}(\eta|y)d\eta - \int \BD(g||f(\cdot;I_h(\eta)))\pi^{(\beta)}(\eta|y)d\eta\nonumber\\
=& \int \BD(g||f(\cdot;I_h(\eta)))\pi^{(\beta)}(\eta|y)d\eta\nonumber\\
&+ \int\int \left\{\BD(g||f(\cdot;\theta))-  \BD(g||f(\cdot;I_h(\eta)))\right\}\pi^{(\beta)}(\theta|y)d\theta \pi^{(\beta)}(\eta|y)d\eta\nonumber
\end{align}
Now we can apply Lemma \ref{Lem:StochasticCDFResult} to random variable $\left\{\BD(g||f(\cdot;\theta))-  \BD(g||f(\cdot;I_h(\eta)))\right\}\in\mathbb{R}$ on $\Theta\times\mathcal{A}$ which by using \eqref{Equ:StochasticCondition_fh_S2} of Condition \ref{Cond:StochasticPosteriorConcentration} applied to the \textBD provides 
\begin{align}
    \int\int \left\{\BD(g||f(\cdot;\theta))-  \BD(g||f(\cdot;I_h(\eta)))\right\}\pi^{(\beta)}(\theta|y)d\theta \pi^{(\beta)}(\eta|y)d\eta\nonumber \leq  \frac{1}{c_2}.\nonumber
\end{align}
We then use the the triangular-type relationship between the \textBD and the \TVD (Lemma \ref{Lem:betaDivTVDTriangle}) to show that
\begin{align}
\BD&(g||m^{(\beta)}_{f}(\cdot|y)) \leq \int \BD(g||f(\cdot;\theta))\pi^{(\beta)}(\theta|y)d\theta\nonumber\\
&\leq \int \BD(g||f(\cdot;I_h(\eta)))\pi^{(\beta)}(\eta|y)d\eta + \frac{1}{c_2}\nonumber\\
&\leq \int \left(\frac{M^{\beta - 1}(3\beta - 2)}{\beta(\beta - 1)}\TVD(f(\cdot;I_h(\eta)),h(\cdot;\eta))+\BD(g||h(\cdot;\eta))\right)\pi^{(\beta)}(\eta|y)d\eta + \frac{1}{c_2}\nonumber\\
&= \int \frac{M^{\beta - 1}(3\beta - 2)}{\beta(\beta - 1)}\TVD(f(\cdot;I_h(\eta)),h(\cdot;\eta))\pi^{(\beta)}(\eta|y)d\eta + \int\BD(g||h(\cdot;\eta))\pi^{(\beta)}(\eta|y)d\eta  + \frac{1}{c_2}\nonumber\\
&\quad +\BD(g||m^{(\beta)}_{h}(\cdot|y))-\BD(g||m^{(\beta)}(\cdot|y)).\nonumber
\end{align}
The same arguments this time using \eqref{Equ:StochasticCondition_fh_S1} of Condition \ref{Cond:StochasticPosteriorConcentration} can also be used to show that
\begin{align}
\BD&(g||m^{(\beta)}_{h}(\cdot|y)) \leq \int \BD(g||h(\cdot;\eta))\pi^{(\beta)}(\eta|y)d\eta\nonumber\\
&\leq \int \BD(g||h(\cdot;I_f(\theta)))\pi^{(\beta)}(\theta|y)d\theta + \frac{1}{c_1}\nonumber\\
&\leq \int \left(\frac{M^{\beta - 1}(3\beta - 2)}{\beta(\beta - 1)}\TVD(f(\cdot;\theta),h(\cdot;I_f(\theta)))+\BD(g||f(\cdot;\theta))\right)\pi^{(\beta)}(\theta|y)d\theta + \frac{1}{c_1}\nonumber\\
&= \int \frac{M^{\beta - 1}(3\beta - 2)}{\beta(\beta - 1)}\TVD(f(\cdot;\theta),h(\cdot;I_f(\theta)))\pi^{(\beta)}(\theta|y)d\theta+\int\BD(g||f(\cdot;\theta))\pi^{(\beta)}(\theta|y)d\theta + \frac{1}{c_1}\nonumber\\
&\quad +\BD(g||m^{(\beta)}_{f}(\cdot|y))-\BD(g||m^{(\beta)}_{f}(\cdot|y))\nonumber
\end{align}

Combining the above two results provides the following bound,
\begin{align}
|\BD(g||m^{(\beta)}_{f}(\cdot|y))-
\BD(g||m^{(\beta)}_{h}(\cdot|y))|&\leq \frac{M^{\beta - 1}(3\beta - 2)}{\beta(\beta - 1)}\epsilon + \frac{1}{c} + C^{(\beta)}(f,h,y), \nonumber 
\end{align}
where $c = \max\{c_1, c_2\}$ and
\begin{align}
C^{(\beta)}(f,h,y):&= \max \left\lbrace\int\BD(g||f(\cdot;\theta))\pi^{(\beta)}(\theta|y)d\theta-\BD(g||m^{(\beta)}(\cdot|y)),\right.\nonumber\\
&\qquad\left.\int\BD(g||h(\cdot;\eta))\pi^{(\beta)}(\eta|y)d\eta-\BD(g||m^{(\beta)}(\cdot|y)) \right\rbrace \nonumber
\end{align}
as required.
\end{proof}

\subsubsection{Proof of Theorem 2}

Theorem 2 uses the convexity of the \textBD (Lemma \ref{Lem:betaDivConvexity}), the triangular relationship between the \textBD and the \TVD (Lemma \ref{Lem:betaDivTVDTriangle}) and the three-point property the \textBD (Lemmas \ref{Lem:betaDivTriangle} and \ref{Lem:remainder_TVD_bound}) to establish the posterior predictive stability to the likelihood model's specification provided by inference using the \textBD.

\begin{proof}

By the convexity of the \textBD for $1< \beta\leq 2$ (Lemma \ref{Lem:betaDivConvexity}) we can apply Jensen's inequality to show that 
\begin{align}
&\BD(m^{(\beta)}_{f}(\cdot|y)||m^{(\beta)}_{h}(\cdot|y)))\leq \int \BD(m^{(\beta)}_{f}(\cdot|y)||h(\cdot;\eta))\pi^{(\beta)}(\eta|y)d\eta\nonumber\\
\leq& \int \left\lbrace\int \BD(f(\cdot;\theta)||h(\cdot;\eta))\pi^{(\beta)}(\theta|y)d\theta\right\rbrace\pi^{(\beta)}(\eta|y)d\eta.\nonumber
\end{align}
Now the three-point property associated with the \textBD (Lemma \ref{Lem:betaDivTriangle}) gives us that 
\begin{equation}
\BD(f||h)=\BD(g||h)-\BD(g||f)+R(g||f||h)\nonumber
\end{equation}
where $R(g||f||h)$ is defined in  \eqref{Equ:Rghf}. 
Using this here provides
\begin{align}
\BD(m^{(\beta)}_{f}(\cdot|y)||&m^{(\beta)}_{h}(\cdot|y)))\nonumber\\
\leq& \int \left\lbrace\int \BD(f(\cdot;\theta)||h(\cdot;\eta))\pi^{(\beta)}(\theta|y)d\theta\right\rbrace\pi^{(\beta)}(\eta|y)d\eta\nonumber\\
=& \int \left\lbrace\int \left[\BD(g||h(\cdot;\eta))-\BD(g||f(\cdot;\theta))\right.\right.\nonumber\\
&\left.\left.+R(g||f(\cdot;\theta)||h(\cdot;\eta)\right]\pi^{(\beta)}(\theta|y)d\theta\right\rbrace\pi^{(\beta)}(\eta|y)d\eta\nonumber\\
=& \int \BD(g||h(\cdot;\eta))\pi^{(\beta)}(\eta|y)d\eta-\int \BD(g||f(\cdot;\theta))\pi^{(\beta)}(\theta|y)d\theta\nonumber\\
&+\int\int R(g||f(\cdot;\theta)||h(\cdot;\eta))\pi^{(\beta)}(\theta|y)d\theta\pi^{(\beta)}(\eta|y)d\eta.\nonumber
\end{align}
%
%
Now adding and subtracting $\int \BD(g||h(\cdot;I_f(\theta)))\pi^{(\beta)}(\theta|y)d\theta$ we have
\begin{align}
&\BD(m^{(\beta)}_{f}(\cdot|y)||m^{(\beta)}_{h}(\cdot|y))) \leq \int \BD(g||h(\cdot;\eta))\pi^{(\beta)}(\eta|y)d\eta-\int \BD(g||f(\cdot;\theta))\pi^{(\beta)}(\theta|y)d\theta\nonumber\\
&+ \int \BD(g||h(\cdot;I_f(\theta)))\pi^{(\beta)}(\theta|y)d\theta -  \int \BD(g||h(\cdot;I_f(\theta)))\pi^{(\beta)}(\theta|y)d\theta\nonumber\\
&+\int\int R(g||f(\cdot;\theta)||h(\cdot;\eta))\pi^{(\beta)}(\theta|y)d\theta\pi^{(\beta)}(\eta|y)d\eta.\nonumber\\
=&\int \left\{\BD(g||h(\cdot;I_f(\theta))) - \BD(g||f(\cdot;\theta))\right\}\pi^{(\beta)}(\theta|y)d\theta\nonumber\\
&+ \int\int\left\{ \BD(g||h(\cdot;\eta)) -  \BD(g||h(\cdot;I_f(\theta)))\right\}\pi^{(\beta)}(\theta|y)d\theta\pi^{(\beta)}(\eta|y)d\eta\nonumber\\
&+\int\int R(g||f(\cdot;\theta)||h(\cdot;\eta))\pi^{(\beta)}(\theta|y)d\theta\pi^{(\beta)}(\eta|y)d\eta.\nonumber
\end{align}
Now we can apply Lemma \ref{Lem:StochasticCDFResult} to random variable $\left\{ \BD(g||h(\cdot;\eta)) -  \BD(g||h(\cdot;I_f(\theta)))\right\}\in\mathbb{R}$ on $\Theta\times\mathcal{A}$ which by using \eqref{Equ:StochasticCondition_fh_S1} of Condition \ref{Cond:StochasticPosteriorConcentration} applied to the \textBD provides 
%
\begin{align}
    \int\int \left\{ \BD(g||h(\cdot;\eta)) -  \BD(g||h(\cdot;I_f(\theta)))\right\}\pi^{(\beta)}(\theta|y)d\theta \pi^{(\beta)}(\eta|y)d\eta\nonumber \leq  \frac{1}{c_1}.\nonumber
\end{align}
As a result
\begin{align}
&\BD(m^{(\beta)}_{f}(\cdot|y)||m^{(\beta)}_{h}(\cdot|y))) \leq \int \left\{\BD(g||h(\cdot;I_f(\theta))) - \BD(g||f(\cdot;\theta))\right\}\pi^{(\beta)}(\theta|y)d\theta\nonumber\\
&+ \frac{1}{c_1} + \int\int R(g||f(\cdot;\theta)||h(\cdot;\eta))\pi^{(\beta)}(\theta|y)d\theta\pi^{(\beta)}(\eta|y)d\eta.\nonumber
\end{align}
We can now apply the triangle type inequality from Lemma \ref{Lem:betaDivTVDTriangle},
\begin{align}
\BD(m^{(\beta)}_{f}(\cdot|y)||&m^{(\beta)}_{h}(\cdot|y)))\nonumber\\
\leq& \int \left(\BD(g||h(\cdot;I_f(\theta)))- \BD(g||f(\cdot;\theta))\right)\pi^{(\beta)}(\theta|y)d\theta\nonumber\\
&+ \frac{1}{c_1}+\int\int R(g||f(\cdot;\theta)||h(\cdot;\eta))\pi^{(\beta)}(\theta|y)d\theta\pi^{(\beta)}(\eta|y)d\eta.\nonumber\\
\leq& \int \frac{M^{\beta - 1}(3\beta - 2)}{\beta(\beta - 1)}\TVD(h(\cdot;I_f(\theta)),f(\cdot;\theta))\pi^{(\beta)}(\theta|y)d\theta\nonumber\\
&+ \frac{1}{c_1}+\int\int R(g||f(\cdot;\theta)||h(\cdot;\eta))\pi^{(\beta)}(\theta|y)d\theta\pi^{(\beta)}(\eta|y)d\eta.\nonumber
\end{align}
Given the neighbourhood of likelihood models $\mathcal{N}_{\epsilon}^{\TVD}$ we can then write
\begin{align}
&\BD(m^{(\beta)}_{f}(\cdot|y)||m^{(\beta)}_{h}(\cdot|y)))\leq \int \frac{M^{\beta - 1}(3\beta - 2)}{\beta(\beta - 1)}\TVD(h(\cdot;I_f(\theta)),f(\cdot;\theta))\pi^{(\beta)}(\theta|y)d\theta\nonumber\\
&+ \frac{1}{c_1}+\int\int R(g||f(\cdot;\theta)||h(\cdot;\eta))\pi^{(\beta)}(\theta|y)d\theta\pi^{(\beta)}(\eta|y)d\eta.\nonumber\\
\leq& \frac{M^{\beta - 1}(3\beta - 2)}{\beta(\beta - 1)}\epsilon+ \frac{1}{c_1}+\int\int R(g||f(\cdot;\theta)||h(\cdot;\eta))\pi^{(\beta)}(\theta|y)d\theta\pi^{(\beta)}(\eta|y)d\eta.\nonumber
\end{align}
Now from Lemma \ref{Lem:remainder_TVD_bound} we have that $    R(g||f(\cdot;\theta)||h(\cdot;\eta))\leq 2\frac{M^{\beta - 1}}{\beta-1}\TVD(g, f(\cdot;\theta))$ and as a result we can bound 
\begin{align}
\BD(m^{(\beta)}_{f}(\cdot|y)||m^{(\beta)}_{h}(\cdot|y)))\leq \frac{M^{\beta - 1}(3\beta - 2)}{\beta(\beta - 1)}\epsilon+ \frac{1}{c_1}+2\frac{M^{\beta - 1}}{\beta-1}\int \TVD(g, f(\cdot;\theta))\pi^{(\beta)}(\theta|y)d\theta.\nonumber
\end{align}
This provides the first part of the required result. We note that we could have instead considered $\BD(m^{(\beta)}_{h}(\cdot|y)||m^{(\beta)}_{f}(\cdot|y)))$, applied the corresponding version of the three-point property of Bregman divergences, with remainder $R(g||h||f)=\int (g-h)\left(\frac{1}{\beta-1}f^{\beta-1}-\frac{1}{\beta-1}h^{\beta-1}\right)d\mu$, used Lemma \ref{Lem:StochasticCDFResult} with  \eqref{Equ:StochasticCondition_fh_S2} of \ref{Cond:StochasticPosteriorConcentration} and Lemma \ref{Lem:remainder_TVD_bound}, therefore we also have that
\begin{align}
\BD(m^{(\beta)}_{h}(\cdot|y)||&m^{(\beta)}_{f}(\cdot|y))) \leq \frac{M^{\beta - 1}(3\beta - 2)}{\beta(\beta - 1)}\epsilon + \frac{1}{c_2} + 2\frac{M^{\beta - 1}}{\beta-1}\int \TVD(g, h(\cdot;\eta))\pi^{(\beta)}(\eta|y)d\eta.\nonumber
\end{align}
providing the second part of the required result.
\end{proof}



\subsection{Proofs: Stability to the \DGP}
\label{ssec:proofs_DGPs}


\subsubsection{A useful Lemma for proving Theorems 3 and 4}

In order to prove Theorems 3 and 4, Lemma \ref{Lem:betaDivTVDTriangle_g1g2} provides a second triangle-type inequality for the \textBD and \TVD in the case where one model is estimated under two \DGP{}s.

\begin{lemma}[Another triangle inequality relating the \textBD and the \TVD]
For densities $f$, $g_1$ and $g_2$ with the property that there exists $M<\infty$ satisfying Condition \ref{Cond:BoundedDensities2} and $1< \beta\leq 2$ we have that 
\begin{align}
\left|\BD(g_1||f) - \BD(g_2||f)\right| &\leq\frac{M^{\beta-1}(\beta + 2)}{\beta(\beta-1)}\TVD(g_1,g_2).\nonumber
\end{align} 
\label{Lem:betaDivTVDTriangle_g1g2}
\end{lemma}

\begin{proof}
By the definition of the \textBD, we can rearrange 
\begin{align}
\BD&(g_1||f)=\BD(g_2||f) +\nonumber\\
&\left(\int\left( \frac{1}{\beta(\beta-1)}g_1(y)^{\beta}-\frac{1}{\beta(\beta-1)}g_2(y)^{\beta}+\frac{1}{\beta-1}g_2(y)f(y)^{\beta-1}-\frac{1}{\beta-1}g_1(y) f(y)^{\beta-1}\right)dy\right)\nonumber\\
=&\BD(g_2||f)+\left(\frac{1}{\beta(\beta-1)}\int \left(g_1(y)^{\beta}-g_2(y)^{\beta}\right) dy +\frac{1}{\beta-1}\int f(y)^{\beta-1}\left(g_2(y)-g_1(y)\right)dy \right)\nonumber
\end{align}
As in Lemma \ref{Lem:TVDsimp}, define $A^{+}:=\left\lbrace y: g_2(y)>g_1(y)\right\rbrace$ and $A^{-}:=\left\lbrace y: g_1(y)>g_2(y)\right\rbrace$. By the monotonicity of the function $y^{\beta}$ when $1\leq \beta\leq 2$ we have that 
\begin{align}
\int_{A^{+}} g_1(y)^{\beta}-g_2(y)^{\beta} dy&<0\nonumber\\
\int_{A^{-}} f(x)^{\beta-1}\left(g_2(y)-g_1(y)\right)dy &<0\nonumber
\end{align}
therefore removing these two terms provides an upper bound 
\begin{align}
\BD&(g_1||f)\nonumber\\
=&\BD(g_2||f)+\left(\frac{1}{\beta(\beta-1)}\int \left(g_1(y)^{\beta}-g_2(y)^{\beta}\right) dy +\frac{1}{\beta-1}\int f(x)^{\beta-1}\left(g_2(y)-g_1(y)\right)dy \right)\nonumber \\
\leq& \BD(g_2||f)+\frac{1}{\beta(\beta-1)}\int_{A^{-}} \left(g_1(y)^{\beta}-g_2(y)^{\beta} \right)dy +\frac{1}{\beta-1}\int_{A^{+}} f(y)^{\beta-1}\left(g_2(y)-g_1(y)\right)dy.\nonumber
\end{align}
Now adding and subtracting $g_1(y)g_2(y)^{\beta-1}$ provides 
\begin{align}
\BD&(g_1||f)\nonumber\\
\leq& \BD(g_2||f)+\frac{1}{\beta(\beta-1)}\int_{A^{-}} \left(g_1(y)^{\beta}-g_2(y)^{\beta} \right)dy +\frac{1}{\beta-1}\int_{A^{+}} f(y)^{\beta-1}\left(g_2(y)-g_1(y)\right)dy\nonumber\\
=& \BD(g_2||f)+\frac{1}{\beta(\beta-1)}\int_{A^{-}} \left(g_1(y)^{\beta} - g_1(y)g_2(y)^{\beta-1} + g_1(y)g_2(y)^{\beta-1} - g_2(y)^{\beta} \right)dy\nonumber\\
&+\frac{1}{\beta-1}\int_{A^{+}} f(y)^{\beta-1}\left(g_2(y)-g_1(y)\right)dy\nonumber\\
=& \BD(g_2||f)+\frac{1}{\beta(\beta-1)}\int_{A^{-}} g_1(y)\left(g_1(y)^{\beta-1} - g_2(y)^{\beta-1}\right)dy\nonumber\\
&+ \frac{1}{\beta(\beta-1)}\int_{A^{-}}g_2(y)^{\beta-1}\left(g_1(y) - g_2(y)\right)dy +\frac{1}{\beta-1}\int_{A^{+}} f(y)^{\beta-1}\left(g_2(y)-g_1(y)\right)dy\nonumber\\
=& \BD(g_2||f)+\frac{1}{\beta(\beta-1)}\int_{A^{-}} g_1(y)^{\beta}\left(1 - \frac{g_2(y)^{\beta-1}}{g_1(y)^{\beta-1}}\right)dy\nonumber\\
&+ \frac{1}{\beta(\beta-1)}\int_{A^{-}}g_2(y)^{\beta-1}\left(g_1(y) - g_2(y)\right)dy +\frac{1}{\beta-1}\int_{A^{+}} f(y)^{\beta-1}\left(g_2(y)-g_1(y)\right)dy\nonumber.
\end{align}
%
%
Next, on $A^{-}$ $g_1(x)>g_2(x)$ which implies that $\left(\frac{g_2(x)}{g_1(x)}\right)^{\beta-1}>\frac{g_2(x)}{g_1(x)}$ for $1\leq\beta\leq2$ and
\begin{equation}
\left(1-\frac{g_2(x)^{\beta-1}}{g_1(x)^{\beta-1}}\right)\leq \left(1-\frac{g_2(x)}{g_1(x)}\right).\nonumber
\end{equation}
We can use this to show that
\begin{align}
\BD&(g_1||f) \leq \BD(g_2||f)+\frac{1}{\beta(\beta-1)}\int_{A^{-}} g_1(y)^{\beta}\left(1 - \frac{g_2(y)^{\beta-1}}{g_1(y)^{\beta-1}}\right)dy\nonumber\\ 
&+ \frac{1}{\beta(\beta-1)}\int_{A^{-}}g_2(y)^{\beta-1}\left(g_1(y) - g_2(y)\right)dy +\frac{1}{\beta-1}\int_{A^{+}} f(y)^{\beta-1}\left(g_2(y)-g_1(y)\right)dy\nonumber\\
\leq& \BD(g_2||f)+\frac{1}{\beta(\beta-1)}\int_{A^{-}} g_1(y)^{\beta}\left(1 - \frac{g_2(y)}{g_1(y)}\right)dy\nonumber\\
&+ \frac{1}{\beta(\beta-1)}\int_{A^{-}}g_2(y)^{\beta-1}\left(g_1(y) - g_2(y)\right)dy +\frac{1}{\beta-1}\int_{A^{+}} f(y)^{\beta-1}\left(g_2(y)-g_1(y)\right)dy\nonumber\\
=& \BD(g_2||f)+\frac{1}{\beta(\beta-1)}\int_{A^{-}} g_1(y)^{\beta - 1}\left(g_1(y) - g_2(y)\right)dy\nonumber\\
&+ \frac{1}{\beta(\beta-1)}\int_{A^{-}}g_2(y)^{\beta-1}\left(g_1(y) - g_2(y)\right)dy+\frac{1}{\beta-1}\int_{A^{+}} f(y)^{\beta-1}\left(g_2(y)-g_1(y)\right)dy\nonumber.
\end{align}
We now use the fact that we defined $\max\left\lbrace \esssup f, \esssup g_1, \esssup g_2\right\rbrace\leq M<\infty$ and Lemma \ref{Lem:TVDsimp} to leave
\begin{align}
\BD&(g_1||f) = \BD(g_2||f)+\frac{1}{\beta(\beta-1)}\int_{A^{-}} g_1(y)^{\beta - 1}\left(g_1(y) - g_2(y)\right)dy\nonumber\\
& + \frac{1}{\beta(\beta-1)}\int_{A^{-}}g_2(y)^{\beta-1}\left(g_1(y) - g_2(y)\right)dy +\frac{1}{\beta-1}\int_{A^{+}} f(y)^{\beta-1}\left(g_2(y)-g_1(y)\right)dy\nonumber\\
\leq& \BD(g_2||f)+\frac{M^{\beta - 1}}{\beta(\beta-1)}\int_{A^{-}} \left(g_1(y) - g_2(y)\right)dy + \frac{M^{\beta - 1}}{\beta(\beta-1)}\int_{A^{-}}\left(g_1(y) - g_2(y)\right)dy\nonumber\\
&+\frac{M^{\beta - 1}}{\beta-1}\int_{A^{+}} \left(g_2(y)-g_1(y)\right)dy\nonumber \\
=& \BD(g_2||f)+2\frac{M^{\beta-1}}{\beta(\beta - 1)}\TVD(g_1,g_2) +\frac{M^{\beta-1}}{\beta-1}\TVD(g_1,g_2)\nonumber\\
=& \BD(g_2||f)+\frac{M^{\beta-1}(\beta + 2)}{\beta(\beta-1)}\TVD(g_1,g_2),\nonumber
\end{align}
providing the required result.
\end{proof}

\subsubsection{Proof of Lemma 2}


Similarly to Lemma 1, we are now able to investigate the stability in the posterior predictive approximation to the \DGP of inference using the \KLD-Bayes.

\begin{proof}
The proof of Lemma 1 established the convexity of the \KLD in its second argument.
%
Further, by the definition of the \KLD, we can see that
\begin{align}
\KLD(g_1 ||f) &= \KLD(g_2 || f) + \int g_1\log g_1 - g_2 \log g_2d\mu + \int (g_2 - g_1)\log fd\mu.\label{Equ:KLDTriangle2}
\end{align}
Now we can use the convexity of the \KLD and  Jensen's inequality, to show that
\begin{align}
\KLD(g_1||&m^{\KLD}_{f}(\cdot|y_{1:n_1}))\leq \int \KLD(g_1||f(\cdot;\theta_1))\pi^{\KLD}(\theta|y_{1:n_1})d\theta_1\nonumber
\end{align}
Now adding and subtracting $\int\KLD(g_1||f(\cdot;\theta_2))\pi^{\KLD}(\theta_2|y^{\prime}_{1:n_2})d\theta_2$ provides
\begin{align}
&\KLD(g_1||m^{\KLD}_{f}(\cdot|y_{1:n_1})) \leq \int \KLD(g_1||f(\cdot;\theta_1))\pi^{\KLD}(\theta_1|y_{1:n_1})d\theta_1\nonumber\\
=& \int \KLD(g_1||f(\cdot;\theta_1))\pi^{\KLD}(\theta_1|y_{1:n_1})d\theta_1 + \int\KLD(g_1||f(\cdot;\theta_2))\pi^{\KLD}(\theta_2|y^{\prime}_{1:n_2})d\theta_2\nonumber\\
&- \int\KLD(g_1||f(\cdot;\theta_2))\pi^{\KLD}(\theta_2|y^{\prime}_{1:n_2})d\theta_2\nonumber\\
=& \int \KLD(g_1||f(\cdot;\theta_2))\pi^{\KLD}(\theta_2|y^{\prime}_{1:n_2})d\theta_2\nonumber\\
&+ \int\int\left\{\KLD(g_1||f(\cdot;\theta_1)) - \KLD(g_1||f(\cdot;\theta_2))\right\}\pi^{\KLD}(\theta_1|y_{1:n_1})d\theta_1\pi^{\KLD}(\theta_2|y^{\prime}_{1:n_2})d\theta_2\nonumber
\end{align}
Now we can apply Lemma \ref{Lem:StochasticCDFResult} to random variable $\left\{\KLD(g_1||f(\cdot;\theta_1)) - \KLD(g_1||f(\cdot;\theta_2))\right\}\in\mathbb{R}$ on $\Theta\times\mathcal{A}$ which by using \eqref{Equ:StochasticCondition_g1g2_S2} of Condition \ref{Cond:StochasticPosteriorConcentration2} applied to the \KLD provides 
\begin{align}
    \int\int \left\{\KLD(g_1||f(\cdot;\theta_1)) - \KLD(g_1||f(\cdot;\theta_2))\right\}\pi^{\KLD}(\theta_1|y_1)d\theta_1\pi^{\KLD}(\theta_2|y_2)d\theta_2\nonumber \leq  \frac{1}{c_{\mathcal{S}^{(2)}}}.\nonumber
\end{align}
We can now use the triangular-type relationship of  \eqref{Equ:KLDTriangle2} to show that
\begin{align}
\KLD&(g_1||m^{\KLD}_{f}(\cdot|y_{1:n_1})) \leq  \int \KLD(g_1||f(\cdot;\theta_2))\pi^{\KLD}(\theta_2|y^{\prime}_{1:n_2})d\theta_2 + \frac{1}{c_{\mathcal{S}^{(2)}}}\nonumber\\
=& \int \left( \KLD(g_2||f(\cdot;\theta_2)) + \int g_1\log g_1 - g_2 \log g_2d\mu + \int (g_2 - g_1)\log f(\cdot;\theta_2)d\mu\right)\pi^{\KLD}(\theta_2|y^{\prime}_{1:n_2})d\theta_2\nonumber\\
&+ \frac{1}{c_{\mathcal{S}^{(2)}}}\nonumber\\
=& \int \int (g_2 - g_1)\log f(\cdot;\theta_2)d\mu\pi^{\KLD}(\theta_2|y^{\prime}_{1:n_2})d\theta_2+\int \KLD(g_2||f(\cdot;\theta_2))\pi^{\KLD}(\theta_2|y^{\prime}_{1:n_2})d\theta_2\nonumber\\
&\quad + \int g_1\log g_1 - g_2 \log g_2d\mu + \frac{1}{c_{\mathcal{S}^{(2)}}} +\KLD(g_2||m^{\KLD}_{f}(\cdot|y^{\prime}_{1:n_2}))-\KLD(g_2||m^{\KLD}_{f}(\cdot|y^{\prime}_{1:n_2})).\nonumber
\end{align}
The same arguments of Lemma \ref{Lem:StochasticCDFResult} and \eqref{Equ:StochasticCondition_g1g2_S1} of Condition \ref{Cond:StochasticPosteriorConcentration2} show
\begin{align}
\KLD&(g_2||m^{\KLD}_{f}(\cdot|y^{\prime}_{1:n_2}))\leq \int \KLD(g_2||f(\cdot;\theta_2))\pi^{\KLD}(\theta_2|y^{\prime}_{1:n_2})d\eta\nonumber\\
\leq& \int \KLD(g_2||f(\cdot;\theta_1))\pi^{\KLD}(\theta_1|y_{1:n_1})d\theta_1+ \frac{1}{c_{\mathcal{S}^{(1)}}}\nonumber\\
=& \int \left( \KLD(g_1||f(\cdot;\theta_1)) + \int g_2\log g_2 - g_1 \log g_1d\mu + \int (g_1 - g_2)\log f(\cdot;\theta_1)d\mu\right)\pi^{\KLD}(\theta_1|y_{1:n_1})d\theta_1\nonumber\\
&+ \frac{1}{c_{\mathcal{S}^{(1)}}}\nonumber\\
=& \int \int (g_1 - g_2)\log f(\cdot;\theta_1)d\mu\pi^{\KLD}(\theta_1|y_{1:n_1})d\theta_1+\int \KLD(g_1||f(\cdot;\theta_1))\pi^{\KLD}(\theta_1|y_{1:n_1})d\theta_1\nonumber\\
&\quad + \int g_2\log g_2 - g_1 \log g_1d\mu + \frac{1}{c_{\mathcal{S}^{(1)}}} +\KLD(g_1||m^{\KLD}_{f}(\cdot|y_{1:n_1}))-\KLD(g_1||m^{\KLD}_{f}(\cdot|y_{1:n_1})).\nonumber
\end{align}

Combining the above two results provides the following bound,
\begin{align}
|\KLD(g||m^{\KLD}_{f}(\cdot|y))-
\KLD(g||m^{\KLD}_{h}(\cdot|y))|&\leq C^{\KLD}(f,y_{1:n_1}, y^{\prime}_{1:n_2}) + \frac{1}{c} + T_1(g_1, g_2) + T_2(f,y_{1:n_1}, y^{\prime}_{1:n_2}) \nonumber
\end{align}
where $c:= \min\{c_{\mathcal{S}^{(1)}}, c_{\mathcal{S}^{(2)}}\}$ and
\begin{align}
T_1(g_1, g_2):&= \max\left\lbrace \int g_2\log g_2 - g_1 \log g_1d\mu, \int g_1\log g_1 - g_2\log g_2d\mu\right\rbrace\nonumber\\
T_2(f,y_{1:n_1}, y^{\prime}_{1:n_2}):&= \max \left\lbrace\int \int (g_1 - g_2)\log f(\cdot;\theta_1)d\mu\pi^{\KLD}(\theta_1|y_{1:n_1})d\theta_1,\right.\nonumber\\
&\qquad\left.\int \int (g_2 - g_1)\log f(\cdot;\theta_2)d\mu\pi^{\KLD}(\theta_2|y^{\prime}_{1:n_2})d\theta_2 \right\rbrace\nonumber \\
C^{\KLD}(f,y_{1:n_1}, y^{\prime}_{1:n_2}):&= \max \left\lbrace\int\KLD(g_1||f(\cdot;\theta_1))\pi^{\KLD}(\theta_1|y_{1:n_1})d\theta_1-\KLD(g_1||m^{\KLD}_{f}(\cdot|y_{1:n_1})),\right.\nonumber\\
&\qquad\left.\int\KLD(g_2||f(\cdot;\theta_2))\pi^{\KLD}(\theta_2|y^{\prime}_{1:n_2})d\theta_2-\KLD(g_2||m^{\KLD}_{f}(\cdot|y^{\prime}_{1:n_2})) \right\rbrace \nonumber
\end{align}
as required.

\end{proof}

\subsubsection{Proof of Theorem 3}

Theorem 3 uses the convexity of the \textBD (Lemma \ref{Lem:betaDivConvexity}) and the second triangular relationship between the \textBD and the \TVD (Lemma \ref{Lem:betaDivTVDTriangle_g1g2}) to prove stability in the posterior predictive approximation to the \DGP of inference using the \textBD.


\begin{proof}
Using the convexity of the \textBD (Lemma \ref{Lem:betaDivConvexity}) and Jensen's inequality, 
\begin{align}
\BD(g_1||&m^{(\beta)}_{f}(\cdot|y_{1:n_1})) \leq \int \BD(g_1||f(\cdot;\theta_1))\pi^{(\beta)}(\theta_1|y_{1:n_1})d\theta_1\nonumber
\end{align}
Now, adding and subtracting $\int\BD(g_1||f(\cdot;\theta_2))\pi^{(\beta)}(\theta_2|y^{\prime}_{1:n:2})d\theta_2$ provides
\begin{align}
&\BD(g_1||m^{(\beta)}_{f}(\cdot|y_{1:n_1})) \leq \int \BD(g_1||f(\cdot;\theta_1))\pi^{(\beta)}(\theta_1|y_{1:n_1})d\theta_1\nonumber\\
=& \int \BD(g_1||f(\cdot;\theta_1))\pi^{(\beta)}(\theta_1|y_{1:n_1})d\theta_1\nonumber\\
&+ \int\BD(g_1||f(\cdot;\theta_2))\pi^{(\beta)}(\theta_2|y^{\prime}_{1:n:2})d\theta_2 - \int\BD(g_1||f(\cdot;\theta_2))\pi^{(\beta)}(\theta_2|y^{\prime}_{1:n:2})d\theta_2\nonumber\\
=& \int \BD(g_1||f(\cdot;\theta_2))\pi^{(\beta)}(\theta_2|y^{\prime}_{1:n:2})d\theta_2\nonumber\\
&+ \int\int\left\{\BD(g_1||f(\cdot;\theta_1)) - \BD(g_1||f(\cdot;\theta_2))\right\}\pi^{(\beta)}(\theta_1|y_{1:n_1})d\theta_1\pi^{(\beta)}(\theta_2|y^{\prime}_{1:n:2})d\theta_2\nonumber
\end{align}
%
%
%
Now we can apply Lemma \ref{Lem:StochasticCDFResult} to random variable $\left\{\BD(g_1||f(\cdot;\theta_1)) - \BD(g_1||f(\cdot;\theta_2))\right\}\in\mathbb{R}$ on $\Theta\times\mathcal{A}$ which by using \eqref{Equ:StochasticCondition_g1g2_S2} of Condition \ref{Cond:StochasticPosteriorConcentration2} applied to the \textBD provides 
\begin{align}
    \int\int \left\{\BD(g_1||f(\cdot;\theta_1)) - \BD(g_1||f(\cdot;\theta_2))\right\}\pi^{(\beta)}(\theta_1|y_1)d\theta_1\pi^{(\beta)}(\theta_2|y_2)d\theta_2\nonumber \leq  \frac{1}{c_{\mathcal{S}^{(2)}}}.\nonumber
\end{align}
We can now use the triangular-type relationship between the \textBD and the \TVD (Lemma \ref{Lem:betaDivTVDTriangle_g1g2}) to show that
\begin{align}
\BD(g_1||&m^{(\beta)}_{f}(\cdot|y_{1:n_1})) \leq  \int \BD(g_1||f(\cdot;\theta_2))\pi^{(\beta)}(\theta_2|y^{\prime}_{1:n:2})d\theta_2 + \frac{1}{c_{\mathcal{S}^{(2)}}}\nonumber\\
&\leq \int \left(\frac{M^{\beta-1}(\beta + 2)}{\beta(\beta-1)}\TVD(g_1,g_2)+\BD(g_2||f(\cdot;\theta_2))\right)\pi^{(\beta)}(\theta_2|y^{\prime}_{1:n_2})d\theta_2 + \frac{1}{c_{\mathcal{S}^{(2)}}}\nonumber\\
&= \frac{M^{\beta-1}(\beta + 2)}{\beta(\beta-1)}\TVD(g_1,g_2) + \int\BD(g_2||f(\cdot;\theta_2))\pi^{(\beta)}(\theta_2|y^{\prime}_{1:n_2})d\theta_2\nonumber\\
&\quad  + \frac{1}{c_{\mathcal{S}^{(2)}}}+\BD(g_2||m^{(\beta)}_{f}(\cdot|y^{\prime}_{1:n_2}))-\BD(g_2||m^{(\beta)}_{f}(\cdot|y^{\prime}_{1:n_2}))\nonumber
\end{align}
and the same arguments applying Lemma \ref{Lem:StochasticCDFResult} and \eqref{Equ:StochasticCondition_g1g2_S1} of Condition \ref{Cond:StochasticPosteriorConcentration2} show
\begin{align}
\BD(g_2||&m^{(\beta)}_{f}(\cdot|y^{\prime}_{1:n_2})) \leq \int \BD(g_2||f(\cdot;\theta_2))\pi^{(\beta)}(\theta_2|y^{\prime}_{1:n_2})d\theta_2\nonumber\\
&\leq \int \BD(g_2||f(\cdot;\theta_1))\pi^{(\beta)}(\theta_1|y_{1:n_1})d\theta_1 + \frac{1}{c_{\mathcal{S}^{(1)}}}\nonumber\\
&\leq \int \left(\frac{M^{\beta-1}(\beta + 2)}{\beta(\beta-1)}\TVD(g_1,g_2)+\BD(g_1||f(\cdot;\theta_1))\right)\pi^{(\beta)}(\theta_1|y)d\theta_1 + \frac{1}{c_{\mathcal{S}^{(1)}}}\nonumber\\
&= \frac{M^{\beta-1}(\beta + 2)}{\beta(\beta-1)}\TVD(g_1,g_2)+\int\BD(g_1||f(\cdot;\theta_1))\pi^{(\beta)}(\theta_1|y_{1:n_1})d\theta_1\nonumber\\
&\quad + \frac{1}{c_{\mathcal{S}^{(1)}}} + \BD(g_1||m^{(\beta)}_{f}(\cdot|y_{1:n_1}))-\BD(g_1||m^{(\beta)}_{f}(\cdot|y_{1:n_1}))\nonumber
\end{align}

Combining the above two results provides the following bound,
\begin{align}
|\BD(g_1||m^{(\beta)}_{f}(\cdot|y_{1:n_1}))-
\BD(g_2||m^{(\beta)}_{f}(\cdot|y^{\prime}_{1:n_2}))|&\leq \frac{M^{\beta-1}(\beta + 2)}{\beta(\beta-1)}\epsilon^{\prime} + \frac{1}{c} + C^{(\beta)}(f,y_{1:n_1}, y^{\prime}_{1:n_2}), \nonumber 
\end{align}
where $c:= \min\{c_{\mathcal{S}^{(1)}}, c_{\mathcal{S}^{(2)}}\}$ and
\begin{align}
C^{(\beta)}(f,y_{1:n_1}, y^{\prime}_{1:n_2}):&= \max \left\lbrace\int\BD(g_1||f(\cdot;\theta_1))\pi^{(\beta)}(\theta_1|y_{1:n_1})d\theta_1-\BD(g_1||m^{(\beta)}_{f}(\cdot|y_{1:n_1})),\right.\nonumber\\
&\qquad\left.\int\BD(g_2||f(\cdot;\theta_2))\pi^{(\beta)}(\theta_2|y^{\prime}_{1:n_2})d\theta_2-\BD(g_2||m^{(\beta)}_{f}(\cdot|y^{\prime}_{1:n_2})) \right\rbrace \nonumber
\end{align}
as required.

\end{proof}

\subsubsection{Proof of Theorem 4}

Similarly to Theorem 2, we use the convexity of the \textBD (Lemma \ref{Lem:betaDivConvexity})
and the three-point property the \textBD (Lemmas \ref{Lem:betaDivTriangle} and \ref{Lem:remainder_TVD_bound}) to prove Theorem 3 which establishes the posterior predictive stability to perturbations of the \DGP provided by inference using the \textBD.

Lemma \ref{Lem:remainder_TVD_bound} is the important lemma for this proof rather than the triangle inequality

\begin{proof}
By the convexity of the \textBD for $1< \beta\leq 2$ (Lemma \ref{Lem:betaDivConvexity}) we can apply Jensen's inequality to show that 
\begin{align}
&\BD(m^{(\beta)}_{f}(\cdot|y_1)||m^{(\beta)}_{f}(\cdot|y_2)))\leq \int \BD(m^{(\beta)}_{f}(\cdot|y_1)||f(\cdot;\theta_2))\pi^{(\beta)}(\theta_2|y_2)d\theta_2\nonumber\\
\leq& \int \left\lbrace\int \BD(f(\cdot;\theta_1)||f(\cdot;\theta_2))\pi^{(\beta)}(\theta_1|y_1)d\theta_1\right\rbrace\pi^{(\beta)}(\theta_2|y_2)d\theta_2.\nonumber
\end{align}
%
%
Now the three-point property associated with the \textBD (Lemma \ref{Lem:betaDivTriangle}) gives us that 
\begin{equation}
\BD(f_1||f_2)=\BD(g_2||f_2)-\BD(g_2||f_1)+R(g_2||f_1||f_2)\nonumber
\end{equation}
where $R(g||f||h)$ is defined in  \eqref{Equ:Rghf}.
Using this here provides
\begin{align}
\BD(m^{(\beta)}_{f}(\cdot|y_1)||&m^{(\beta)}_{f}(\cdot|y_2)))\nonumber\\
\leq& \int \left\lbrace\int \BD(f(\cdot;\theta_1)||f(\cdot;\theta_2))\pi^{(\beta)}(\theta_1|y_1)d\theta_1\right\rbrace\pi^{(\beta)}(\theta_2|y_2)d\theta_2\nonumber\\
=& \int \left\lbrace\int \left(\BD(g_2||f(\cdot;\theta_2))-\BD(g_2||f(\cdot;\theta_1))\right.\right.\nonumber\\
&\left.\left.+R(g_2||f(\cdot;\theta_1)||f(\cdot;\theta_2)\right)\pi^{(\beta)}(\theta_1|y_1)d\theta_1\right\rbrace\pi^{(\beta)}(\theta_2|y_2)d\theta_2\nonumber\\
=& \int\int \left(\BD(g_2||f(\cdot;\theta_2))- \BD(g_2||f(\cdot;\theta_1))\right)\pi^{(\beta)}(\theta_1|y_1)d\theta_2\pi^{(\beta)}(\theta_2|y_2)d\theta_1\nonumber\\
&+\int\int R(g_2||f(\cdot;\theta_1)||f(\cdot;\theta_2))\pi^{(\beta)}(\theta_1|y_1)d\theta_1\pi^{(\beta)}(\theta_2|y_2)d\theta_2.\nonumber
\end{align}
Now we can apply Lemma \ref{Lem:StochasticCDFResult} to random variable $\left\{\BD(g_2||f(\cdot;\theta_2)) - \BD(g_2||f(\cdot;\theta_1))\right\}\in\mathbb{R}$ on $\Theta\times\mathcal{A}$ which by using \eqref{Equ:StochasticCondition_g1g2_S1} of Condition \ref{Cond:StochasticPosteriorConcentration2} applied to the \textBD provides 
\begin{align}
    \int\int \left\{\BD(g_2||f(\cdot;\theta_2)) - \BD(g_2||f(\cdot;\theta_1))\right\}\pi^{(\beta)}(\theta_1|y_1)d\theta_1\pi^{(\beta)}(\theta_2|y_2)d\theta_2\nonumber \leq  \frac{1}{c_{\mathcal{S}^{(1)}}}.\nonumber
\end{align}
Therefore
\begin{align}
\BD(m^{(\beta)}_{f}(\cdot|y_1)||&m^{(\beta)}_{f}(\cdot|y_2)))\nonumber\\
\leq& \int\int \left(\BD(g_2||f(\cdot;\theta_2))- \BD(g_2||f(\cdot;\theta_1))\right)\pi^{(\beta)}(\theta_1|y_1)d\theta_2\pi^{(\beta)}(\theta_2|y_2)d\theta_1\nonumber\\
&+\int\int R(g_2||f(\cdot;\theta_1)||f(\cdot;\theta_2))\pi^{(\beta)}(\theta_1|y_1)d\theta_1\pi^{(\beta)}(\theta_2|y_2)d\theta_2\nonumber\\
\leq& \frac{1}{c_{\mathcal{S}^{(1)}}} +\int\int R(g_2||f(\cdot;\theta_1)||f(\cdot;\theta_2))\pi^{(\beta)}(\theta_1|y_1)d\theta_1\pi^{(\beta)}(\theta_2|y_2)d\theta_2.\nonumber
\end{align}
Now from Lemma \ref{Lem:remainder_TVD_bound} we have that $    R(g_2||f(\cdot;\theta_1)||f(\cdot;\theta_2))\leq 2\frac{M^{\beta - 1}}{\beta-1}\TVD(g_2, f(\cdot;\theta_1))$ which is itself not necessarily small. We can however apply the triangle inequality to the \TVD and say that 
\begin{align}
R(g_2||f(\cdot;\theta_1)||f(\cdot;\theta_2))&\leq 2\frac{M^{\beta - 1}}{\beta-1}\TVD(g_2, f(\cdot;\theta_1))\nonumber  \\
&\leq 2\frac{M^{\beta - 1}}{\beta-1}\left(\TVD(g_1, g_2) + \TVD(g_1, f(\cdot;\theta_1))\right)\nonumber\\
&\leq 2\frac{M^{\beta - 1}}{\beta-1}\left(\epsilon + \TVD(g_1, f(\cdot;\theta_1))\right),\nonumber
\end{align}
given the neighbourhood of data generating processes defined by $\mathcal{G}_{\epsilon}^{\TVD}$. As a result
\begin{align}
\BD(m^{(\beta)}_{f}(\cdot|y_1)||m^{(\beta)}_{f}(\cdot|y_2)))\leq& 2\frac{M^{\beta - 1}}{\beta-1}\epsilon + \frac{1}{c_{\mathcal{S}^{(1)}}} + 2\frac{M^{\beta - 1}}{\beta-1}\int \TVD(g_1, f(\cdot;\theta_1))\pi^{(\beta)}(\theta_1|y_1)d\theta_1.\nonumber    
\end{align}

We note that we could have instead considered $\BD(m^{(\beta)}_{f}(\cdot|y_2)||m^{(\beta)}_{f}(\cdot|y_1)))$, applied the corresponding version of the three-point property of \textBD, used Lemma \ref{Lem:StochasticCDFResult}, \eqref{Equ:StochasticCondition_g1g2_S2} of Condition \ref{Cond:StochasticPosteriorConcentration2} applied to the \textBD and Lemma \ref{Lem:remainder_TVD_bound} to also show that
\begin{align}
\BD(m^{(\beta)}_{f}(\cdot|y_2)||m^{(\beta)}_{f}(\cdot|y_1))) &\leq 2\frac{M^{\beta - 1}}{\beta-1}\epsilon + \frac{1}{c_{\mathcal{S}^{(2)}}} + 2\frac{M^{\beta - 1}}{\beta-1}\int \TVD(g_2, f(\cdot;\theta_2))\pi^{(\beta)}(\theta_2|y_2)d\theta_2.\nonumber
\end{align}
Which proves the required result.
\end{proof}

\color{black}

\subsection{An Alternative Neighbourhood Definition}{\label{Sub:NeighbourhoodDefinition}}

An alternative neighbourhood for likelihood models of Definition 1 and data generating processes of Definition 2 could be to replace the \TVD with the \KLD. We show here, using an example, that such a neighbourhood would be prohibitively strict to provide useful stability to a \DM.

Consider $f(x;\theta) = f_1(x; \theta_1)$ and $h(x;\eta) = (1-\epsilon) f_1(x;\theta_1) + \epsilon f_2(x; \theta_2)$. The \TVD between $f$ and $h$ is less than $\epsilon$ independent of the form of $f_1$ and $f_2$ as
\begin{align*}
    \TVD(h(x;\eta), f(x; \theta))=& \frac{1}{2}\int \left|f_1(x; \theta_1) - (1-\epsilon)f_1(x;\theta_1) - \epsilon f_2(x;\theta_2)\right|dx\\
    \leq& \frac{1}{2}(1-\epsilon)\int \left|f_1(x; \theta_1) - f_1(x;\theta_1)\right|dx + \frac{1}{2}\epsilon \int \left| f_1(x; \theta_1) - f_2(x;\theta_2)\right|dx\\
    =&\frac{1}{2}\epsilon \int \left| f_1(x; \theta_1) - f_2(x;\theta_2)\right|dx\\
    \leq&\epsilon.
\end{align*}
On the other hand, when $f_1 = \mathcal{N}(x; 0, 1)$ and $f_2 = \mathcal{N}(x; \mu_c, 1)$, the bounds of \cite{durrieu2012lower} can be used to show that
\begin{align*}
    \KLD(h(x;\eta) || f(x; \theta) &\geq (1-\epsilon)\log\left\{(1-\epsilon) + \epsilon \exp\left(-\frac{1}{2}\mu_c^2\right)\right\}\\
    &+\epsilon\log\left\{(1-\epsilon)\exp\left(-\frac{1}{4}\mu_c^2\right) + \epsilon \exp\left(\frac{1}{4}\mu_c^2\right) \right\}\\
    &+\frac{1}{2}\log(2/e)
\end{align*}
and therefore that $\KLD(h(x;\eta) || f(x; \theta) = O(\sqrt{n})$ even when $\epsilon = O(1/\sqrt{n})$ if $\mu_c = O(\sqrt{n})$. Therefore, a very small contamination can lead to a large \KLD distance but not a large \TVD distance.

While we argue that the \TVD is practical and interpretable, one limitation is that it is not available in closed form, even between standard families. However, recent work from \cite{devroye2018total} provides bounds on the \TVD between two multivariate Gaussian distributions which may prove useful for interpreting the \TVD in terms of distributions most users are familiar with. For example, in the 1-dimensional case, the result of \cite{devroye2018total} becomes
\begin{align*}
        \frac{1}{100}\min\left\{1, \left(\frac{\eta}{\sigma^2}\right)^2\right\} \leq TVD(\mathcal{N}(\mu, \sigma^2), \mathcal{N}(\mu, \sigma^2 + \eta)) \leq 1.5 \min\left\{1, \left(\frac{\eta}{\sigma^2}\right)^2\right\}
\end{align*}
and therefore a \TVD bound of $\epsilon = 0.05$ can be shown to be equivalent to $\frac{\eta}{\sigma^2} \approx 0.18$ i.e. an approximately 20\% increase in the variance of the Gaussian.

\color{black}

\section{Extended Experimental Results}{\label{Sec:ExtExperimentalResults}}

This section contains additional details and results for the experiments of Section 6.

\subsection{Gaussian and \Student likelihood}

\subsubsection{Posteriors}

Figure \ref{Fig:norm_t_posteriors} plots the posterior distribution of model parameters $\mu$ and $\sigma^2$ of the Gaussian and \Student models (8) discussed in Section 6.1.

\begin{figure}[!ht]
\begin{center}
\includegraphics[trim= {0.0cm 0.00cm 0.0cm 0.0cm}, clip,  
width=0.49\columnwidth]{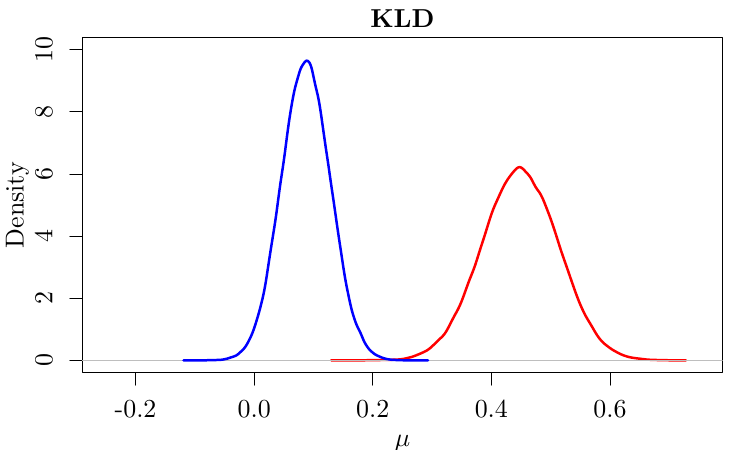}
\includegraphics[trim= {0.0cm 0.00cm 0.0cm 0.0cm}, clip,  
width=0.49\columnwidth]{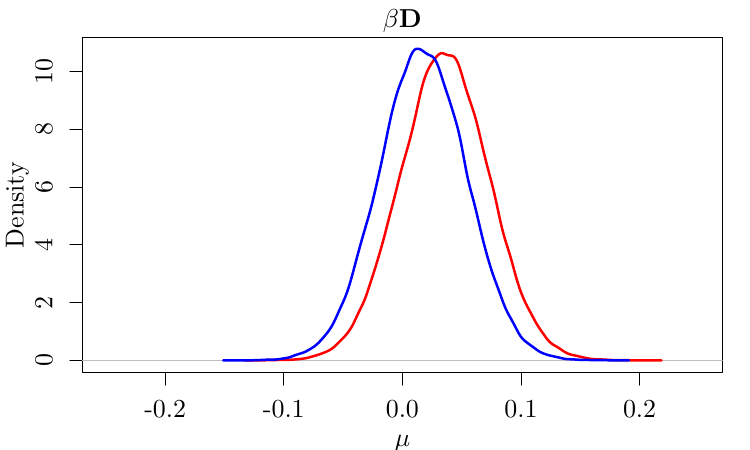}\\
\includegraphics[trim= {0.0cm 0.00cm 0.0cm 0.0cm}, clip,  
width=0.49\columnwidth]{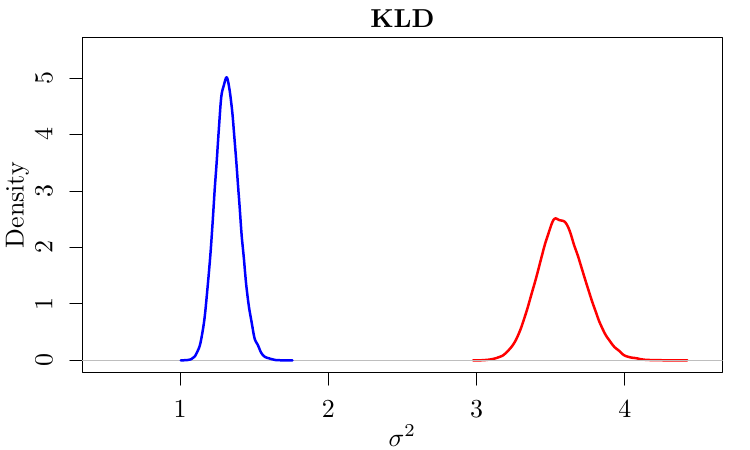}
\includegraphics[trim= {0.0cm 0.00cm 0.0cm 0.0cm}, clip,  
width=0.49\columnwidth]{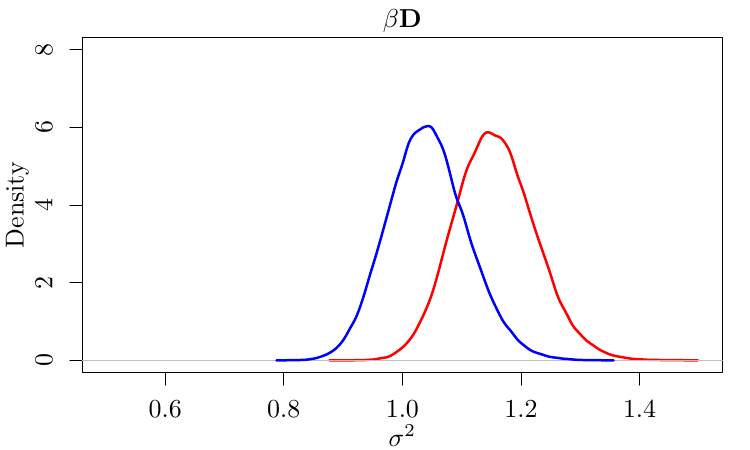}
\caption{Parameter posterior distributions for $\mu$ and $\sigma^2$ under \bayesrule updating (\KLD-Bayes) (\textbf{left}) and \textBD-Bayes with $\beta = 1.22$ (\textbf{right}) under the likelihood functions  $f(y;\theta)=\mathcal{N}\left(y;\mu,\sigma_{adj}^2\sigma^2\right)$ (\textbf{{\color{black}{red}}}) and $h(y;\eta)=t_{\nu}(y;\mu,\sigma^2)$ ({\textbf{\color{black}{blue}}}) where $\nu=5$ and $\sigma_{adj}^2=1.16$.}
\label{Fig:norm_t_posteriors}
\end{center}
\end{figure}

The left-hand side of Figure \ref{Fig:norm_t_posteriors} demonstrates what most statistical practitioners expect when comparing the performance of a Gaussian and a Student's-$t$ under outlier contamination \citep{o1979outlier}. Under the Student's-$t$ likelihood, the inference is much less affected by the outlying contamination than under the Gaussian likelihood. The parameter $\mu$ is shifted less towards the contaminant population and the parameter $\sigma^2$ is inflated much less. In short, very different inferences are produced using a Student's-$t$ and a Gaussian under outlier contamination. Updating using the \textBD-Bayes presents a striking juxtaposition to this. The \textBD-Bayes produces almost identical posteriors for both $\mu$ and $\sigma^2$ under both models resulting in almost identical posterior predictive densities in Figure 1.

Estimating the \TVD or the \textBD between the two predictive distributions is hampered by the fact that they are not available in closed form. However, the energy distance \cite{szekely2013energy} provides a metric that can be easily estimated from samples of the predictive. The energy distance between the Gaussian and \Student predictive distributions under traditional Bayesian updating was $0.125$, while under \textBD-Bayes updating the energy distance was $2.13 \times 10 ^{-3}$.  
%

\subsubsection{Sensitivity analysis}{\label{Sec:Sensitivity}}

As was noted in Section 5, it is encouraging to note the stability of the \textBD-Bayes inference appears not to be overly sensitive to the exact value of $\beta$. 

Complimenting Figure 2, the left-hand side of Figure \ref{Fig:beta_sensitivity} plots $\frac{M^{\beta-1}(3\beta - 2)}{\beta(\beta-1)}$, the multiplier of the \TVD from Theorems 1 and 2, as a function of $\beta$ for various $M$. We see that as $\beta$ increases away from 1 this multiplier initially decreases rapidly, indicating a large increase in guaranteed stability by moving away from the \KLD. However, after this point the multiplier plateaus, indicating that a similar amount of stability results from a range of values of $\beta$. The right-hand plot of Figure \ref{Fig:beta_sensitivity} has a very similar shape. Here we plot the energy distance \citep{szekely2013energy} between the posterior predictives of fitting a Gaussian likelihood model and  \Student likelihood modes, used in Figure 1 for different values of $\beta$. 
Once again, we see that taking $\beta>1$ results in a large increase in a posteriori stability but that after a point that stability achieved is fairly constant with $\beta$.

\begin{figure}[t!]
\begin{center}
\includegraphics[trim= {0.0cm 0.00cm 0.0cm 0.0cm}, clip,  
width=0.49\columnwidth]{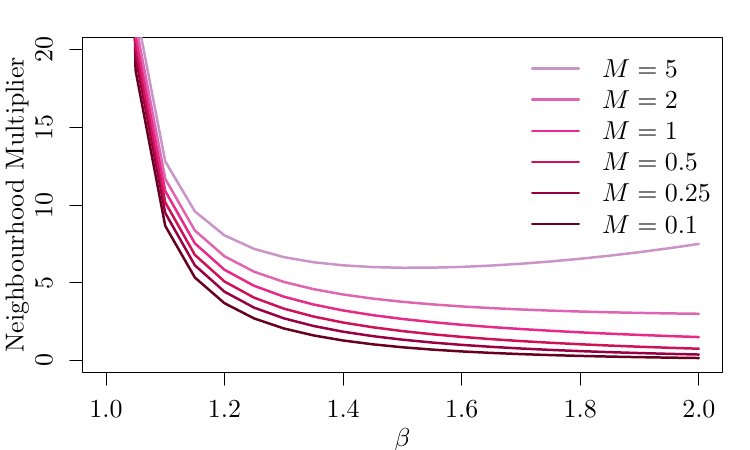}
\includegraphics[trim= {0.0cm 0.00cm 0.0cm 0.0cm}, clip,  
width=0.49\columnwidth]{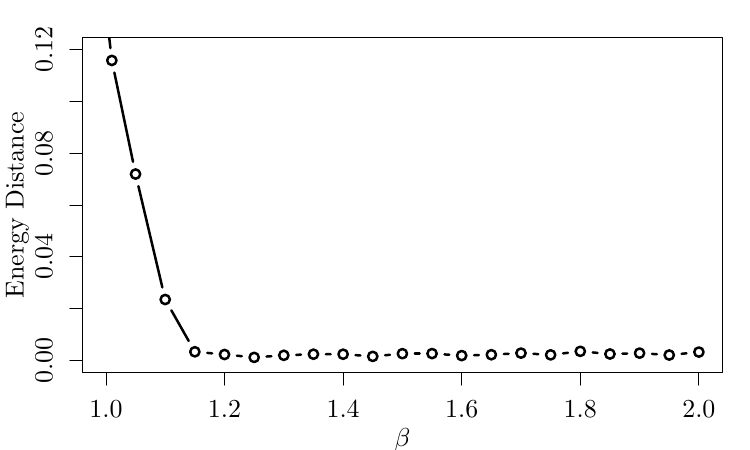}
\caption{\textbf{Left:} The  multiplier $\frac{M^{\beta-1}(3\beta - 2)}{\beta(\beta-1)}$ from from Theorems 1 and 2 for different values of $M$ and $1<\beta<2$. \textbf{Right:} The energy distance between the posterior predictives after fitting a Gaussian and a \Student model as in Figure 1 under the \textBD-Bayes for different $\beta$.  
}
\label{Fig:beta_sensitivity}
\end{center}
\end{figure}

\subsubsection{Comparison with the $\gamma$-divergence}{\label{App:gammaD}}

Similarly to the \textBD, the $\gamma$-divergence (\textGD) provides a loss function that does not require an estimator of the underlying density, and has been shown to have good robustness properties \cite{hung2018robust, knoblauch2022generalized}. Here, we show its stability performance appears comparable with the \textBD. Firstly the \textGD is defined as follows.


\begin{definition}[The $\gamma$-divergence (\GD)  \citep{fujisawa2008robust,hung2018robust}]
The \textGD  is defined as
\begin{align}
\GD(g||f) = \frac{1}{(\gamma - 1)\gamma}\left\lbrace\left(\int g^{\gamma}d\mu\right)^{\frac{1}{\gamma}} - \int\frac{f^{\gamma - 1}}{\left(\int f^{\gamma}d\mu\right)^{\frac{\gamma-1}{\gamma}}}gd\mu\right\rbrace,\nonumber
\end{align}
where $\gamma\in\mathbb{R}\setminus \left\lbrace0,1\right\rbrace$.
\end{definition}

The corresponding loss function allowing generalised Bayesian inference for $\theta_g^{(\gamma)}$ is 
\begin{equation}
\ell^{(\gamma)}(y,f(\cdot;\theta))= -\frac{1}{\gamma-1}f(y;\theta)^{\gamma-1}\cdot\frac{1}{\gamma}\frac{1}{\left(\int f(z;\theta)^{\gamma}dz\right)^{\frac{\gamma-1}{\gamma}}}.\nonumber
\end{equation} 
Similarly to the \textBD-loss in (3), the \textGD-loss raises the likelihood to the power $\gamma - 1$ and `adjusts' by the integral of the likelihood to the power $\gamma$, except for the \textGD this `adjustment' term is multiplicative rather than additive as it was in the \textBD. The integral term is independent of location parameters e.g. $\mu$ from the Gaussian and \Student examples, and therefore inference for these will be very similar under the \textBD and \textGD \citep{fujisawa2008robust}. Figure \ref{Fig:norm_t_posteriors_gamma} shows that in the example introduced in Section 6.1, inference for $\sigma^2$ is also very similar under the \textGD and \textBD and as a result they estimate identical posterior predictives under both the Gaussian and \Student models.

\begin{figure}
\begin{center}
\includegraphics[trim= {0.0cm 0.00cm 0.0cm 0.0cm}, clip,  
width=0.49\columnwidth]{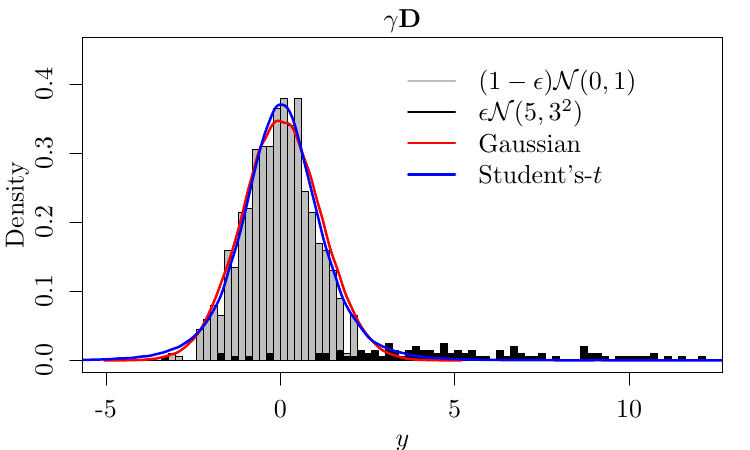}\\
\includegraphics[trim= {0.0cm 0.00cm 0.0cm 0.0cm}, clip,  
width=0.49\columnwidth]{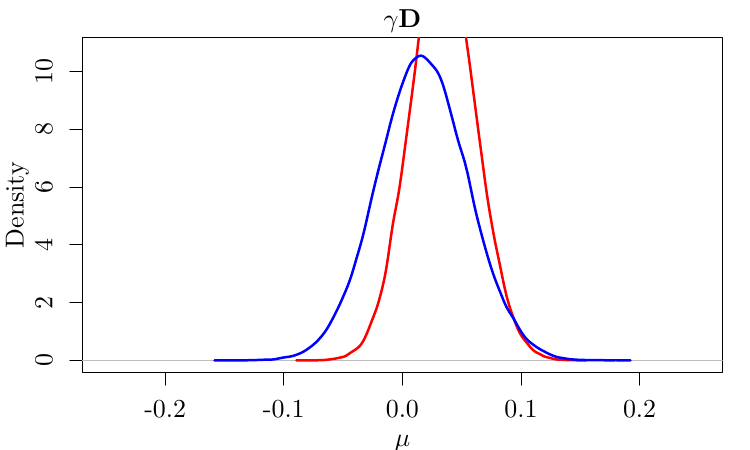}
\includegraphics[trim= {0.0cm 0.00cm 0.0cm 0.0cm}, clip,  
width=0.49\columnwidth]{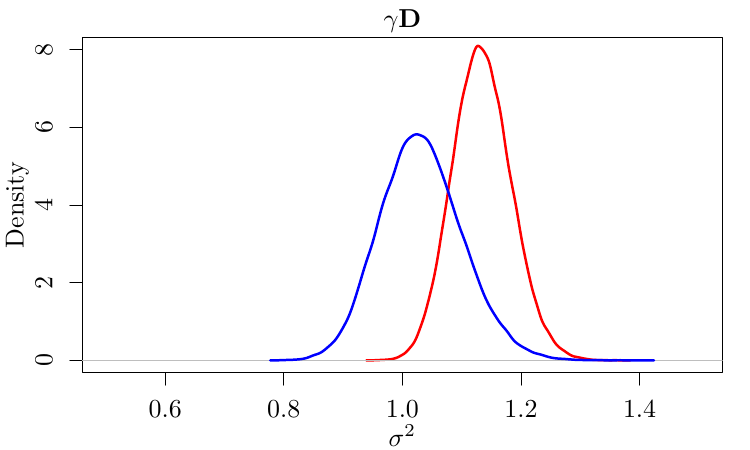}
\caption{Posterior predictive and parameter posterior distributions for $\left(\mu,\sigma^2\right)$ under \textGD-Bayes with $\gamma = 1.22$ and likelihood functions  $f(y;\theta)=\mathcal{N}\left(y;\mu,\sigma_{adj}^2\sigma^2\right)$ (\textbf{{\color{black}{red}}}) and $h(y;\eta)=t_{\nu}(y;\mu,\sigma^2)$ ({\textbf{\color{black}{blue}}}) where $\nu=5$ and $\sigma_{adj}^2=1.16$.}
\label{Fig:norm_t_posteriors_gamma}
\end{center}
\end{figure}

\color{black}
\subsubsection{Stability to an $\epsilon$ contamination of the \DGP}

Following the discussion in Section 4.1, consider a \DM with likelihood model $f(y;\theta)=\mathcal{N}\left(y;\mu,\sigma_{adj}^2\sigma^2\right)$ and \DGPs $g_1(y) = \mathcal{N}(y; 0, 1)$ and $g_2(y) = 0.9\times\mathcal{N}\left(y;0,1\right) + 0.1 \times \mathcal{N}\left(y;5,3^2\right)$. The top of Figure \ref{Fig:eps_contamination_posteriors} demonstrates that despite the absolute difference between the densities of the \DGPs being small everywhere, if this is weighted  by $\log f(y)$ where $f(y) = \mathcal{N}(y; 0, 1)$
we see that the distance becomes large in tail areas of the \DGP's.  The bottom of Figure \ref{Fig:eps_contamination_posteriors} demonstrates that the \KLD-Bayes posterior predictive is substantially different when fitting model $f(y;\theta)$ to \DGPs $g_1$ and $g_2$ while the \textBD-Bayes produces very similar posterior predictive distributions under data from either \DGP.

\begin{figure}
\begin{center}
\includegraphics[trim= {0.0cm 0.00cm 0.0cm 0.0cm}, clip,  
width=0.49\columnwidth]{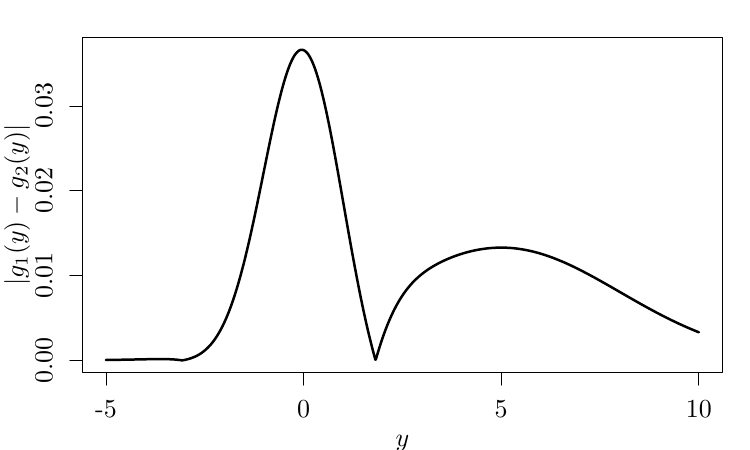}
\includegraphics[trim= {0.0cm 0.00cm 0.0cm 0.0cm}, clip,  
width=0.49\columnwidth]{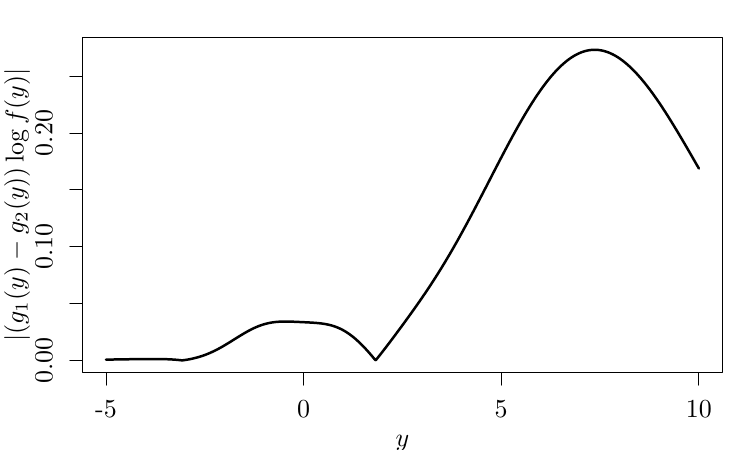}\\
\includegraphics[trim= {0.0cm 0.00cm 0.0cm 0.0cm}, clip,  
width=0.49\columnwidth]{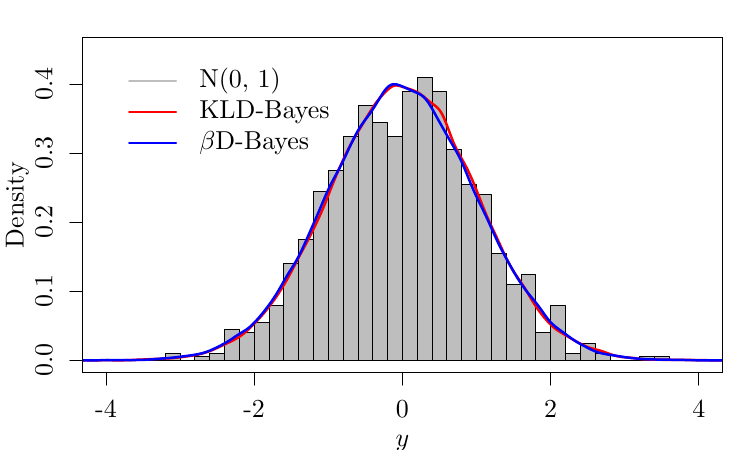}
\includegraphics[trim= {0.0cm 0.00cm 0.0cm 0.0cm}, clip,  
width=0.49\columnwidth]{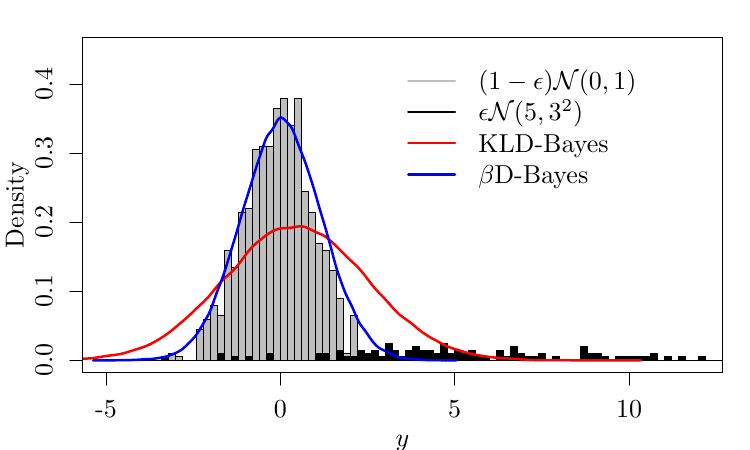}
\caption{\textbf{Top:} Absolute difference between \DGP's $g_1(y) = \mathcal{N}(y; 0, 1)$ and $g_2(y) = 0.9\times\mathcal{N}\left(y;0,1\right) + 0.1 \times \mathcal{N}\left(y;5,3^2\right)$ (\textbf{left}) and the absolute difference weighted by $\log f(y)$ for $f(y) = \mathcal{N}(y, 0, 1)$.
\textbf{Bottom:} Posterior predictive distributions for fitting model $f(y;\theta)=\mathcal{N}\left(y;\mu,\sigma_{adj}^2\sigma^2\right)$ to data rom $g_1(y)$ (\textbf{left}) and $g_2(y)$ (\textbf{right}) under the \KLD-Bayes (\textbf{red}) and \textBD-Bayes with $\beta - 1.22$ (\textbf{blue}).}
\label{Fig:eps_contamination_posteriors}
\end{center}
\end{figure}

\color{black}

\subsubsection{\DLD Data}

For the \DLD data discussed in Section 6.1, we provide additional Q-Q normal and histogram plots. These demonstrate the heavy-tailed nature of the \DLD data and the reasonable fit of the standardised residuals produced by the \textBD-Bayes.

 
 
\begin{figure}
\begin{center}
\includegraphics[trim= {0.0cm 0.00cm 0.0cm 0.0cm}, clip,  
width=0.49\columnwidth]{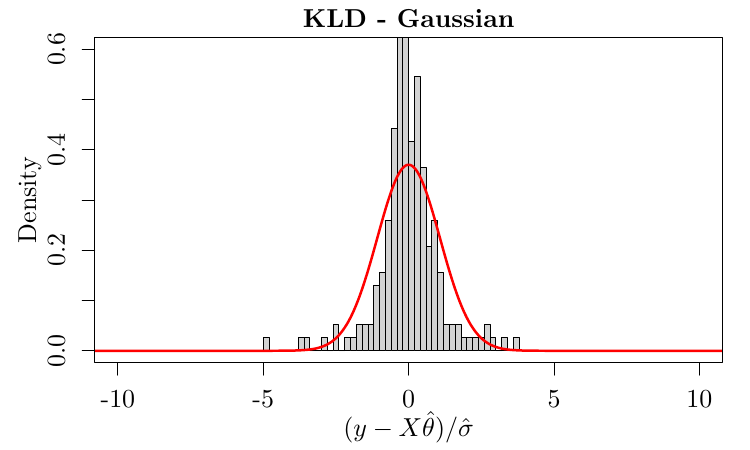}
\includegraphics[trim= {0.0cm 0.00cm 0.0cm 0.0cm}, clip,  
width=0.49\columnwidth]{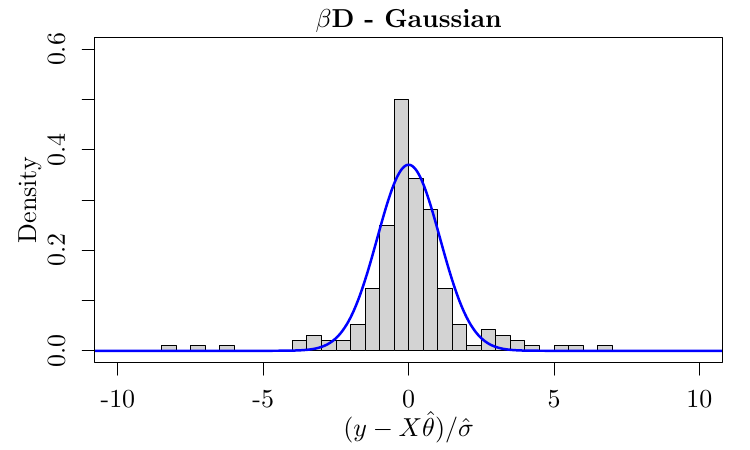}\\
\includegraphics[trim= {0.0cm 0.00cm 0.0cm 0.0cm}, clip,  
width=0.49\columnwidth]{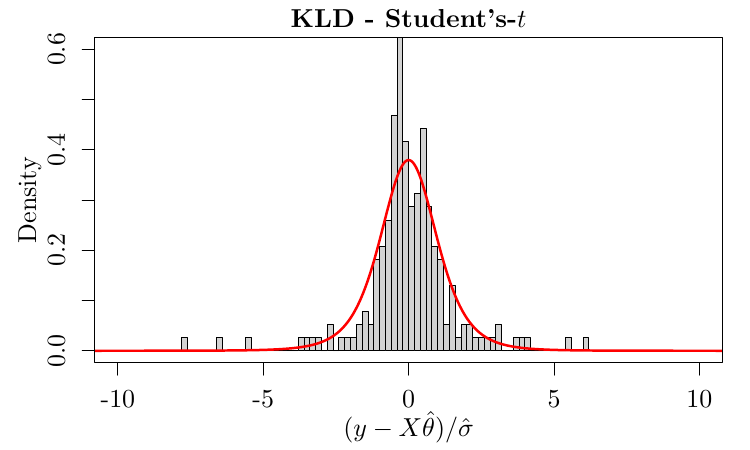}
\includegraphics[trim= {0.0cm 0.00cm 0.0cm 0.0cm}, clip,  
width=0.49\columnwidth]{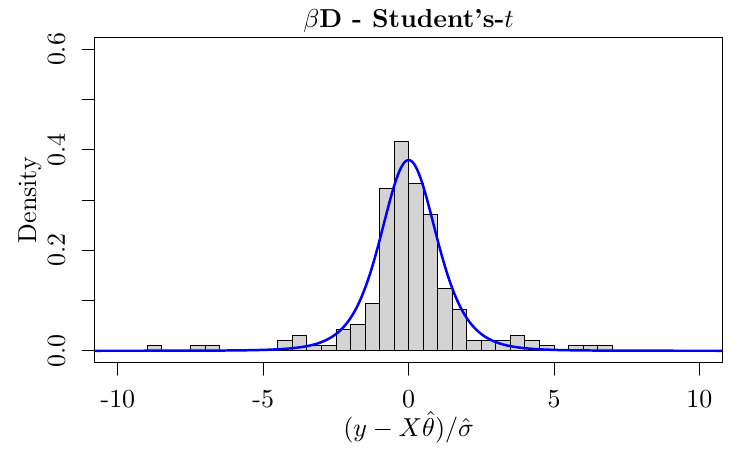}\\
\caption{\DLD data - Posterior mean distribution for standardised residuals under the Gaussian (\textbf{top}) and \Student model (\textbf{bottom}) of \KLD-Bayes (\textbf{left}) and \textBD-Bayes with $\beta = 1.34$ (\textbf{right}).}
\label{Fig:DLDRegressionsHist}
\end{center}
\end{figure} 


\begin{figure}
\begin{center}
\includegraphics[trim= {0.0cm 0.00cm 0.0cm 0.0cm}, clip,  
width=0.49\columnwidth]{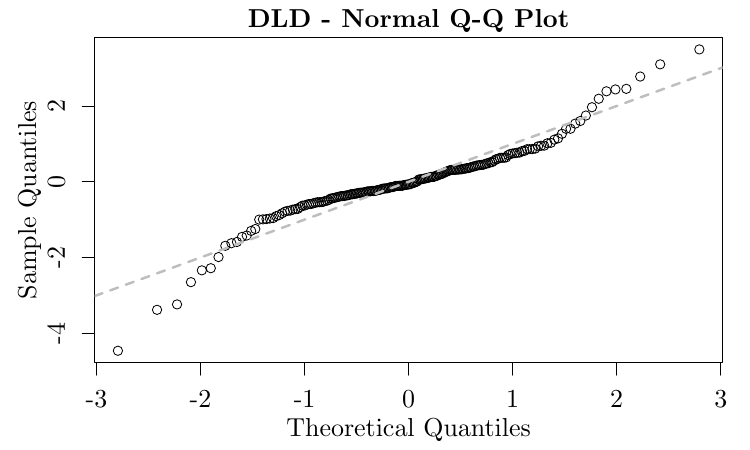}
\includegraphics[trim= {0.0cm 0.00cm 0.0cm 0.0cm}, clip,  
width=0.49\columnwidth]{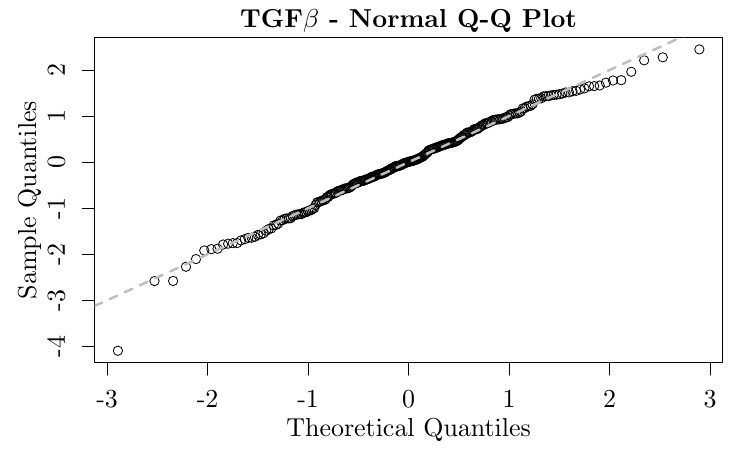}\\
\caption{\textit{Q-Q} normal plot of the fitted residuals according to the Gaussian model under the \KLD-Bayes for the \DLD data (\textbf{left}) and \TGFB data (\textbf{right})}
\label{Fig:QQNormal}
\end{center}
\end{figure}

\subsubsection{\TGFB data}{\label{App:TGFB}}

We consider another regression example to illustrate the stability to the selection between a Gaussian and \Student example when using \textBD-Bayes updating.

The dataset from \cite{calon2012dependency} concerns gene expression data for $n=262$ colon cancer patients. Previous work \citep{rossell2017nonlocal, rossell2018tractable} focused on selecting genes that affect the expression levels of \TGFB, a gene known to play an important role in colon cancer progression. Instead, we study the relation between \TGFB and the 7 genes (listed in Section \ref{App:SelectedVariables}) that appear in the `\TGFB1 pathway' according to the \textit{KEGGREST} package in \R{} \citep{tenenbaum2016keggrest}, so that $p=8$ after including the intercept. We fitted regression models using the neighbouring models in (8) for the residuals.
\color{black}
Once again, we used the method of \cite{yonekura2023adaptation} to set $\beta = 1.03$ when using the Gaussian distribution and use the same value for the Student's-$t$. 
\color{black}

\begin{figure}
\begin{center}
\includegraphics[trim= {0.0cm 0.00cm 0.0cm 0.0cm}, clip,  
width=0.49\columnwidth]{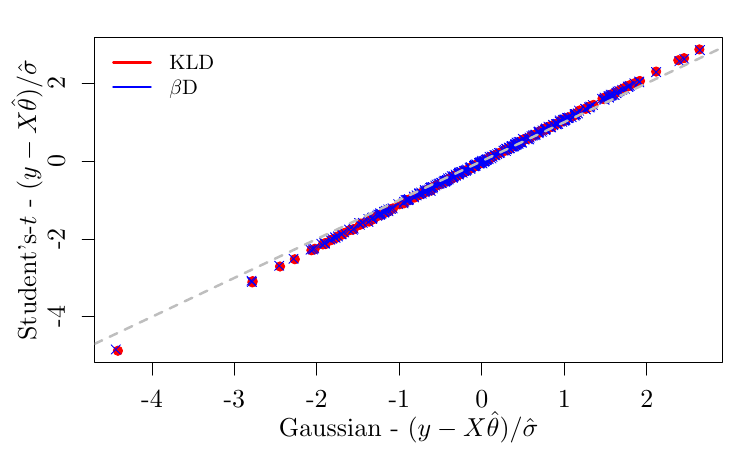}
\includegraphics[trim= {0.0cm 0.00cm 0.0cm 0.0cm}, clip,  
width=0.49\columnwidth]{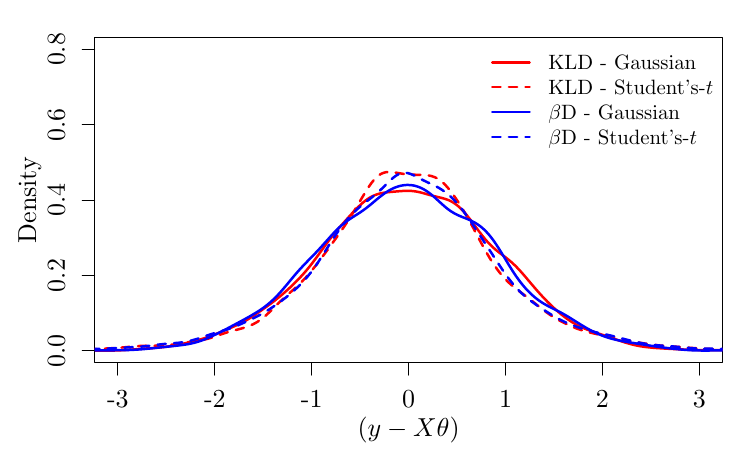}\\
\includegraphics[trim= {0.0cm 0.00cm 0.0cm 0.0cm}, clip,  
width=0.49\columnwidth]{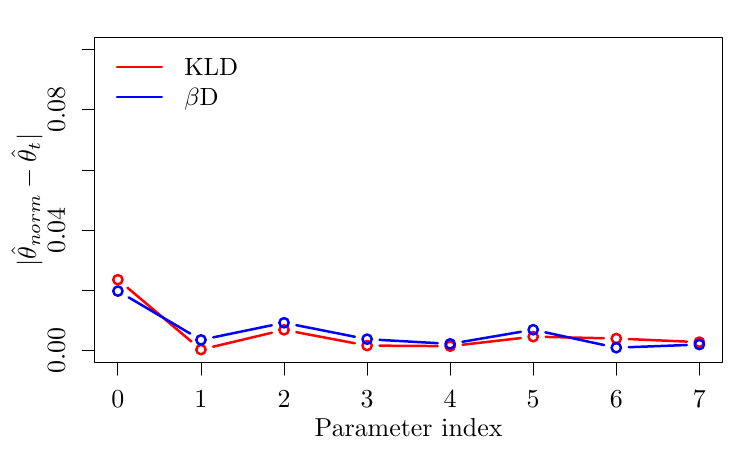}
\includegraphics[trim= {0.0cm 0.00cm 0.0cm 0.0cm}, clip,  
width=0.49\columnwidth]{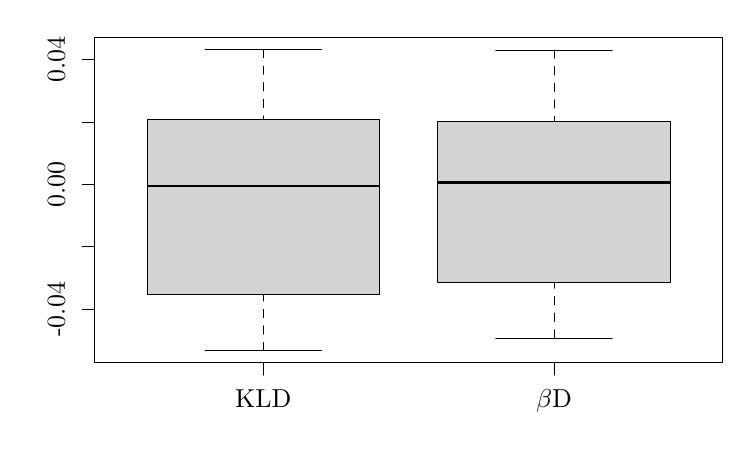}\\
\caption{\TGFB data - Posterior mean estimates of standardised residuals (\textbf{top left}), posterior mean estimated residuals distribution (\textbf{top-right}), absolute difference in posterior mean parameter estimates (\textbf{bottom left}) and difference in posterior predictive densities of the observations (\textbf{bottom right}) under the Gaussian and \Student model of \KLD-Bayes and \textBD-Bayes ($\beta = 1.03$) for the \TGFB data.}
\label{Fig:TGFBRegressions}
\end{center}
\end{figure}

Figure \ref{Fig:TGFBRegressions} shows that both inference procedures are stable to the choice of the model here. The \textBD-Bayes appears to be marginally more stable in estimating the fitted residuals and predictive density (top-left and bottom-right), while the \KLD-Bayes appears marginally more stable when estimating parameters (bottom-left).  Figure \ref{Fig:TGFBRegressionsHist} shows the fit of the models to the standardised residuals under posterior mean estimates.  \cite{rossell2018tractable, jewson2021general} find considerable evidence that a Gaussian model is compatible with this data, further demonstrated in Figure \ref{Fig:QQNormal}. 

 
\begin{figure}
\begin{center}
\includegraphics[trim= {0.0cm 0.00cm 0.0cm 0.0cm}, clip,  
width=0.49\columnwidth]{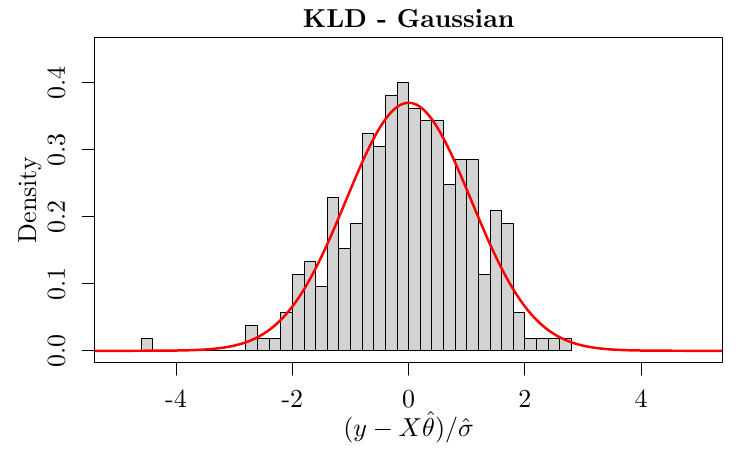}
\includegraphics[trim= {0.0cm 0.00cm 0.0cm 0.0cm}, clip,  
width=0.49\columnwidth]{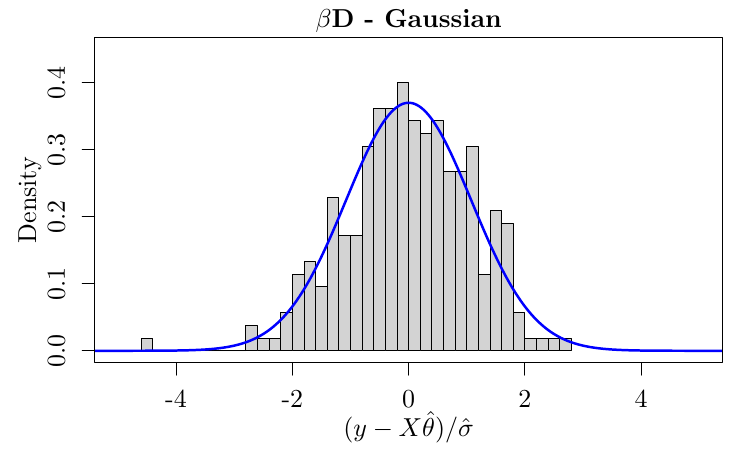}\\
\includegraphics[trim= {0.0cm 0.00cm 0.0cm 0.0cm}, clip,  
width=0.49\columnwidth]{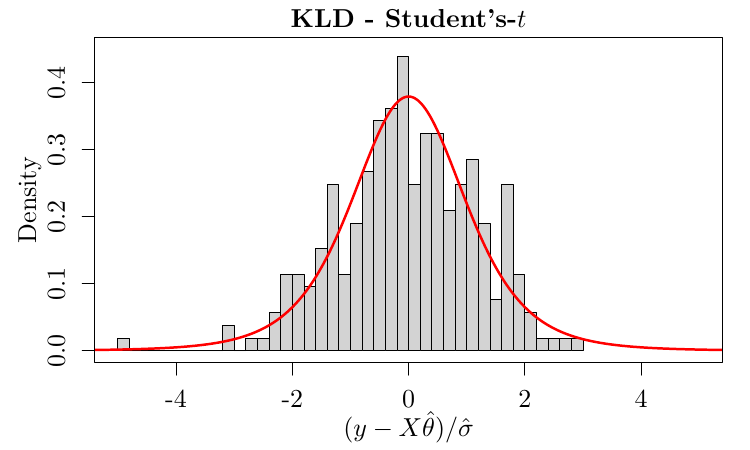}
\includegraphics[trim= {0.0cm 0.00cm 0.0cm 0.0cm}, clip,  
width=0.49\columnwidth]{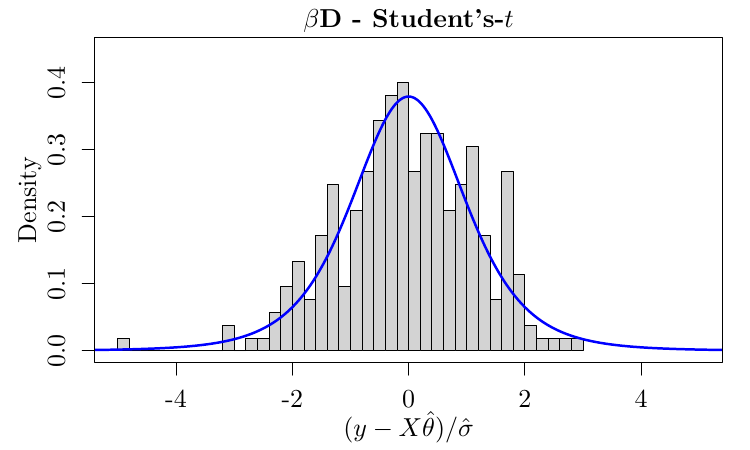}\\
\caption{\TGFB data - Posterior mean distribution for standardised residuals under the Gaussian (\textbf{top}) and \Student model (\textbf{bottom}) of \KLD-Bayes (\textbf{left}) and \textBD-Bayes with $\beta = 1.03$ (\textbf{right}).}
\label{Fig:TGFBRegressionsHist}
\end{center}
\end{figure}

\subsubsection{Variable selection}{\label{App:SelectedVariables}}

When investigating the stability of the \textBD-Bayes inference to the Gaussian and \Student likelihoods we regressed the \DLD and \TGFB gene expressions on a subset of the variables available in the full data sets. The procedures for which variables were selected as outlined in Sections 6.1 and \ref{App:TGFB}. To ensure that our results are reproducible, below we indicate the selected covariates and the supplementary material contains code for these variable pre-screening steps.

\paragraph{\DLD:}

For the \DLD analysis, we selected the 15 genes with the 5 highest loadings in the first 3 principal components of the original 57 predictors. This procedure selected the following genes C15orf52, BRAT1, CYP26C1, SLC35B4, GRLF1, RXRA, RAB3GAP2, NOTCH2NL, SDC4, TTC22, PTCH2, ECH1, CSF2RA, TP53AIP1, and RRP1B. 

\paragraph{\TGFB:}

For the \TGFB analysis we focused on 7 of the 10172 genes available in the data set that appear in the ‘\TGFB 1 pathway’ according
to the KEGGREST package in \R{} \citep{tenenbaum2016keggrest}. These were the VIT, PDE4B, ATP8B1, MAGEA11, PDE6C, PDE9A, and SEPTIN4 genes.

\clearpage

\subsection{Binary Classification}

We provide additional details of the binary classification experiments in Section 6.2.

\subsubsection{$t$-logistic regression}{\label{Sec:tLogistic}}

Following \cite{ding2010t,ding2013t}, the $t$-exponential function used in the $t$-logistic regression is defined as 
\begin{align}
    \exp_t(x) := \begin{cases}
    \exp(x) &\textrm{ if } t = 1\\
    \max\left\{1 + (1 - t)x, 0\right\}^{1/(1-t)}&\textrm{ otherwise}\\
    \end{cases},\nonumber
\end{align}
and $G_t(X\theta)$ is defined as the solution of 
\begin{align}
    \exp_t(0.5X\theta-G_t(X\theta)) + \exp_t(-0.5X\theta-G_t(X\theta)) = 1.\label{Equ:tExponential_Normaliser}
\end{align}
In general, there is no closed form for $G_t(X\theta)$ but Algorithm 1 of \cite{ding2013t} computes it efficiently.


\subsubsection{Transformations of probit and $t$-logistic $\beta$'s}{\label{Sec:LogisticTransform}}

To minimise the \textit{a priori} \TVD between the probit and $t$-logistic alternative models and the logistic canonical model the $\beta$'s of the alternative model are scalar multiplied. For the probit model, the canonical parameters are multiplied by $0.5876364$ and for the $t$-logistic the canonical parameters are multiplied by $1.331078$.

\subsubsection{Variables selected Colon Cancer data}

To prepare the Colon Cancer data for our analysis we first took the natural logarithm of the gene expression levels to remove some of their skewness. We then used the \textit{glmnet} package in \R{} \cite{friedman2010regularization} to conduct LASSO variable selection using cross-validation to choose the hyperparameter. This process left us with the intercept and genes
\begin{verbatim}
   genes.249, genes.377, genes.493, genes.625, genes.1325, genes.1473, 
   genes.1582, genes.1671, genes.1772. 
\end{verbatim}

\subsection{Mixture Models}{\label{App:FiniteGaussianMixDetails}}

Here we provide full details of the models and priors considered in Section 6.3. We estimated Gaussian mixture models of the form 
\begin{align}
    f(y; \omega, \mu, \sigma, K) = \sum_{k=1}^K \omega_j N(y; \mu_j, \sigma_j)\nonumber
\end{align}
using Normal-Inverse-Gamma-Dirichlet priors 
\begin{align}
    (\omega_1, \ldots, \omega_K) &\sim \textup{Dir}(\alpha_1, \ldots, \alpha_K)\nonumber\\
    \sigma_k^2 &\sim \mathcal{IG}\left(\frac{\nu_0}{2}, \frac{S_0}{2}\right), \quad k = 1,\ldots, K\nonumber\\
    \mu_k | \sigma_k &\sim \mathcal{N}(0, \sqrt{\kappa}\sigma_j), \quad k = 1,\ldots, K\nonumber
\end{align}
with $\alpha_k = a$ $k = 1,\ldots, K$, $\nu_0 = 5$, $S_0 = 0.2$ and $\kappa = 5.68$ following the recommendations of \cite{fuquene2019choosing}. We elicited the parameter $a$ to ensure that the marginal prior probability that any of the component weights was greater than 0.05 was 0.95.

\end{document}